\documentclass[a4paper,UKenglish,cleveref, autoref, thm-restate]{lipics-v2021}

\pdfoutput=1 
\hideLIPIcs  

\title{Collision-free Movement on Grids and Beyond} 

\author{Hendrik Molter}{Chair for Algorithm Engineering, Hasso Plattner Institute, Potsdam, Germany \and Department of Computer Science, Ben-Gurion~University~of~the~Negev, 
	Beer-Sheva, 
	Israel}{hendrik.molter@hpi.de}{https://orcid.org/0000-0002-4590-798X}{Supported by the Israel Science Foundation, grant nr.~1470/24, by the European Union's Horizon Europe research and innovation programme under grant agreement 949707, by the European Research Council, grant nr.~101039913 (PARAPATH), and by the Deutsche Forschungsgemeinschaft (DFG, German Research Foundation) -- 565381415.}

\author{Meirav Zehavi}{Department of Computer Science, Ben-Gurion~University~of~the~Negev, Beer-Sheva, Israel}{meiravze@bgu.ac.il}{https://orcid.org/0000-0002-3636-5322}{Supported by the Israel Science Foundation, grant nr.~1470/24, and by the European Research Council, grant nr.~101039913 (PARAPATH).}

\authorrunning{Hendrik Molter and Meirav Zehavi} 

\Copyright{Hendrik Molter and Meirav Zehavi} 



\category{} 

\relatedversion{} 

\acknowledgements{Main work was done while Hendrik Molter was affiliated with the Ben-Gurion University of the Negev.}

\nolinenumbers 

\EventEditors{John Q. Open and Joan R. Access}
\EventNoEds{2}
\EventLongTitle{42nd Conference on Very Important Topics (CVIT 2016)}
\EventShortTitle{CVIT 2016}
\EventAcronym{CVIT}
\EventYear{2016}
\EventDate{December 24--27, 2016}
\EventLocation{Little Whinging, United Kingdom}
\EventLogo{}
\SeriesVolume{42}
\ArticleNo{23}

\usepackage[dvipsnames,svgnames,x11names]{xcolor}
\usepackage[utf8]{inputenc}
\usepackage[T1]{fontenc}
\usepackage{tabularx}
\usepackage{amsthm}
\usepackage{amsmath}
\usepackage{amssymb}
\usepackage{amsfonts}
\usepackage{mathtools}
\usepackage{enumitem}
\usepackage[sort,numbers]{natbib}
\usepackage{todonotes}
\usepackage{comment}

\usepackage{tikz}

\usetikzlibrary{arrows,decorations.pathmorphing,decorations.pathreplacing,backgrounds,positioning,fit,matrix}
\usetikzlibrary{shapes,calc,patterns,arrows.meta}
\tikzset{
	vert/.style={circle,inner sep=1.5,fill=white,draw=black,minimum size=.3cm},
	vert2/.style={inner sep=1.5,fill=white,draw=black,minimum size=.3cm},
    dummy/.style={circle,fill=black,draw=black,inner sep=2.5},
	edge/.style={color=black, line width=1pt},
	diredge/.style={->,>={Stealth[width=8pt,length=8pt]},color=black, line width=1pt},
	timelabel/.style={fill=white,font=\footnotesize, text centered},
	wave/.style={decorate,decoration={coil,aspect=0}},
	dirwave/.style={->, >={Stealth[width=8pt,length=8pt]},decorate,decoration={coil,aspect=0}},
	diredge2/.style={->,>={Stealth[width=8pt,length=8pt]}}
}

\usepackage{contour}

\usepackage{multirow}
\usepackage{booktabs}
\usepackage{makecell}
\usepackage{caption}
\usepackage{xspace}

\theoremstyle{definition}
\newtheorem{construction}{Construction}
\newtheorem{reductionrule}{Reduction Rule}
\usepackage{soul}

\crefname{reductionrule}{Reduction Rule}{Reduction Rules}

\crefname{figure}{Figure}{Figures}

\newcommand{\commentout}[1]{}

\newcommand{\problemdef}[3]{
	\begin{center}\fbox{
	\begin{minipage}{0.95\textwidth}
		\vspace{1pt}
 
		\noindent
		\textbf{#1}
  
		\vspace{5pt}
  
  \begin{minipage}{0.99\textwidth}
		\setlength{\tabcolsep}{3pt}
		\begin{tabularx}{\textwidth}{@{}lX@{}}
			\textrm{Input:}     & #2 \\
			\textrm{Question:}  & #3
		\end{tabularx}
  \end{minipage}
	\end{minipage}}
	\end{center}
}

\newcommand{\MoveToPi}{\textsc{Collision-Free Movement to~$\Pi$}\xspace}
\newcommand{\MoveToConnected}{\textsc{Collision-Free Movement to Connectivity}\xspace}
\newcommand{\MoveToConnectedShort}{\textsc{Collision-Free Movement to Connectivity}\xspace}
\newcommand{\MoveToMatching}{\textsc{Collision-Free Movement to Matching}\xspace}

\newcommand{\CMoveToPi}{\textsc{Movement to~$\Pi$}\xspace}

\newcommand{\Movement}{\textsc{Collision-Free Movement}\xspace}

\DeclareMathOperator{\len}{len}
\DeclareMathOperator{\dist}{dist}
\DeclareMathOperator{\diam}{diam}

\DeclareMathOperator{\tw}{tw}

\DeclareMathOperator{\main}{main}
\DeclareMathOperator{\obnoxious}{obnoxious}
\DeclareMathOperator{\robot}{robot}
\DeclareMathOperator{\inc}{inc}
\DeclareMathOperator{\vertex}{vertex}
\DeclareMathOperator{\edge}{edge}
\DeclareMathOperator{\step}{step}
\DeclareMathOperator{\nextstep}{nextstep}

\newcommand{\OO}{\mathcal{O}}

\newcommand{\NP}{\textsf{NP}\xspace}

\newcommand{\NoKernelAssume}{coNP $\subseteq$ NP/poly\xspace}
\newcommand{\NNoKernelAssume}{coNP $\not\subseteq$ NP/poly\xspace}

\newcommand{\yes}{yes\xspace}

\begin{document}

\maketitle

\begin{abstract}
We study \emph{collision-free movement problems} on graphs, where the task is to coordinate a set of robots so that they reach a target formation satisfying a desired property while minimizing the total travel distance. This framework extends two classical models: (a) \emph{minimizing movement} [Demaine et al., TALG '09, '14], which does not enforce collision avoidance, and (b) \emph{coordinated motion planning} or \emph{multi-agent path finding} [Eiben et al., SoCG '23, Deligkas et al., ICALP '24, among many others], where each robot is assigned an explicit target position.



We focus on the setting where the target formation of the robots should be \emph{connected}. We analyze the parameterized complexity of the problem with respect to the number of (main) robots and the total travel length on grid graphs and two natural generalizations thereof: planar graphs and unit disk graphs. 
\end{abstract}








% ...

\clearpage

\section{Introduction}

Coordinated motion of robots is a fundamental task in computational geometry and robotics. The environment is frequently modeled by a grid, often consisting also of ``obnoxious robots'' in addition to the ``main robots''. Within the environment, the main robots are set to move so as to perform a predetermined task; the obnoxious robots, on the other hand, are simply meant to model movable obstacles. To date, most efforts were directed at the following: 
\begin{enumerate}
\item Given a target position for each main robot, plan a {\em collision-free} motion for all robots so that each main robot will reach its target~(see, e.g., \cite{DeligkasEGK024,eiben23,Stern19,SternSFK0WLA0KB19,FioravantesKKMO24,papadimitriou1994motion,salzman2020research}). This setting, widely known as {\em MAPF (Multi-Agent Path Finding)},  was part of the 3rd Computational Geometry Challenge of SoCG 2021~\cite{fekete2022computing}. 
\item Given a formation for all main robots (where a formation of particular interest is to induce a connected set), plan a  motion {\em (collisions allowed)} of all robots as to achieve this formation~(see, e.g., \cite{DemaineHM14,DemaineHMSGZ09,berman20111,friggstad2011minimizing}).
\end{enumerate}
Both objectives have been extensively studied in a variety of discrete and continuous settings, from both a theoretical perspective---particularly in the frameworks of parameterized complexity~\cite{DeligkasEGK024,eiben23,DemaineHM14,FioravantesKKMO24,eiben2025minor} and approximation algorithms~\cite{DemaineHMSGZ09,berman20111,papadimitriou1994motion}---and a practical perspective~\cite{fekete2022computing,sharon2015conflict}.
The above leaves open a very obvious and natural question:

\begin{center}\fbox{
	\begin{minipage}{0.95\textwidth}
{\bf What about the task of, given a characterization of a formation for all main robots, planning a {\em collision-free} motion of all robots as to achieve this formation?}
\end{minipage}
}
\end{center}

To the best of our knowledge, the present work is the first to address this problem. As in the classical settings discussed above, two natural optimization goals arise: {\em (i)} minimize the total number of moves taken by all robots, i.e.,the {\em total energy}; {\em (ii)} minimize the number of time steps needed, i.e., the {\em makespan}. Here, we focus on the first goal, while the desired formation is to establish a \emph{connected set}. We provide a spectrum of parameterized results for grids and generalizations thereof, with respect to the two most natural parameters in this context: the number of main robots ($k$), and the total energy to spend ($c$). Notably, our results delineate the subtle boundaries between various tractable and intractable cases.

Previous works most closely related to ours---i.e., which focus on parameterized complexity, environments modeled by graphs (particularly grids), and total energy minimization, are Eiben et al.~\cite{eiben23} and Deligkas et al.~\cite{DeligkasEGK024} for the first task, and Demaine et al.~\cite{DemaineHM14} for the second task. Discussion (non-exhaustive) of some other related works can be found at the end of this section.

Eiben et al.~\cite{eiben23} studied MAPF on solid grids with both (total) energy  and makespan minimization as objectives, and without obnoxious robots. This generalizes the well-known $(n^2-1)$-Puzzle, which was extensively studied before and is known to be NP-hard~\cite{ratner1990n2,demaine2018simple}. They showed (for solid grids only) that both energy and makespan minimization problems are fixed-parameter tractable (FPT) w.r.t.~the number of (main) robots $k$, the energy minimization problem is FPT w.r.t.~the energy $c$, and the makespan minimization problem is paraNP-hard w.r.t.~the makespan. The follow-up work by Deligkas et al.~\cite{DeligkasEGK024} focused on (total) energy minimization, in the presence of obnoxious robots (thereby closely related to the Rush Hour problem, which was also extensively studied before~\cite{FlakeB02,FernauHNRR03}).  Further, they considered general graphs, and showed:  FPT w.r.t.~the number of obnoxious robots $\ell$ when $k=1$; FPT w.r.t.~the total number of robots $k+\ell$ plus the treewidth of the input graph;  FPT approximation w.r.t.~$k+\ell$ and additive error; W[1]-hardness w.r.t.~$c$; FPT w.r.t.~$c$ on planar graphs. In further follow-up, Deligkas et al.~\cite{DeligkasEGKLR26} show that the problem is W[1]-hard w.r.t.~$c$ even if $k=1$ and FPT w.r.t.~$k$ plus the treedepth of the input graph. It remains open whether the problem is FPT w.r.t.~$k+\ell$.

Demaine et al.~\cite{DemaineHM14} studied the parameterized complexity of a very general problem addressing movement to specific formations and for several objectives, including both main and obnoxious robots. Here, recall, collisions are allowed (also in the final formation itself). They provided a dichotomy that asserted that, when the family of possible formations, $\cal F$, is decidable and closed under edge-addition, the problem is FPT w.r.t.~$k$ if the edge-minimal formations in $\cal F$ have bounded treewidth, and W[1]-hard otherwise. In particular, this asserts that the problem where we aim to move robots so that the main robots will form a connected set, and with the objective of energy minimization, is FPT w.r.t.~$k$. We remark that this  case (of connectivity) has received previous focused attention, particularly from the viewpoint of approximation~(see, e.g.,~\cite{berman20111}).
Additionally, Demaine et al.~\cite{DemaineHM14} provided W[1]-hardness w.r.t.~$k$ for the setting where the formation is a hereditary property that does not contain all cliques or all independent sets.

\subsection{Our Contributions}
We refer to the problem introduced here as \MoveToConnectedShort. If results hold for more general formations we call the problem \MoveToPi. All formal definitions are given in Section~\ref{sec:prelims}. Recall that 
 $k$ is the number of main robots, $\ell$ is the number of obnoxious robots, and $c$ is the (total) energy bound.
We start by discussing some results that are easy to obtain or follow straightforwardly from the literature~\cite{DemaineHMSGZ09,DemaineHM14} in \cref{sec:prelims}. Notably, we have the following. 
\begin{itemize}
    \item If $\ell=0$, then \MoveToConnectedShort is fixed-parameter tractable when parameterized by $k$.
    \item If $\Pi$ is hereditary and does not contain all complete graphs and all empty graphs, then \MoveToPi is W[1]-hard when parameterized by $k+c$ even if~$\ell=0$.
    \item \MoveToConnectedShort is NP-hard even if $\ell=0$ and the input graph is planar.
    \item If deciding whether a graph has property $\Pi$ is polynomial-time solvable, then \MoveToPi is in XP when parameterized by $c$.
\end{itemize}

In the following, we give a brief description of each of our results and point to the sections where we present the full proofs.
In \cref{sec:overview} we give a more detailed overview on our results and explain the techniques we use to obtain them (deferring the proofs to the later sections).

Note that since connectivity is not hereditary, the W[1]-hardness result for parameter $k+c$ mentioned above does not apply to \MoveToConnectedShort. In \cref{sec:hardness}, show the following as our first hardness result.
\begin{restatable}{theorem}{whard}
\label{thm:w1hard}
    \MoveToConnectedShort is W[1]-hard when parameterized by $k+c$.
\end{restatable}

It follows that if we wish to obtain fixed-parameter tractability results for $k$, $c$, or $k+c$, we have to restrict the class of input graphs. We focus our investigation on grid graphs, and two natural generalizations: planar graphs and unit disk graphs.

Our main technical hardness result for grids is the following (\cref{sec:hardness}).
\begin{restatable}{theorem}{nphardness}
\label{thm:nphardness}
    \MoveToConnectedShort is NP-hard even if $k=2$ and the input graph is a grid.
\end{restatable}
That is, it is hard for just two robots to meet on a grid! We remark that in both hardness results, the instances produced by the reductions have an unbounded number of obnoxious robots, which is not completely unexpected since \MoveToConnectedShort is fixed-parameter tractable when parameterized by $k$ if there are no obnoxious robots.

On the positive side, we provide the following kernelization results in \cref{sec:grids}. 
\begin{restatable}{theorem}{kernel}
\label{thm:kernel}
    \MoveToPi parameterized by $k+c$ admits a polynomial kernel of size $O(k\cdot c^2)$ if the input graph is a grid.
\end{restatable}
For the case where $\Pi$ is connectivity, we can give a better bound on the kernel size.
\begin{restatable}{theorem}{kernell}
\label{thm:kernel2}
    \MoveToConnectedShort parameterized by $k+c$ admits a polynomial kernel of size $O(\min\{k\cdot c^2,k^2+c^2\})$ if the input graph is a grid.
\end{restatable}

Furthermore, we give a number of kernelization lower bounds in \cref{sec:grids}. 
The constructions used are all based on ideas from the reduction behind \cref{thm:nphardness}. We first show that we presumably cannot obtain polynomial kernels for the parameter $k+c$ on planar graphs. Hence, exploiting the properties of grids is essential for the kernelization algorithms.

\begin{restatable}{theorem}{nopkplanar}
\label{thm:nopkplanar}
    \MoveToConnectedShort parameterized by $c$ does not admit a polynomial kernel unless \NoKernelAssume, even if the input graph is planar and $k=2$.
\end{restatable}
Further, we show that the dependency of the polynomial kernel size on $k$ for grids presumably is unavoidable.
\begin{restatable}{theorem}{nopkgrid}
\label{thm:nopkgrid}
    \MoveToConnectedShort parameterized by $c$ does not admit a polynomial kernel unless \NoKernelAssume, even if the input graph is a grid.
\end{restatable}
Finally, we show that we presumably cannot significantly improve the kernel size given in \cref{thm:kernel2}. 
\begin{restatable}{theorem}{nopkgridd}
\label{thm:nopkgrid2}
    \MoveToConnectedShort parameterized by $k+c$ does not admit a polynomial kernel of size $O((k+c)^{2-\varepsilon})$ for any $\varepsilon>0$ unless \NoKernelAssume, even if the input graph is a grid.
\end{restatable}

In \cref{sec:planar}, we turn our attention to planar graphs and show the following result.

\begin{restatable}{theorem}{mso}
\label{thm:mso}
    \MoveToPi is fixed-parameter tractable when parameterized by $k+c$ if the input graph is planar.
\end{restatable}
We obtain this result by providing a monadic second-order logic (MSO) formulation of the problem and then employing Courcelle's famous theorem~\cite{arnborg1991easy,courcelle1990monadic,courcelle2012graph}. 

In \cref{sec:udgs}, we show that \MoveToConnectedShort on unit disk graphs admits an FPT approximation algorithm when parameterized by $k+c$ for the canonical optimization variant of the problem where we aim to minimize the energy. 
\begin{restatable}{theorem}{approx}
\label{thm:approx}
    \MoveToConnectedShort admits a fixed-parameter approximation when parameterized by $k+c$ if the input graph is a unit disk graph, that produces a solution with energy at most $2\cdot \text{OPT}+3k$, where OPT is the energy of an optimal solution.
\end{restatable}

Finally, we give some open questions and future research directions in \cref{sec:conclusion}.

\subsection{Further Related Work}

In recent years there has been intensive research on multi-agent path finding (MAPF) problems in multiple variations~\cite{Stern19,SternSFK0WLA0KB19,salzman2020research}. 
As mentioned, MAPF can be seen as a generalization of the $(n^2-1)$-Puzzle and other so-called \emph{pebble motion problems} on graphs~\cite{goldreich2011finding,yu2013structure,ratner1990n2,demaine2018simple}. It is also closely related to the problem such as Rush Hour~\cite{FlakeB02,FernauHNRR03}, Temporally Disjoint Paths~\cite{KunzMZ23,KlobasMMNZ23}, the Snake Game~\cite{gupta2020parameterized}, and Chip Reconfiguration~\cite{CalinescuDP08}. To the best of our knowledge, the current state-of-the-art (optimal) algorithms for MAPF problems employ the so-called conflict-based search approach~\cite{sharon2015conflict}.

The following works are more closely related to our problem. Papadimitriou et al.~\cite{papadimitriou1994motion} study coordinated motion planning there is only one main robot and show that the problem is NP-hard on planar graphs. They give some further approximation results. Fioravantes et al.~\cite{FioravantesKKMO24} study coordinated motion planning where the goal is to minimize the makespan rather than the total energy. They obtain an extensive number of parameterized complexity results. Eiben et al.~\cite{eiben2025minor} investigate the parameterized complexity a problem variant, where robots may move an arbitrary distance in each step and the number of step shall be minimized. There is also work on the parameterized complexity of coordinated motion planning problems where the robots have geometric shapes~\cite{KanjP24,AgarwalBHSS25}.

Various version of the movement problem (where collisions do not matter) have been studied by Demaine et al.~\cite{DemaineHMSGZ09} from an approximation perspective. Improved approximation results have been given by Berman et al.~\cite{berman20111}. Friggstad and Salavatipour~\cite{friggstad2011minimizing} study a closely related problem: Mobile Facility Location. Here, facilities and clients should be moved in a graph such that every client is close to a facility. They study approximability of this problem.

\section{Overview of Results and Techniques}\label{sec:overview}
In this section, we present a more detailed overview on our results, together with intuitive and informal explanations and descriptions of the techniques that we use. Full details and formal proofs can be found in the later sections.

\subparagraph{Hardness Resuls (\cref{sec:hardness}).} We first show that \MoveToConnectedShort is W[1]-hard when parameterized by $k+c$ (\cref{thm:w1hard}). To this end, we present a parameterized reduction from \textsc{Multicolored Clique}~\cite{fellows2009multipleinterval}. Here, given an $\ell$-partite graph, we are asked whether it contains a clique of size $\ell$. The parameter is~$\ell$. The reduction uses an edge-selection gadget for every color-combination. Each one has two main robots and one obnoxious robot for each edge. One of the obnoxious robots may be moved (the corresponding edge is ``selected'') which enables the two main robots to move to verification gadgets for the colors in the color-combination. Each verification gadget of a color contains ``immovable'' formations of main robots which can only be connected if the selected edges of the color-combinations containing the color have the same endpoints. For an illustration see \cref{fig:w1hardness0}.
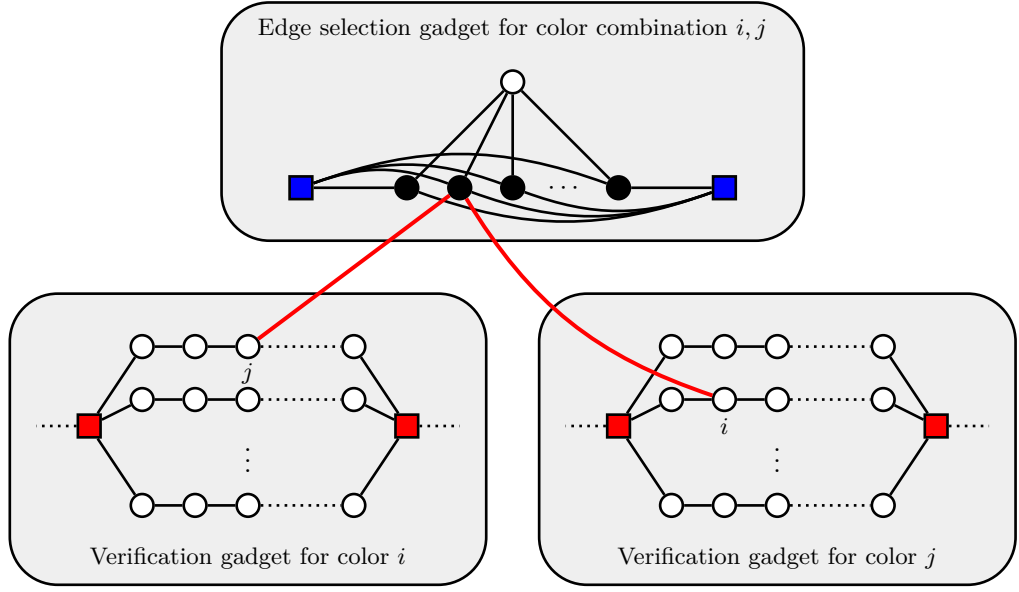
\begin{figure}[t]
\begin{center}
\begin{tikzpicture}[line width=1pt,scale=.7]

\draw[rounded corners=18pt,fill=lightgray!25!white] (-5.5, 1.5) rectangle (5.5, -3);
\node (E) at (0,1) {\small Edge selection gadget for color combination $i,j$};

\draw[rounded corners=18pt,fill=lightgray!25!white] (-9.5, -4) rectangle (-.5, -9.5);
\node (G1) at (-5,-9) {\small Verification gadget for color $i$};

\draw[rounded corners=18pt,fill=lightgray!25!white] (.5, -4) rectangle (9.5, -9.5);
\node (G1) at (5,-9) {\small Verification gadget for color $j$};

\node[vert] (X2) at (0,0) {};

\node[vert,fill=black] (A2) at (-2,-2) {};
\node[vert,fill=black] (B2) at (-1,-2) {};
\node[vert,fill=black] (C2) at (0,-2) {};

\node (L1) at (1,-2) {$\cdots$};

\node[vert,fill=black] (D2) at (2,-2) {};

\node[vert2,fill=blue] (S1) at (-4,-2) {};
\node[vert2,fill=blue] (S2) at (4,-2) {};

\node[vert2,fill=red] (N1) at (-8,-6.5) {};
\node[vert] (U1) at (-7,-5) {};
\node[vert] (U2) at (-6,-5) {};
\node[vert] (U3) at (-5,-5) {};
\node[vert] (U4) at (-3,-5) {};

\node[vert] (V1) at (-7,-6) {};
\node[vert] (V2) at (-6,-6) {};
\node[vert] (V3) at (-5,-6) {};
\node[vert] (V4) at (-3,-6) {};

\node (J) at (-5,-5.5) {\small $j$};

\node (L2) at (-5,-7) {$\vdots$};

\node[vert] (W1) at (-7,-8) {};
\node[vert] (W2) at (-6,-8) {};
\node[vert] (W3) at (-5,-8) {};
\node[vert] (W4) at (-3,-8) {};
\node[vert2,fill=red] (N2) at (-2,-6.5) {};

\node[vert2,fill=red] (N3) at (2,-6.5) {};
\node[vert] (U5) at (3,-5) {};
\node[vert] (U6) at (4,-5) {};
\node[vert] (U7) at (5,-5) {};
\node[vert] (U8) at (7,-5) {};

\node[vert] (V5) at (3,-6) {};
\node[vert] (V6) at (4,-6) {};
\node[vert] (V7) at (5,-6) {};
\node[vert] (V8) at (7,-6) {};

\node (I) at (4,-6.5) {\small $i$};

\node (L2) at (5,-7) {$\vdots$};

\node[vert] (W5) at (3,-8) {};
\node[vert] (W6) at (4,-8) {};
\node[vert] (W7) at (5,-8) {};
\node[vert] (W8) at (7,-8) {};
\node[vert2,fill=red] (N4) at (8,-6.5) {};

\draw (X2) -- (A2);
\draw (X2) -- (B2);
\draw (X2) -- (C2);
\draw (X2) -- (D2);

\draw (S1) -- (A2);
\draw (S1) edge[bend left=20] (B2);
\draw (S1) edge[bend left=20] (C2);
\draw (S1) edge[bend left=20] (D2);

\draw (S2) -- (D2);
\draw (S2) edge[bend left=20] (B2);
\draw (S2) edge[bend left=20] (C2);
\draw (S2) edge[bend left=20] (A2);

\draw[dotted] (N1) -- (-9,-6.5);

\draw (N1) -- (U1);
\draw (U1) -- (U2);
\draw (U2) -- (U3);
\draw[dotted] (U3) -- (U4);
\draw (U4) -- (N2);

\draw (N1) -- (V1);
\draw (V1) -- (V2);
\draw (V2) -- (V3);
\draw[dotted] (V3) -- (V4);
\draw (V4) -- (N2);

\draw (N1) -- (W1);
\draw (W1) -- (W2);
\draw (W2) -- (W3);
\draw[dotted] (W3) -- (W4);
\draw (W4) -- (N2);

\draw[dotted] (N2) -- (-1,-6.5);

\draw[dotted] (N3) -- (1,-6.5);

\draw (N3) -- (U5);
\draw (U5) -- (U6);
\draw (U6) -- (U7);
\draw[dotted] (U7) -- (U8);
\draw (U8) -- (N4);

\draw (N3) -- (V5);
\draw (V5) -- (V6);
\draw (V6) -- (V7);
\draw[dotted] (V7) -- (V8);
\draw (V8) -- (N4);

\draw (N3) -- (W5);
\draw (W5) -- (W6);
\draw (W6) -- (W7);
\draw[dotted] (W7) -- (W8);
\draw (W8) -- (N4);

\draw[dotted] (N4) -- (9,-6.5);

\draw[color=red, line width=1.5pt] (B2) -- (U3);
\draw[color=red, line width=1.5pt] (B2) edge[bend left=-20] (V6);
\end{tikzpicture}
    \end{center}
    \caption{Illustration of the graph  W[1]-hardness reduction. The black filled round vertices have obnoxious robots on them, and the blue filled square vertices have main robots on them.  
    The red filled square vertices have main robots on them and paths of sufficient length attached to them that have main robots on each vertex (which are not depicted).
    The two red edges correspond to connections between the edge selecting gadget and the verification gadgets.}\label{fig:w1hardness0}
\end{figure}

Our main hardness result is that \MoveToConnectedShort is NP-hard on a grid even if there are only two main robots (\cref{thm:nphardness}). We obtain the result via a reduction from \textsc{Linked Planar 3-SAT}~\cite{Pilz19}. This reduction is described in two construction steps: We first construct a planar instance of \MoveToConnectedShort and then show how the produced graph can be embedded into a grid (with appropriate modifications). 
In \textsc{Linked Planar 3-SAT}, we are given a Boolean formula $\phi$ in 3-CNF with clause set $Y$ and variable set $X$, and a graph $H=(Y\cup X, F)$ such that the following holds:
\begin{itemize}
    \item Graph $H$ is the union of a Hamiltonian cycle that first visits all elements of $X$ and then all elements of $Y$, and the incidence graph of $\phi$.
    \item Graph $H$ is planar and there is an embedding where each edge between a clause and a variable that appears negated in the clause is inside the cycle, and each edge between a clause and a variable that appears non-negated in the clause is outside of the cycle.
    \item Each variable appears in at most three clauses.
\end{itemize}
We may assume that we know the embedding of $H$ and the Hamiltonian cycle~\cite{Pilz19}.

Informally speaking, the main idea on how to reduce from \textsc{Linked Planar 3-SAT} to \MoveToConnectedShort planar graphs is as follows. We split the Hamiltonian cycle in $H$ between the last clause vertex and the first variable vertex, in that way obtaining a Hamiltonian path that first traverses all variable vertices and then all clause vertices. If we draw this path horizontally (from left to right), then (in a planar drawing) each edge between a clause and a variable that appears negated in the clause is below the path, and each edge between a clause and a variable that appears non-negated in the clause is above of the path. We add an extra vertex to the beginning of the path and an extra vertex to the end, and place a main robot on each of them. Now we ``replace'' each variable vertex with the variable gadget illustrated in \cref{fig:vargadget0}. We connect edges above the path to the three vertices on top of the gadget, and edge below the path to the three vertices below. For an illustration see \cref{fig:nphardness20}. Then these edges (that originally connected the clause and variable vertices) are subdivided a number of times and obnoxious robots are places on all newly created vertices. Furthermore, obnoxious robots are places on all clause vertices.
The idea is that the three obnoxious robots in each variable gadget (see \cref{fig:vargadget0}) need to move either to the three upper or lower vertices, which will encode setting the variable to true or false. Now the main robot on the left can freely move until the last variable gadget. Then the obnoxious robot on each clause vertex needs to move out of the way. To this end, all obnoxious robots on one of the paths connecting the clause vertex to a variable gadget need to move. This is only possible if the vertex of the variable gadget where the path arrives is free, hence encoding that the corresponding literal satisfies the clause. Now the main robots can freely meet on some adjacent clause vertices.

\begin{figure}[t]
\begin{center}
\begin{tikzpicture}[line width=1pt,scale=.65,xscale=1.6]

\node[vert,fill=black] (A) at (1,0) {};
\node[vert,fill=black] (B) at (2,0) {};
\node[vert,fill=black] (C) at (3,0) {};
\node[vert] (XT1) at (4,1) {};
\node[vert] (XT2) at (5,1) {};
\node[vert] (XT3) at (6,1) {};
\node[vert] (XF1) at (4,-1) {};
\node[vert] (XF2) at (5,-1) {};
\node[vert] (XF3) at (6,-1) {};
\node[vert] (N1) at (7,0) {};

\draw (A) -- (B);
\draw (B) -- (C);
\draw (C) -- (XT1);
\draw (XT1) -- (XT2);
\draw (XT2) -- (XT3);
\draw (XT3) -- (N1);
\draw (C) -- (XF1);
\draw (XF1) -- (XF2);
\draw (XF2) -- (XF3);
\draw (XF3) -- (N1);

\end{tikzpicture}
    \end{center}
    \caption{Illustration of a variable gadget. The black vertices have obnoxious robots on them.}\label{fig:vargadget0}
\end{figure}
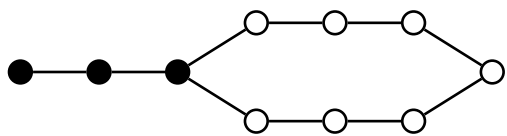

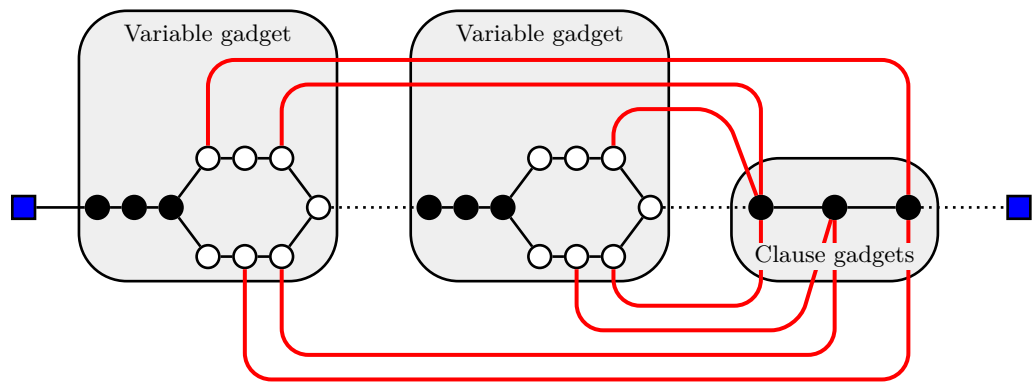
\begin{figure}[t]
\begin{center}
\begin{tikzpicture}[line width=1pt,scale=.65,xscale=.75]

\draw[rounded corners=18pt,fill=lightgray!25!white] (.5, 4) rectangle (7.5, -1.5);
\node (E) at (4,3.5) {\small Variable gadget};

\draw[rounded corners=18pt,fill=lightgray!25!white] (9.5, 4) rectangle (16.5, -1.5);
\node (E) at (13,3.5) {\small Variable gadget};

\draw[rounded corners=18pt,fill=lightgray!25!white] (18.2, 1) rectangle (23.8, -1.5);

\node[vert2,fill=blue] (R1) at (-1,0) {};
\node[vert,fill=black] (A) at (1,0) {};
\node[vert,fill=black] (B) at (2,0) {};
\node[vert,fill=black] (C) at (3,0) {};
\node[vert] (XT1) at (4,1) {};
\node[vert] (XT2) at (5,1) {};
\node[vert] (XT3) at (6,1) {};
\node[vert] (XF1) at (4,-1) {};
\node[vert] (XF2) at (5,-1) {};
\node[vert] (XF3) at (6,-1) {};
\node[vert] (N1) at (7,0) {};

\node[vert,fill=black] (A2) at (10,0) {};
\node[vert,fill=black] (B2) at (11,0) {};
\node[vert,fill=black] (C2) at (12,0) {};
\node[vert] (YT1) at (13,1) {};
\node[vert] (YT2) at (14,1) {};
\node[vert] (YT3) at (15,1) {};
\node[vert] (YF1) at (13,-1) {};
\node[vert] (YF2) at (14,-1) {};
\node[vert] (YF3) at (15,-1) {};
\node[vert] (N2) at (16,0) {};

\node[vert,fill=black] (N3) at (19,0) {};

\node[vert,fill=black] (CL1) at (21,0) {};
\node[vert,fill=black] (CL2) at (23,0) {};

\node[vert2,fill=blue] (R2) at (26,0) {};

\draw (R1) -- (A);

\draw (A) -- (B);
\draw (B) -- (C);
\draw (C) -- (XT1);
\draw (XT1) -- (XT2);
\draw (XT2) -- (XT3);
\draw (XT3) -- (N1);
\draw (C) -- (XF1);
\draw (XF1) -- (XF2);
\draw (XF2) -- (XF3);
\draw (XF3) -- (N1);

\draw[dotted] (N1) -- (A2);

\draw (A2) -- (B2);
\draw (B2) -- (C2);
\draw (C2) -- (YT1);
\draw (YT1) -- (YT2);
\draw (YT2) -- (YT3);
\draw (YT3) -- (N2);
\draw (C2) -- (YF1);
\draw (YF1) -- (YF2);
\draw (YF2) -- (YF3);
\draw (YF3) -- (N2);

\draw[dotted] (N2) -- (N3);

\draw (N3) -- (CL1);
\draw (CL1) -- (CL2);
\draw[dotted] (CL2) -- (R2);

\draw[color=red, line width=1.5pt, rounded corners=10pt] (N3) -- (19,-2) -- (15,-2) -- (YF3);
\draw[color=red, line width=1.5pt, rounded corners=10pt] (CL1) -- (20,-2.5) -- (14,-2.5) -- (YF2);
\draw[color=red, line width=1.5pt, rounded corners=10pt] (CL1) -- (21,-3) -- (6,-3) -- (XF3);
\draw[color=red, line width=1.5pt, rounded corners=10pt] (CL2) -- (23,-3.5) -- (5,-3.5) -- (XF2);

\draw[color=red, line width=1.5pt, rounded corners=10pt] (N3) -- (18,2) -- (15,2) -- (YT3);

\draw[color=red, line width=1.5pt, rounded corners=10pt] (N3) -- (19,2.5) -- (6,2.5) -- (XT3);

\draw[color=red, line width=1.5pt, rounded corners=10pt] (CL2) -- (23,3) -- (4,3) -- (XT1);

\node[fill=lightgray!25!white,inner sep=1.1] (E) at (21,-1) {\small Clause gadgets};

\end{tikzpicture}
    \end{center}
    \caption{Illustration of the informally described graph $G$. The blue filled round vertices have obnoxious robots on them, and the black filled square vertices have main robots on them. The red edges correspond to long paths that have obnoxious robots on all inner vertices.}\label{fig:nphardness20}
\end{figure}

Afterwards, we modify the informally described graph (call it $G$) depicted in \cref{fig:nphardness20} to produce a graph $G'$ that is a grid (together with an embedding that shows this). Note that there are three main obstacles to overcome:
\begin{itemize}
    \item The clause vertices in $G$ may have degree 5, which is not possible in a grid. Hence, we have to modify the vertices to ``clause gadgets'' that are grids.
    \item All paths connecting clause gadgets to variable gadgets need to have the same length in order for the reduction to work.
    \item The path between the last variable gadget and the first clause gadget has to be sufficiently long (that is, longer than any path from a clause gadget to a variable gadget) for the reduction to work.
\end{itemize}
We introduce clause gadgets, which are, informally speaking, small tree-like grids rooted at the clause vertices and three leaves pointing upward, and three pointing downward. All vertices in the gadgets have obnoxious robots, and the paths to the variable gadgets now start at the leaves.
In order to ensure that all paths between variable and clause gadgets have the appropriate length, we introduce a ``zone'' in the grid between all variable gadgets and all clause gadgets, which is sufficiently large such that we can change the length of each path to the desired length by using ``snake-like'' embeddings in the grid. For an illustration see \cref{fig:nphardnessgrid0}.

\begin{figure}[t]
\begin{center}
\begin{tikzpicture}[line width=1pt,scale=.29,xscale=1]

\draw[blue!15!white,fill=blue!15!white] (20.5, 12.5) rectangle (25.5, -15.5);

\draw[rounded corners=6pt,lightgray!25!white,fill=lightgray!25!white] (-.5, 1.5) rectangle (8.6, -1.5);

\draw[rounded corners=6pt,lightgray!25!white,fill=lightgray!25!white] (10.5, 1.5) rectangle (19.6, -1.5);

\draw[rounded corners=6pt,lightgray!25!white,fill=lightgray!25!white] (26.6, 1.5) rectangle (31.2, -1.5);

\draw[rounded corners=6pt,lightgray!25!white,fill=lightgray!25!white] (31.6, 1.5) rectangle (36.2, -1.5);

\draw[rounded corners=6pt,lightgray!25!white,fill=lightgray!25!white] (36.6, 1.5) rectangle (41.2, -1.5);

\foreach \i in {-15,...,12}
{
    \draw[line width=.5pt,color=lightgray,dotted] (-2.8,\i) -- (43.8,\i);
}

\foreach \i in {-2,...,43}
{
    \draw[line width=.5pt,color=lightgray,dotted] (\i,-15.8) -- (\i,12.8);
}

\draw[line width=2pt,red] (27,1) -- (27,5) -- (25,5) -- (25,6) -- (24,6) -- (24,4) -- (23,4) -- (23,6) -- (22,6) -- (22,4) -- (21,4) -- (21,5) -- (19,5) -- (19,1);

\draw[line width=2pt,red] (29,1) -- (29,8) -- (25,8) -- (25,9) -- (24,9) -- (24,7) -- (23,7) -- (23,9) -- (22,9) -- (22,7) -- (21,7) -- (21,8) -- (8,8) -- (8,1);

\draw[line width=2pt,red] (37,1) -- (37,11) -- (25,11) -- (25,12) -- (24,12) -- (24,10) -- (23,10) -- (23,12) -- (22,12) -- (22,10) -- (21,10) -- (21,11) -- (4,11) -- (4,1);

\draw[line width=2pt,red] (27,-1) -- (27,-5) -- (25,-5) -- (25,-4) -- (24,-4) -- (24,-6) -- (23,-6) -- (23,-4) -- (22,-4) -- (22,-6) -- (21,-6) -- (21,-5) -- (19,-5) -- (19,-1);

\draw[line width=2pt,red] (32,-1) -- (32,-8) -- (25,-8) -- (25,-7) -- (24,-7) -- (24,-9) -- (23,-9) -- (23,-7) -- (22,-7) -- (22,-9) -- (21,-9) -- (21,-8) -- (17,-8) -- (17,-1);

\draw[line width=2pt,red] (34,-1) -- (34,-11) -- (25,-11) -- (25,-10) -- (24,-10) -- (24,-12) -- (23,-12) -- (23,-10) -- (22,-10) -- (22,-12) -- (21,-12) -- (21,-11) -- (8,-11) -- (8,-1);

\draw[line width=2pt,red] (37,-1) -- (37,-14) -- (25,-14) -- (25,-13) -- (24,-13) -- (24,-15) -- (23,-15) -- (23,-13) -- (22,-13) -- (22,-15) -- (21,-15) -- (21,-14) -- (4,-14) -- (4,-1);

\draw[dashed] (25.5,12.5) -- (25.5,-15.5) -- (20.5,-15.5) -- (20.5,12.5) -- (25.5,12.5);

\draw (-2,0) -- (0,0);
\draw[line width=4pt] (0,0) -- (4,0);
\draw (8,0) -- (8,-1) -- (4,-1) -- (4,1) -- (8,1) -- (8,0) -- (9,0);

\draw[dotted] (9,0) -- (11,0);
\draw[line width=4pt] (11,0) -- (15,0);
\draw (19,0) -- (19,-1) -- (15,-1) -- (15,1) -- (19,1) -- (19,0) -- (20,0);

\draw[dotted] (20,0) -- (21,0) -- (21,-3) -- (22,-3) -- (22,3) -- (23,3) -- (23,-3) -- (24,-3) -- (24,3) -- (25,3) -- (25,0) -- (26,0);

\draw[line width=4pt] (27,1) -- (31,1);
\draw[line width=4pt] (27,-1) -- (31,-1);
\draw[line width=4pt] (29,1) -- (29,-1);
\draw (26,0) -- (31,0);

\draw[line width=4pt] (32,1) -- (36,1);
\draw[line width=4pt] (32,-1) -- (36,-1);
\draw[line width=4pt] (34,1) -- (34,-1);
\draw (31,0) -- (36,0);

\draw[line width=4pt] (37,1) -- (41,1);
\draw[line width=4pt] (37,-1) -- (41,-1);
\draw[line width=4pt] (39,1) -- (39,-1);
\draw (36,0) -- (41,0);

\draw[dotted] (41,0) -- (43,0);

\node[vert2,fill=blue] (R1) at (-2,0) {};
\node[vert2,fill=blue] (R2) at (43,0) {};

\end{tikzpicture}
    \end{center}
    \caption{Illustration of an embedding of the graph $G'$ obtained from graph $G$. The drawing is not up to scale but rather should give the rough idea on how the embedding looks like. The grid points on the thick black lines have obnoxious robots on them, and the blue filled square vertices have main robots on them. The red edges correspond to long paths that have obnoxious robots on all inner vertices. The light blue area surrounded by the black dashed line indicates the part of the grid, where we make sure that all red paths have the same length, and the path connecting the last variable gadget with the first clause gadget is sufficiently long. This is also indicated by the snake-like shape of the paths in the area.}\label{fig:nphardnessgrid0}
\end{figure}
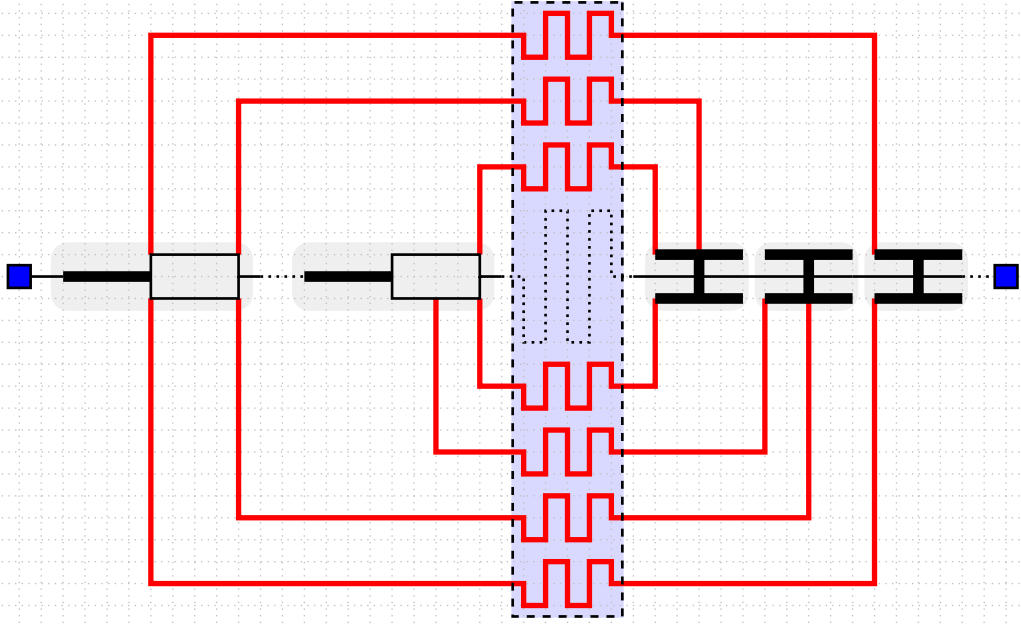

\subparagraph{Collision Free Movement on Grids (\cref{sec:grids}).} We give several kernelization algorithms and kernelization lower bounds for \MoveToConnectedShort on grids.
In fact, our first kernelization algorithm works for general formations $\Pi$.
The main insight here is that vertices that are far away from each of the main robots (regardless of whether they are occupied by obnoxious robots or not) cannot be relevant for any solution. This allows us to obtain a reduction rule that, when applied exhaustively, yields a kernel of size $O(k\cdot c^2)$ for \MoveToPi (\cref{thm:kernel}). If $\Pi$ is connectivity, we can improve the size of the kernel. In this case, informally speaking, we have that if a vertex is too far away from \emph{some} main robot, it cannot be relevant, since eventually, all main robots have to be close together. This allows us to obtain an additional reduction rule. Exhaustively applying both rules yields a kernel of size $O(\min\{k\cdot c^2,k^2+c^2\})$ for \MoveToConnectedShort (\cref{thm:kernel2}).

Furthermore, we give a number of kernelization lower bounds, using the well-known framework of OR-cross-compositions~\cite{BJK14,Fom+19}. For all results, we cross-compose from \textsc{Linked Planar 3-SAT} and reuse (modified versions of) the constructions devised to obtain \cref{thm:nphardness}. 
We first show that we presumably cannot obtain polynomial kernels for \MoveToConnectedShort parameterized by $c$ on planar graphs, even if $k=2$ (\cref{thm:nopkplanar}). To this end, we use the planar instances of \MoveToConnectedShort produced by the reduction behind \cref{thm:nphardness} before the modifications to obtain a grid are applied (\cref{fig:nphardness20}). We remove the main robots from each instance, and connect the (previous) starting locations of the main robots to a new starting location with sufficiently long paths (ensuring that the main robots cannot move between instances). For an illustration see \cref{fig:nopk10}. This result shows that the restriction to grids is essential for obtaining polynomial kernels.

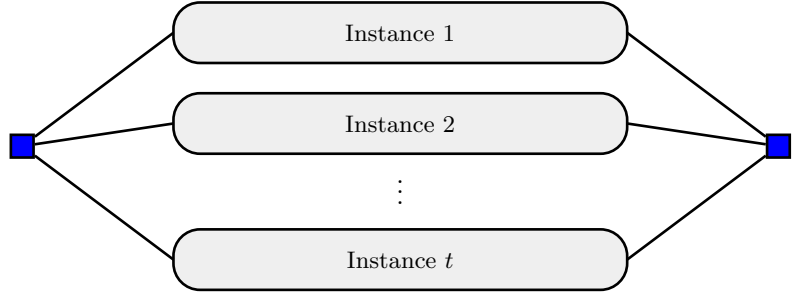
\begin{figure}[t]
\begin{center}
\begin{tikzpicture}[line width=1pt,scale=1,yscale=.8]

\draw[rounded corners=10pt,fill=lightgray!25!white] (0, 0) rectangle (6, 1);
\node (E) at (3,.5) {\small Instance 1};

\draw[rounded corners=10pt,fill=lightgray!25!white] (0, -.5) rectangle (6, -1.5);
\node (E) at (3,-1) {\small Instance 2};

\node (E) at (3, -2) {$\vdots$};

\draw[rounded corners=10pt,fill=lightgray!25!white] (0, -2.75) rectangle (6, -3.75);
\node (E) at (3,-3.25) {\small Instance $t$};

\node[vert2,fill=blue] (R1) at (-2,-1.375) {};
\node[vert2,fill=blue] (R2) at (8,-1.375) {};

\draw (R1) -- (0,.5);
\draw (R1) -- (0,-1);
\draw (R1) -- (0,-3.25);

\draw (R2) -- (6,.5);
\draw (R2) -- (6,-1);
\draw (R2) -- (6,-3.25);

\end{tikzpicture}
    \end{center}
    \caption{Illustration of the OR-cross-composition. Blue filled square vertices have main robots on them.}\label{fig:nopk10}
\end{figure}

We further show that we cannot obtain polynomial kernels for \MoveToConnectedShort on grids for the single parameter $c$ (\cref{thm:nopkgrid2}). Note that \cref{thm:nphardness} obviously rules out polynomial kernels for the single parameter $k$. Here, we use the grid instances produced by the reduction behind \cref{thm:nphardness}. Note that we cannot connect the instances in the same fashion as in the previous cross-composition, as this introduces a high-degree vertex which is not embeddable in a grid. We replace this with a tree-like grid that has a main robots each vertex. Intuitively, we cannot use a single main robot since the distance between its starting location and the instances would be too large. Now, instead of single main robots moving through one of the instances, a ``snake'' of main robots moves through it to establish a connection between the right and the left side of the instance. To accommodate this, we have to also make some modifications to instances produced by the reduction behind \cref{thm:nphardness}. 

Finally, we show that we presumably cannot significantly improve the size of the kernel produced by \cref{thm:kernel2} (\cref{thm:nopkgrid2}). Here, we use the weak cross-compositions framework~\cite{hermelin2012weak,Fom+19}. Informally speaking, we have to ensure here that the parameter of the instance produced by the weak cross-composition can be upper-bounded by the maximum size among the input instances multiplied by the square root of the number of input instances. We can achieve this by again using the instances produced by the reduction behind \cref{thm:nphardness} and arranging the input instances in a grid-like fashion, see \cref{fig:nopk30}. Furthermore, we remove the main robots from each input instance and use one main robot starting in the lower left of the arrangement and one in the upper right. Together with appropriate connections between the instances, this ensures that the sum of the distances of both main robots to any of the instances is in $O(\sqrt{t}\cdot x_{\max})$, where $t$ is the number of input instances and $x_{\max}$ is the maximum instance size.

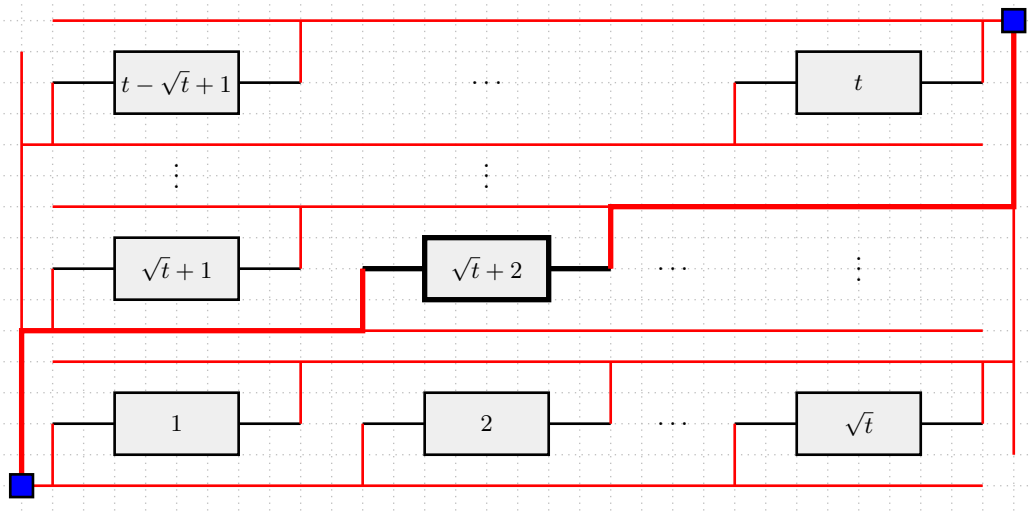
\begin{figure}[t]
\begin{center}
\begin{tikzpicture}[line width=1pt,scale=.41]

\foreach \i in {-1,...,14}
{
    \draw[line width=.5pt,color=lightgray,dotted] (-3.8,\i) -- (29.8,\i);
}

\foreach \i in {-3,...,29}
{
    \draw[line width=.5pt,color=lightgray,dotted] (\i,-1.8) -- (\i,14.8);
}

\draw (0,1) -- (-2,1);
\draw[fill=lightgray!25!white] (0, 0) rectangle (4, 2);
\node (E) at (2,1) {\small $1$};
\draw (4,1) -- (6,1);

\draw (8,1) -- (10,1);
\draw[fill=lightgray!25!white] (10, 0) rectangle (14, 2);
\node (E) at (12,1) {\small $2$};
\draw (14,1) -- (16,1);

\node (E) at (18, 1) {$\ldots$};

\draw (20,1) -- (22,1);
\draw[fill=lightgray!25!white] (22, 0) rectangle (26, 2);
\node (E) at (24,1) {\small ${\sqrt{t}}$};
\draw (26,1) -- (28,1);

\draw (0,6) -- (-2,6);
\draw[fill=lightgray!25!white] (0, 5) rectangle (4, 7);
\node (E) at (2,6) {\small ${\sqrt{t}+1}$};
\draw (4,6) -- (6,6);

\draw (8,6) -- (10,6);
\draw[fill=lightgray!25!white] (10, 5) rectangle (14, 7);
\node (E) at (12,6) {\small ${\sqrt{t}+2}$};
\draw (14,6) -- (16,6);

\node (E) at (18, 6) {$\ldots$};

\node (E) at (24, 6.25) {$\vdots$};

\node (E) at (2, 9.25) {$\vdots$};

\draw (0,12) -- (-2,12);
\draw[fill=lightgray!25!white] (0, 11) rectangle (4, 13);
\node (E) at (2,12) {\small ${t-\sqrt{t}+1}$};
\draw (4,12) -- (6,12);

\node (E) at (12, 12) {$\ldots$};

\node (E) at (12, 9.25) {$\vdots$};

\draw (20,12) -- (22,12);
\draw[fill=lightgray!25!white] (22, 11) rectangle (26, 13);
\node (E) at (24,12) {\small ${t}$};
\draw (26,12) -- (28,12);

\draw[red] (-3,-1) -- (-3,13);

\draw[red] (-3,-1) -- (28,-1);
\draw[red] (-2,-1) -- (-2,1);
\draw[red] (8,-1) -- (8,1);
\draw[red] (20,-1) -- (20,1);

\draw[red] (-3,4) -- (28,4);
\draw[red] (-2,4) -- (-2,6);
\draw[red] (8,4) -- (8,6);

\draw[red] (-3,10) -- (28,10);
\draw[red] (-2,10) -- (-2,12);
\draw[red] (20,10) -- (20,12);

\draw[red] (29,0) -- (29,14);

\draw[red] (-2,3) -- (29,3);
\draw[red] (6,3) -- (6,1);
\draw[red] (16,3) -- (16,1);
\draw[red] (28,3) -- (28,1);

\draw[red] (-2,8) -- (29,8);
\draw[red] (6,8) -- (6,6);
\draw[red] (16,8) -- (16,6);

\draw[red] (-2,14) -- (29,14);
\draw[red] (6,14) -- (6,12);
\draw[red] (28,14) -- (28,12);

\draw[line width=2pt]  (8,6) -- (10,6);
\draw[line width=2pt] (10, 5) rectangle (14, 7);
\draw[line width=2pt]  (14,6) -- (16,6);
\draw[red,line width=2pt] (-3,-1) -- (-3,4) -- (8,4) -- (8,6);
\draw[red,line width=2pt] (16,6) -- (16,8) -- (29,8) -- (29,14);

\node[vert2,fill=blue] (R1) at (-3,-1) {};
\node[vert2,fill=blue] (R2) at (29,14) {};

\end{tikzpicture}
    \end{center}
    \caption{Illustration of the weak cross-composition. The drawing is not up to scale but rather should give the rough idea on how the embedding looks like. Black boxes correspond to the instances and red lines to the connecting paths. Blue filled square vertices have main robots on them. The thick lines correspond to the paths visited by the main robots when instance nr.\ ${\sqrt{t}+2}$ is selected.}\label{fig:nopk30}
\end{figure}

\subparagraph{Collision Free Movement on Planar Graphs (\cref{sec:planar}).} We show that \MoveToPi is fixed-parameter tractable when parameterized by $k+c$ if the input graph is planar (\cref{thm:mso}). Note that this result holds for general formation $\Pi$.
We obtain it by providing a monadic second-order logic (MSO) formulation of the problem and then employing Courcelle's famous theorem~\cite{arnborg1991easy,courcelle1990monadic,courcelle2012graph}, stating that \textsc{MSO Model Checking} is fixed-parameter tractable w.r.t\ the size of the formula and the treewidth of the relational structure on which the formula is defined.

As the relational structure, we take the input graph together with ``time steps'' that model at which point in a movement a vertex is visited by a robot or an edge is traversed by a robot. Since the number of moves is upper-bounded by $c$, we need $c$ of such time steps. Informally speaking, we associate each time step with each of the vertices and edges, which increases the treewidth of the graph by some factor depending on $c$. Additionally, we can remove vertices that are far away from any main robot (similar to the reduction rule in the kernelization algorithm). This allows us to upper-bound the diameter of the input graph by a function in $k+c$ and, since it is planar, also its treewidth. It follows that the treewidth of the relational structure can be upper-bounded by a function in $k+c$.

Now we define movements as sets of vertex-time step pairs with the following properties:
\begin{itemize}
\item For each time step, there is one vertex-time step pair in the set.
\item For each pair of consecutive time steps, the corresponding vertices are either equal (the robot does not move) or connected by an edge.
\item The vertex corresponding to the first time step is the starting location of a robot.
\end{itemize}
These properties can be checked with MSO formulas. Furthermore, we can check (with an MSO formula) for each pair of movements, whether the two corresponding robots do not collide. If during a movement, a robot visits a vertex that is the starting location of another robot, there must be another movement moving that robot out of the way. Finally, we can have at most $c$ non-trivial movements (that is, the robot actually moves at least once). Putting this together, we can obtain an MSO formula that checks for a set of $c$ movements, whether they are collision-free and also do not collide with any robot that is not moved. Finally, we can straightforwardly determine the locations of all main robots after the movements and check whether the formation has property $\Pi$.
We remark that the size of the overall MSO formula only depends on $c$ and the size of a subformula that checks whether a graph on $k$ vertices has property $\Pi$ (which, notably, is constant if $\Pi$ is connectivity). 

\subparagraph{Collision Free Movement on Unit Disk Graphs (\cref{sec:udgs}).} We show that \MoveToConnectedShort on unit disk graphs admits an FPT approximation algorithm when parameterized by $k+c$ for the canonical optimization variant of the problem, where we aim to minimize the energy. 
More specifically, we give a fixed-parameter approximation for \MoveToConnectedShort on unit disk graphs when parameterized by $k+c$, that produces a solution with energy at most $2\cdot \text{OPT}+3k$, where OPT is the energy of an optimal solution (\cref{thm:approx}).
The algorithm makes use of the fact that unit disk graphs are a subclass of so-called clique grids. Informally speaking, a clique grid, each vertex is associated with a grid cell, and all vertices in the same cell form a clique. All additional edges are between vertices in grid cells of distance at most two. 

We start by ``guessing'' the optimal value for $c$, and removing all vertices that are too far away from any main robot starting location (as in the kernelization algorithms).
Afterwards, the number of cliques / grid cells in the clique grid is upper-bounded by some function in $k+c$.
Furthermore, we guess a spanning tree that witnesses the connectivity of the formation of the main robots after the movements.

We introduce a \emph{type} for each vertex of a clique that, intuitively, describes via which cliques the vertex can be reached from each main robot and whether it is occupied by an obnoxious robot. We show how to upper-bound the number of types in a function of $c$ and $k$. We guess the type of each vertex in the spanning tree.
Now, informally speaking, we guess up to $c$ movement vectors of types that move the main robots to vertices of the correct type (as specified by the guesses for the spanning tree), and check whether they are collision-free. When two robots move to vertices of the same type, we need to ensure that these are different vertices to be collision-free.
This, intuitively, is where the factor two in the approximation comes from. We potentially need to move robots to a ``free'' vertex of a certain type, which requires at most one extra move as all vertices of the same type form a clique. Finally, at the end we may need to move all main robots to the ``correct'' vertex of the guessed type in order to obtain the guessed spanning tree, which, as we show, requires at most three extra moves per main robot.

\section{Preliminaries and Known Results}\label{sec:prelims}

In this section, we provide all necessary notation and terminology, as well as some preliminary results that are easy to obtain or straightforwardly follow from known results.

\subparagraph{Parameterized Complexity.} We use the following standard concepts from parameterized complexity theory~\cite{Cyg+15,DF13,FG06,Fom+19}.
A \emph{parameterized problem}~$L\subseteq \Sigma^*\times \mathbb{N}$ is a subset of all instances~$(x,k)$ from~$\Sigma^*\times \mathbb N$, where~$k$ denotes the \emph{parameter}.
A parameterized problem~$L$ is 
in the class FPT (or \emph{fixed-parameter tractable}) if there is an algorithm that decides every instance~$(x,k)$ for~$L$ in~$f(k)\cdot |x|^{O(1)}$ time for some computable function~$f$ that depends only on the parameter. 
A decidable parameterized problem $L$ admits a \emph{kernel} if there is a polynomial-time algorithm that transforms each instance $(x,k)\in\{0,1\}\times \mathbb{N}$ into an instance $(x',k')\in\{0,1\}\times \mathbb{N}$ such that $|x'|+ k'\le f(k)$,
for some computable function~$f$, and $(x',k')\in L$ if and only if $(x,k)\in L$. 
It is known that a decidable parameterized problem $L$ admits a kernel if and only if it is fixed-parameter tractable~\cite{Fom+19}.
If $f$ is a polynomial, that is, $|x'|+ k'\in r^{O(1)}$, we say that $L$ admits a \emph{polynomial kernel}. 
Presumably, not all problems that are fixed-parameter tractable admit a polynomial kernel~\cite{Fom+19}. We give more details on how to exclude polynomial kernels under standard complexity assumptions in the respective section of the paper.
A parameterized problem~$L$ is in the class XP, if there is an algorithm that decides every instance~$(x,k)$ for~$L$ in~$|x|^{f(k)}$ time for some computable function~$f$ that depends only on the parameter.
If a parameterized problem $L$ is W[1]-hard, then it is presumably not contained in FPT~\cite{Cyg+15,DF13,FG06}.

\subparagraph{Graph Theory.}
We use standard notation and terminology from graph theory~\cite{Die16}. 
Let $G=(V,E)$ with $E\subseteq \binom{V}{2}$ be a graph. We denote by~$V(G)=V$ and~$E(G)=E$ the sets of its vertices and edges, respectively.
We use $n$ to denote the number of vertices of $G$.
We call two vertices $u,v\in V$ \emph{adjacent} if~$\{u,v\}\in E$.
The \emph{degree} of a vertex $v\in V$ is $\deg(v)=|\{u \mid \{u,v\}\in E\}|$.
We say that a graph $G'$ is a \emph{subgraph} of $G$ if $V(G')\subseteq V(G)$ and $E(G')\subseteq E(G)$.
For some vertex set $V'\subseteq V$, we denote by~$G[V']$ the \emph{induced} subgraph of $G$ on the vertex set $V'$, that is, $G[V']=(V',E')$ where $E' = \{\{v,w\}\mid \{v,w\}\in E\wedge v\in V'\wedge w\in V'\}$. We further denote $G-V'=G[V\setminus V']$.
We say that a sequence $P=(\{v_{i-1},v_i\})_{i=1}^\ell$ of edges in $E$ forms a \emph{path} of length $\ell$ in $G$ from $v_0$ to $v_\ell$ if $v_{i}\neq v_j$ for all $0\le i<j\le \ell$.
We say that vertices $u$ and $v$ are \emph{connected} in $G$ if there is a path from $u$ to $v$ in $G$. We say that graph $G$ is connected if for all $u,v\in V$ we have that $u$ and $v$ are connected. We say that $G[V']$ for some $V'\subseteq V$ is a \emph{connected component} of $G$ if $G[V']$ is connected and for all $v\in V\setminus V'$ we have that $G[V'\cup\{v\}]$ is not connected.
We say that a path $P$ of length $\ell$ in $G$ from vertex $u$ to $v$ is a \emph{shortest path} if there is no path from $u$ to $v$ in $G$ of length $\ell-1$. The \emph{distance} between $u$ and $v$ in $G$ is the length of a shortest path from $u$ to $v$ in $G$ and denoted $\dist(u,v)$. If $u$ and $v$ are in different connected components of $G$, then we set $\dist(u,v)=\infty$. The \emph{diameter} of a graph $G$ is the maximum distance between any two vertices that are in the same connected component, that is, $\diam(G)=\max_{u,v\in V(G)\mid \dist(u,v)\neq\infty}\dist(u,v)$.

A graph $G$ is \emph{planar} if it can be embedded in the plane (drawn with points for vertices and curves for edges) without crossing edges. 
A graph $G=(V,E)$ is a \emph{unit disk graph} if it has an intersection model consisting of disks of unit size, that is, there is a $f:V\rightarrow \mathbb{R}$ such that for all $\{u,v\}\in E$ it holds that $|f(u)-f(v)|\le 1$.
The set $D=\{f(v)\mid v\in V\}$ is called the \emph{point set} of $G$.
A graph $G=(V,E)$ is a \emph{grid} if there exists a function $f:V\rightarrow [t]\times [t']$ for some $t,t'\in\mathbb{N}$ such that for all $\{u,v\}\in E$ it holds that if $f(u)=(i,j)$ and $f(v)=(i',j')$ then $|i-i'|= 1$ and $j=j'$, or $|j-j'|= 1$ and $i=i'$. Note that grids are planar graphs and unit disk graphs. 

\subparagraph{Problem Definition.}
Assume we are given a graph $G=(V,E)$ and two sets $R,O\subseteq V$ with $R\cap O=\emptyset$. 
Here, $R$ denotes the starting positions of the \emph{main robots} and $O$ denotes the starting positions of the \emph{obnoxious robots}. We denote $k=|R|$ and $\ell=|O|$.
Let $m\in \mathbb{N}$ $r\in R\cup O$. 
We call a vector $M_r=(v_1, v_2, v_3, \ldots, v_m)$ a \emph{movement} for $r$ with $m$ steps if $v_1=r$ and for all $i\in[m-1]$ it holds that $v_i=v_{i+1}$ or $\{v_i,v_{i+1}\}\in E$. 
Intuitively, in every step of a movement, the robot may stay on its current position or move to an adjacent vertex.
The \emph{length} $\len(M_r)$ of a movement $M_r=(v_i)_{i\in[m]}$ is $\len(M_r)=|\{i\in[m-1]\mid v_i\neq v_{i+1}\}|$, that is, the number of steps in which $r$ moves to an adjacent vertex.
We say that a movement $M_r$ is \emph{trivial} if $\len(M_r)=0$.
We say that two movements $M_r=(v_i)_{i\in[m]}$ and $M_{r'}=(v'_i)_{i\in[m]}$ for some $r,r'\in R\cup O$ with $m$ steps each are \emph{collision-free} if for all $i\in[m]$ it holds that $v_i \neq v'_i$ and for all $i\in[m-1]$ it holds that $\{v_i,v_{i+1}\}\neq\{v'_i,v'_{i+1}\}$.
Intuitively, robots collide if they occupy the same vertex in the same step or if they use the same edge (in opposite directions) to move to an adjacent vertex in the same step.
A \emph{movement plan} for $R\cup O$ is a set $\mathcal{M}=\{M_r\mid r\in R\cup O\}$ containing a movement for each robot in $R\cup O$ with the same number of steps.
The length of a movement plan $\len(\mathcal{M})$ is the sum of length of all movements in $\mathcal{M}$, that is, $\len(\mathcal{M})=\sum_{M_r\in \mathcal{M}}\len(M_r)$.
We say that a movement plan $\mathcal{M}$ is \emph{collision-free} if for all $M_r,M_{r'}\in \mathcal{M}$ with $r\neq r'$ we have that the movements $M_r$ and $M_{r'}$ are collision-free.
The \emph{end-configuration} of $R$ after a movement plan $\mathcal{M}=\{M_r\mid r\in R\cup O\}$ is $R^\star(\mathcal{M})=\{v_m\mid \exists M_r\in\mathcal{M} \text{ s.t.\ } r\in R\text{ and } M_r=(v_i)_{i\in[m]}\}$, that is, the set of vertices that are occupied by robots in $R$ after all movements.

Let $\Pi$ be a recursively enumerable graph property. We define the problem \MoveToPi as follows.
\problemdef{\MoveToPi}
{A graph $G=(V,E)$, two sets $R,O\subseteq V$ with $R\cap O=\emptyset$, and an integer $c\in\mathbb{N}$.}
{Is there a collision-free movement plan $\mathcal{M}$ for $R\cup O$ such that $\len(\mathcal{M})\le c$ and $G[R^\star(\mathcal{M})]\in \Pi$?}

In this work, we will focus on the case of $\Pi$ being the property of a graph to be \emph{connected}. We call this problem \MoveToConnected. 

\subparagraph{First Observations and Results.}
Since \MoveToPi contains the problem of deciding whether a graph has property $\Pi$ as a special case, we immediately get the following.
\begin{observation}
    If deciding whether a graph has property $\Pi$ is NP-hard, then \MoveToPi is NP-hard even if $\ell=0$ and $c=0$.
\end{observation}

If we do not require the movement plan to be collision-free, then we call the problem \CMoveToPi. A generalization of \CMoveToPi has been studied by Demaine et al.~\cite{DemaineHMSGZ09,DemaineHM14}. We can observe that \MoveToPi and \CMoveToPi coincide if~$\ell=0$.
\begin{lemma}\label{lem:equiv}
    Let $I=(G=(V,E),R,O,c)$ be an instance of \MoveToPi. If $O=\emptyset$, then $I$ is a yes-instance of \MoveToPi if and only if $I$ is a yes-instance of \CMoveToPi.
\end{lemma}
\begin{proof}
If $I$ is a yes-instance of \MoveToPi, then clearly it is also a yes-instance of \CMoveToPi. For the other direction, assume that $I$ is a yes-instance of $\CMoveToPi$ and let $\mathcal{M}$ be a movement plan that is a solution for $I$. Assume that $\mathcal{M}$ is not collision-free.
Now let $M=(v_1,v_2,\ldots, v_m)$ and $M'=(v'_1,v'_2,\ldots, v'_{m'})$ be two movements for robots $r,r'\in R$, respectively, in $\mathcal{M}$ that collide. Let $t$ be the step in which a collision happens, that is, $v_t=v'_t$, or $v_t=v'_{t-1}$ and $v'_t=v_{t-1}$. 
Consider the latter case first, that is, $v_t=v'_{t-1}$ and $v'_t=v_{t-1}$. Here, instead of the robots swapping their position, they can ``swap their identities'' instead. Formally, we can define a new movement $M^\star$ for robot $r$ which equals $M$ for the first $t-1$ steps, and equals $M'$ for all subsequent steps, that is, $M^\star=(v_1,v_2,\ldots,v_{t-1},v'_t,v'_{t+1},\ldots,v'_{m'})$. Analogously, we define a new movement for $r'$. Clearly, the end-configuration of $\mathcal{M}$ stays the same. By iterating this modification, we can eliminate all collisions of this type.
Now consider the former case, that is, $v_t=v'_t$. First consider the case that both $r$ and $r'$ moved to $v_t=v'_t$ in step $t$, that is $v_{t-1}\neq v_t$ and $v'_{t-1}\neq v'_t$. Then we create a new movement for $r$ where the movement from $v_{t-1}$ to $v_t$ happens one time step later. By iterating this modification, we can create a movement plan where no two robots simultaneously move to the same vertex.
Now assume w.l.o.g.\ that robot $r$ moves to~$v_t$ in step $t$, that is, $v_{t-1}\neq v_t$ and robot $r'$ was already at $v_t=v'_t$ at the previous step, that is, $v'_t=v'_{t-1}$. Furthermore, assume that $t$ is the latest step where such a collision happens.
Then at some later time step $t'>t$ we must have that one of the robots moves away. Assume w.l.o.g.\ that $r'$ moves away first (if only $r$ moves away, they can ``swap identities'' according to the modification above). This implies that $v'_{t-1}=v'_t=\ldots=v_{t'-1}$ in $M'$. 
Then we can create a new movement $M^\star$ for $r'$ that is equal to $M'$ until step $t-1$ and then ``skips'' steps $t$ to $t'-1$ and immediately moves to $v'_{t'}$ at step $t$ and then continues with the steps of $M'$ from $t'$ onward, that is, $M^\star=(v'_1,v'_2,\ldots,v'_{t-1},v'_{t'},v'_{t'+1},\ldots,v'_{m'})$. Note that this may introduce new collisions, but it (strictly) shortens one of the movement vectors. It follows that iterating this modification will eventually eliminate all collisions.
\end{proof}
\cref{lem:equiv} implies that algorithmic results for \CMoveToPi carry over to \MoveToPi (and vice versa) if~$\ell=0$. In particular, we get the following.

\begin{theorem}[\cite{DemaineHM14}]
    If $\ell=0$, then \MoveToConnectedShort is fixed-parameter tractable when parameterized by $k$.
\end{theorem}

Due to the similarity of \CMoveToPi and \MoveToPi, the reductions behind some known hardness results for \CMoveToPi can be straightforwardly modified to work for \MoveToPi. In particular, Demaine et al.~\cite{DemaineHM14} show that if $\Pi$ is hereditary and does not contain all complete graphs and all empty graphs, then \CMoveToPi is W[1]-hard when parameterized by $k+c$ even if~$\ell=0$ via a reduction from the well-known $\Pi$-subgraph detection problem~\cite{khot2002parameterized}. The reduction formalizes the intuition that already finding a suitable target location for the main robots is hard, rather than finding out in which way the robots move to that location. The reduction can be straightforwardly modified to yield the following result.
\begin{theorem}[\cite{DemaineHM14,khot2002parameterized}]
    If $\Pi$ is hereditary and does not contain all complete graphs and all empty graphs, then \MoveToPi is W[1]-hard when parameterized by $k+c$ even if~$\ell=0$.
\end{theorem}

Demaine et al.~\cite{DemaineHMSGZ09} show that \textsc{Movement to Connectivity} is NP-hard even if $\ell=0$ via a reduction from \textsc{Hamiltonian Path}~\cite{Kar72}. This reduction can be straightforwardly modified as well. If additionally we reduce from \textsc{Hamiltonian Path} on planar graphs~\cite{garey1974some}, we obtain the following result.
\begin{theorem}[\cite{DemaineHMSGZ09}]
    \MoveToConnectedShort is NP-hard even if $\ell=0$ and the input graph is planar.
\end{theorem}

If the main robots have prescribed destinations, that is, for each $r\in R$ there is a $d_r\in V$ given as part of the input, and we ask for a collision free movement plan $\mathcal{M}$ (of length at most~$c$) that moves each robot in $R$ to its destination, that is, for each $r\in R$ we have $M_r\in\mathcal{M}$ with $M_r=(v_i)_{i\in[m]}$ and $v_m=d_r$, then the problem is called \Movement. A generalization of \Movement has been studied by Deligkas et al.~\cite{DeligkasEGK024}. 

Observe that, given an instance $(G,R,O,c)$ of \MoveToPi, if we can ``guess'' a destination vertex $d_r$ for each main robot $r\in R$ such that the subgraph of $G$ induced on the destination vertices has property $\Pi$, then we can reduce the problem of finding a collision-free movement plan of length at most $c$ to an instance of \Movement. In particular, this gives us the following.

\begin{observation}
If deciding whether a graph has property $\Pi$ is in NP, then \MoveToPi is in NP.
\end{observation}

Last but not least, we can observe that \MoveToPi can be solved by a simple guess-and-check algorithm if deciding whether a graph has property $\Pi$ is polynomial-time solvable. We can iterate over all $n^{O(c)}$ possibilities to pick $c$ moves and assign each one to one of the robots. Then we check whether this yields a collision-free movement plan and whether the end-configuration yields a subgraph with property $\Pi$. Hence, we get the following.
\begin{theorem}\label{thm:xp}
    If deciding whether a graph has property $\Pi$ is polynomial-time solvable, then \MoveToPi is in XP when parameterized by $c$.
\end{theorem}

\section{Hardness Results}\label{sec:hardness}

In this section, we give our main hardness results for \MoveToConnectedShort. We begin by proving \cref{thm:w1hard} (restated below), that is, showing that the problem is W[1]-hard when parameterized by $k+c$ on general graphs. Afterwards, we investigate the cases where the input graph is a grids.

\whard*

To prove \cref{thm:w1hard}, we present a parameterized reduction from \textsc{Multicolored Clique}~\cite{fellows2009multipleinterval}. Here, given an $\ell$-partite graph $H=(W_1\uplus W_2 \uplus\ldots\uplus W_\ell, F)$, we are asked whether $H$ contains a clique of size $\ell$. If $v\in W_i$, then we say that $v$ has \emph{color} $i$.  
Let $F_{i,j}$ denote the set of all edges between vertices from $W_i$ and $W_j$. This problem is known to be W[1]-hard when parameterized by $\ell$~\cite{fellows2009multipleinterval}.

\begin{construction}\label{constr:w1hardness}
Given an instance $H=(W_1\uplus W_2 \uplus\ldots\uplus W_\ell, F)$ of \textsc{Multicolored Clique}, we construct an instance $(G=(V,E), R, O, c)$ of \MoveToConnectedShort as follows.

For each color combination $i,j$ with $i<j$ we create an edge selection gadget $G_{i,j}$.
\begin{itemize}
    \item For every edge $e\in F_{i,j}$ we add a vertex $v_e$, place an obnoxious robot on it.
    \item We add three additional vertices $v_1,v_2,v_3$. We connect $v_1$ to each $v_e$ with $e\in F_{i,j}$ and place a main robot on $v_1$. We connect $v_2$ to each $v_e$ with $e\in F_{i,j}$ and place a main robot on $v_2$. We connect $v_3$ to each $v_e$.
\end{itemize}
For an illustration see \cref{fig:w1hardness1}.
For each color $i$ we create a validation gadget $G_i$ as follows.
\begin{itemize}
    \item For every vertex $w\in W_i$ we create $\ell-1$ vertices $w_1,w_2,\ldots,w_{i-1},w_{i+1},\ldots,w_\ell$ and connect them to a path (from $w_1$ to $w_\ell$).
    \item We add two vertices $u^{(i)}_1,u^{(i)}_2$ and place main robots on them. We connect $u^{(i)}_1$ to each $w_1$ with $w\in W_i$. We connect $u^{(i)}_2$ to each $w_\ell$ with $w\in W_i$. We create a path on $c$ vertices, place main robots on each vertex, and connect one of the endpoints to $u^{(i)}_1$, where we specify $c$ later. We create a second path on $c$ vertices, place main robots on each vertex, and connect one of the endpoints to $u^{(i)}_2$.
\end{itemize}

The graph $G$ is now the disjoint union of the edge selection gadgets of all color combinations $i,j$ with $i<j$ and the verification gadgets of all colors $i$. We connect the gadgets as follows.
Let $j=i+1\le \ell$, then we connect vertex $u^{(i)}_2$ from the verification gadget $G_i$ and vertex $u^{(j)}_1$ from the verification gadget $G_j$ with an edge. Finally, let $e=\{w,w'\}\in F_{i,j}$ such that $w\in W_i$ and $w'\in W_j$. We connect vertex $v_e$ from the edge selection gadget $G_{i,j}$ with vertex $w_j$ from the verification gadget $G_i$. Then we subdivide $\ell^2$ times the edge $\{v_e,w_j\}$. We connect vertex $v_e$ from the edge selection gadget $G_{i,j}$ with vertex $w'_i$ from the verification gadget $G_j$. Then we subdivide $\ell^2$ times the edge $\{v_e,w'_i\}$.
Finally, we set $c=(2\ell^2+3)\binom{\ell}{2}$. 
\end{construction}

\begin{figure}[t]
\begin{center}
\begin{tikzpicture}[line width=1pt,scale=.7]

\draw[rounded corners=18pt,fill=lightgray!25!white] (-5.5, 1.5) rectangle (5.5, -3);
\node (E) at (0,1) {\small Edge selection gadget for color combination $i,j$};

\draw[rounded corners=18pt,fill=lightgray!25!white] (-9.5, -4) rectangle (-.5, -9.5);
\node (G1) at (-5,-9) {\small Verification gadget for color $i$};

\draw[rounded corners=18pt,fill=lightgray!25!white] (.5, -4) rectangle (9.5, -9.5);
\node (G1) at (5,-9) {\small Verification gadget for color $j$};

\node[vert] (X2) at (0,0) {};

\node[vert,fill=black] (A2) at (-2,-2) {};
\node[vert,fill=black] (B2) at (-1,-2) {};
\node[vert,fill=black] (C2) at (0,-2) {};

\node (L1) at (1,-2) {$\cdots$};

\node[vert,fill=black] (D2) at (2,-2) {};

\node[vert2,fill=blue] (S1) at (-4,-2) {};
\node[vert2,fill=blue] (S2) at (4,-2) {};

\node[vert2,fill=red] (N1) at (-8,-6.5) {};
\node[vert] (U1) at (-7,-5) {};
\node[vert] (U2) at (-6,-5) {};
\node[vert] (U3) at (-5,-5) {};
\node[vert] (U4) at (-3,-5) {};

\node[vert] (V1) at (-7,-6) {};
\node[vert] (V2) at (-6,-6) {};
\node[vert] (V3) at (-5,-6) {};
\node[vert] (V4) at (-3,-6) {};

\node (J) at (-5,-5.5) {\small $j$};

\node (L2) at (-5,-7) {$\vdots$};

\node[vert] (W1) at (-7,-8) {};
\node[vert] (W2) at (-6,-8) {};
\node[vert] (W3) at (-5,-8) {};
\node[vert] (W4) at (-3,-8) {};
\node[vert2,fill=red] (N2) at (-2,-6.5) {};

\node[vert2,fill=red] (N3) at (2,-6.5) {};
\node[vert] (U5) at (3,-5) {};
\node[vert] (U6) at (4,-5) {};
\node[vert] (U7) at (5,-5) {};
\node[vert] (U8) at (7,-5) {};

\node[vert] (V5) at (3,-6) {};
\node[vert] (V6) at (4,-6) {};
\node[vert] (V7) at (5,-6) {};
\node[vert] (V8) at (7,-6) {};

\node (I) at (4,-6.5) {\small $i$};

\node (L2) at (5,-7) {$\vdots$};

\node[vert] (W5) at (3,-8) {};
\node[vert] (W6) at (4,-8) {};
\node[vert] (W7) at (5,-8) {};
\node[vert] (W8) at (7,-8) {};
\node[vert2,fill=red] (N4) at (8,-6.5) {};

\draw (X2) -- (A2);
\draw (X2) -- (B2);
\draw (X2) -- (C2);
\draw (X2) -- (D2);

\draw (S1) -- (A2);
\draw (S1) edge[bend left=20] (B2);
\draw (S1) edge[bend left=20] (C2);
\draw (S1) edge[bend left=20] (D2);

\draw (S2) -- (D2);
\draw (S2) edge[bend left=20] (B2);
\draw (S2) edge[bend left=20] (C2);
\draw (S2) edge[bend left=20] (A2);

\draw[dotted] (N1) -- (-9,-6.5);

\draw (N1) -- (U1);
\draw (U1) -- (U2);
\draw (U2) -- (U3);
\draw[dotted] (U3) -- (U4);
\draw (U4) -- (N2);

\draw (N1) -- (V1);
\draw (V1) -- (V2);
\draw (V2) -- (V3);
\draw[dotted] (V3) -- (V4);
\draw (V4) -- (N2);

\draw (N1) -- (W1);
\draw (W1) -- (W2);
\draw (W2) -- (W3);
\draw[dotted] (W3) -- (W4);
\draw (W4) -- (N2);

\draw[dotted] (N2) -- (-1,-6.5);

\draw[dotted] (N3) -- (1,-6.5);

\draw (N3) -- (U5);
\draw (U5) -- (U6);
\draw (U6) -- (U7);
\draw[dotted] (U7) -- (U8);
\draw (U8) -- (N4);

\draw (N3) -- (V5);
\draw (V5) -- (V6);
\draw (V6) -- (V7);
\draw[dotted] (V7) -- (V8);
\draw (V8) -- (N4);

\draw (N3) -- (W5);
\draw (W5) -- (W6);
\draw (W6) -- (W7);
\draw[dotted] (W7) -- (W8);
\draw (W8) -- (N4);

\draw[dotted] (N4) -- (9,-6.5);

\draw[color=red, line width=1.5pt] (B2) -- (U3);
\draw[color=red, line width=1.5pt] (B2) edge[bend left=-20] (V6);
\end{tikzpicture}
    \end{center}
    \caption{Illustration of the graph $G$ produced by  \cref{constr:w1hardness}. The black filled round vertices have obnoxious robots on them, and the blue filled square vertices have main robots on them.  
    The red filled square vertices have main robots on them and paths of length $c$ attached to them that have main robots on each vertex (which are not depicted).
    The two red edges correspond to paths of length $\ell^2$ connecting the edge selecting gadget to the verification gadgets.}\label{fig:w1hardness1}
\end{figure}
    
This finished the construction, which can clearly be performed in polynomial time. Note that the number of main robots is in $O(\ell^4)$ and $c\in O(\ell^4)$. Now we prove the correctness of the reduction, that is, the constructed instance $(G,R,O,c)$ is a yes-instance of \MoveToConnectedShort if and only if $H$ is a yes-instance of \textsc{Multicolored Clique}. This together with the above observations yields \cref{thm:w1hard}.

\begin{proof}[Proof of \cref{thm:w1hard}]
    Let be $H=(W_1\uplus W_2 \uplus\ldots\uplus W_\ell, F)$ an instance of \textsc{Multicolored Clique} and let $(G,R,O,c)$ be an instance of \MoveToConnectedShort created from $H$ by \cref{constr:w1hardness}.
    
    $(\Rightarrow):$ Assume that $H$ is a yes-instance of \textsc{Multicolored Clique} and we are given a clique $X\subseteq V(H)$. We construct a collision-free movement plan $\mathcal{M}$ of length $c$ for $(G,R,O,c)$ as follows.

    For each color combination $i,j$ with $i<j$ let $e=\{w,w'\}\in F_{i,j}$ such that $w\in W_i$, $w'\in W_j$, and $w,w'\in X$. We move the obnoxious robot on vertex $v_e$ of the edge selection gadget $G_{i,j}$ to vertex $v_3$ of the edge selection gadget $G_{i,j}$. By \cref{constr:w1hardness}, we have that $v_e$ is connected to vertex $w_j$ from the verification gadget $G_i$ via a path of length $\ell^2+1$ and is connected to vertex $w'_i$ from the verification gadget $G_j$ via a path of length $\ell^2+1$. We move the main robot on $v_1$ to $w_j$ via $v_e$ and the connecting path, and we move the main robot on $v_2$ to $w'_i$ via $v_e$ and the connecting path. Note that this is collision-free and takes $2\ell^2+2$ moves. This finishes the description of the movement plan. Note that it uses $2\ell^2+3$ moves for every edge selection gadget and hence $c=(2\ell^2+3)\binom{\ell}{2}$ moves in total. 
    
    It remains to show that the end-configuration forms a connected subgraph. To this end, note that for some fixed $i$, all edges of a color combination that contains $i$ that are contained in the clique $X$ have the same endpoint in $W_i$. Let that endpoint be $w\in W_i$ Hence, we have that in the verification gadget for color $i$, one main robot is moved to each of $w_1,w_2,\ldots,w_{i-1},w_{i+1},\ldots,w_\ell$ and hence vertices $u^{(i)}_1,u^{(i)}_2$ (which have main robots on them) are connected via a path that has a main robot on each vertex. Since this is the case for every verification gadget and there are no main robots left in the edge selection gadgets, we have that the main robots form a connected subgraph. 

    $(\Leftarrow):$ 
    Assume that $(G,R,O,c)$ is a yes-instance of \MoveToConnectedShort and we are given a collision-free movement plan $\mathcal{M}$ of length $c$ that is a solution for $(G,R,O,c)$. We construct a clique of size $\ell$ in $H$ as follows.

    Assume that $\mathcal{M}$ is a movement plan with a minimum number of moves. First, observe the following. For each verification gadget for a color $i$ it holds that the main robots on $u^{(i)}_1,u^{(i)}_2$ do not move. Recall that by \cref{constr:w1hardness}, each of the vertices is connected to a path of length $c$ that has main robots on every vertex. Moving the main robot on $u^{(i)}_1$ or $u^{(i)}_2$ breaks the connection to the main robots on the respective path. To reestablish connection, either the robot needs to move back to $u^{(i)}_1$ or $u^{(i)}_2$, respectively, which does not happen in a movement plan with a minimum number of moves, or all robots on the path need to be moved by at least one step, which costs more than $c$ moves in total and hence cannot happen in a solution. It follows that only the $2\binom{\ell}{2}$ main robots in edge selection gadgets move.
    Furthermore, notice that the paths between verification gadgets and edge selection gadgets have length $\ell^2>2\binom{\ell}{2}$. It follows that a connection cannot be established via those paths since there are not enough main robots that can move.
    We can conclude that for each color $i$ the vertices $w_1,w_2,\ldots,w_{i-1},w_{i+1},\ldots,w_\ell$ for some $w\in W_i$ need to be occupied by main robots in the end-configuration. The distance of a starting position of a main robot in an edge selection gadget to such a vertex is at least $\ell^2+2$. Furthermore, all neighboring vertices of starting position of a main robot in an edge selection gadget are occupied by obnoxious robots that need to be moved away first. Note that for each edge selection gadget, there are at most one move available to do this. It follows that in each edge selection gadget, we move exactly one obnoxious robot from some $v_e$ to the $v_3$ vertex of that gadget. We say that edge $e$ is ``selected''. Furthermore, we have that for each color $i$, exactly the vertices $w_1,w_2,\ldots,w_{i-1},w_{i+1},\ldots,w_\ell$ for some $w\in W_i$ are occupied by main robots in the end-configuration, which each were moved there with exactly $\ell^2+2$ moves.
    
    Let $X$ be the set of all endpoints of selected edges. We claim that $X$ forms a clique of size $\ell$ in $H$. Note that $X$ contains at least one vertex of each color, hence $|X|\ge \ell$. In the remainder, we show that also $|X|\le \ell$. Note that this implies that $X$ is a clique, since for each color $i$ we have that all edges of a color combination that contains $i$ have the same endpoint of color $i$. Assume for contradiction that $w,w'\in X$ with $w,w'\in W_i$ for some color~$i$. Then there are selected edges $e,e'$ such that $w$ is an endpoint of $e$ and $w'$ is an endpoint of $e'$. Consider the edge selection gadgets where $e$ and $e'$ were selected. Both correspond to a color combination that involves $i$, so both are connected to the verification gadget of color $i$. 
    If two main robots from one of the gadgets are moved into the verification gadget of color $i$, this involves at least $\ell^2+3$ moves. Hence, by arguments made above, this is not possible in a solution. We can conclude that from both edge selection gadgets, one main robot must be moved to the verification gadget of color $i$ with $\ell^2+2$ moves. Then, however, the main robots are moved to vertices $w_j$ and $w'_{j'}$ for some $j,j'$ in the verification gadget for color $i$. This is a contradiction to the observation above, that for color $i$ the vertices $w_1,w_2,\ldots,w_{i-1},w_{i+1},\ldots,w_\ell$ for some $w\in W_i$ need to be occupied by main robots in the end-configuration. We can conclude that $X$ is a clique in $H$.
\end{proof}

\subparagraph{Hardness on Grids.} Now we prove \cref{thm:nphardness} (restated below), that is, we show that \MoveToConnectedShort is NP-hard if the input graph is a grid, even if there are only two main robots. To this end, we first give a reduction that produces a planar graph, and then we show how to modify the construction such that we obtain a grid.

\nphardness*

To prove \cref{thm:nphardness}, we present a polynomial-time reduction from \textsc{Linked Planar 3-SAT}~\cite{Pilz19}. Here, we are given a Boolean formula $\phi$ in 3-CNF with clause set $Y$ and variable set $X$, and a graph $H=(Y\cup X, F)$ such that the following holds:
\begin{itemize}
    \item Graph $H$ is the union of a Hamiltonian cycle that first visits all elements of $X$ and then all elements of $Y$, and the incidence graph of $\phi$ (i.e., the graph on $Y\cup X$ where $c\in Y$ and $x\in X$ are connected by an edge if and only if variable $x$ appears on clause $c$).
    \item Graph $H$ is planar and w.l.o.g.\ there is an embedding where each edge between a clause and a variable that appears negated in the clause is inside the cycle, and each edge between a clause and a variable that appears non-negated in the clause is outside of the cycle.
    \item Each variable appears in at most three clauses.
\end{itemize}
This problem is known to be NP-hard even if $H$ together with its embedding are given in the input~\cite{Pilz19}. Note that the Hamiltonian cycle can easily be computed from $H$.

In the following, we first describe a construction that, given an instance of \textsc{Linked Planar 3-SAT}, produces an instance of \MoveToConnectedShort with a planar graph. Later, we describe how to modify the instance such that the graph is a grid.

\begin{construction}\label{constr:planarnph}
Given an instance $(\phi, H)$ of \textsc{Linked Planar 3-SAT}, we construct an instance $(G=(V,E), R, O, c)$ of \MoveToConnectedShort with $k=2$ as follows. 
Let the variables $X$ and clauses $Y$ in $\phi$ be ordered in the way they appear in the Hamiltonian cycle in $H$.

For each variable $x\in X$, we create a variable gadget $G_x=(V_x, E_x)$ with 
\begin{itemize}
    \item $V_x=\{v_1^x,v_2^x,v_3^x,v_4^x,u_1^x,u_2^x,u_3^x,w_1^x,w_2^x,w_3^x\}$, and
    \item $E_x=\{\{v_1^x,v_2^x\},\{v_2^x,v_3^x\},\{v_3^x,u_1^x\},\{u_1^x,u_2^x\},\{u_2^x,u_3^x\},\{u_3^x,v_4^x\},\{v_3^x,w_3^x\},\{w_3^x,w_2^x\},$ 
    
    $\{w_2^x,w_1^x\},\{w_1^x,v_4^x\}\}$. 
\end{itemize}
Furthermore, we place obnoxious robots on the vertices $v_1^x,v_2^x,v_3^x$. For an illustration of the variable gadget see \cref{fig:vargadget}.

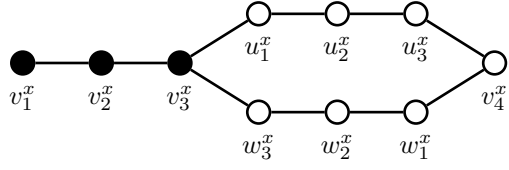
\begin{figure}[t]
\begin{center}
\begin{tikzpicture}[line width=1pt,scale=.65,xscale=1.6]

\node[vert,fill=black,label=below:$v_1^x$] (A) at (1,0) {};
\node[vert,fill=black,label=below:$v_2^x$] (B) at (2,0) {};
\node[vert,fill=black,label=below:$v_3^x$] (C) at (3,0) {};
\node[vert,label=below:$u_1^x$] (XT1) at (4,1) {};
\node[vert,label=below:$u_2^x$] (XT2) at (5,1) {};
\node[vert,label=below:$u_3^x$] (XT3) at (6,1) {};
\node[vert,label=below:$w_3^x$] (XF1) at (4,-1) {};
\node[vert,label=below:$w_2^x$] (XF2) at (5,-1) {};
\node[vert,label=below:$w_1^x$] (XF3) at (6,-1) {};
\node[vert,label=below:$v_4^x$] (N1) at (7,0) {};

\draw (A) -- (B);
\draw (B) -- (C);
\draw (C) -- (XT1);
\draw (XT1) -- (XT2);
\draw (XT2) -- (XT3);
\draw (XT3) -- (N1);
\draw (C) -- (XF1);
\draw (XF1) -- (XF2);
\draw (XF2) -- (XF3);
\draw (XF3) -- (N1);

\end{tikzpicture}
    \end{center}
    \caption{Illustration of a variable gadget $G_x$ for a variable $x$ (\cref{constr:planarnph}). The black filled vertices have obnoxious robots on them.}\label{fig:vargadget}
\end{figure}

We construct the graph $G$ of the \MoveToConnectedShort instance by first taking the graph $H$ and replacing each vertex corresponding to a variable, say $x$, with variable gadget $G_x$ in the following way. If two vertices corresponding to variables $x,y$ are adjacent in $H$, that is, they appear consecutive in the ordering (first $x$ and then $y$), then we add edge $\{v_4^x,v_1^y\}$ to $G$. 
Let $x$ be the last variable in the ordering and let $c$ be the first clause in the ordering, then we add edge $\{v_4^x,v^c\}$ to $G$ and subdivide it $s+2$ times, where $s=|X|+|Y|$.
For each clause $c\in C$ we add call the vertex corresponding to it $v^c$ and place an obnoxious robot on $v^c$. 
 Let variable~$x$ appear in clauses $c_1,c_2,c_3$. Then the vertex $x$ in $H$ is connected to vertices $v^{c_1},v^{c_2},v^{c_3}$ in~$H$. Let w.l.o.g.\ the order $c_1,c_2,c_3$ be the linearization of the cyclic ordering in which the edges incident with $x$ to the vertices $c_1,c_2,c_3$ are embedded around $x$ in the planar embedding of $H$, such that all clauses in which $x$ appears non-negated are ordered before the clauses in which $x$ appears negated. For $i\in\{1,2,3\}$, if~$x$ appears in $c_i$ non-negated, then we connect $u_i^x$ with $v^{c_i}$ with an edge and subdivide the edge~$s$ times. Furthermore, we place an obnoxious robot on each vertex created by a subdivision. If $x$ appears in~$c_i$ negated, then we connect $w_i^x$ with $v^{c_i}$ with an edge and subdivide the edge $s$ times. Furthermore, we place an obnoxious robot on each vertex created by a subdivision. 
Finally, we add two vertices $r_1,r_2$ to $G$ and place a main robot on each of them. Let $x$ be the first variable in the ordering, then we add edge $\{r_1,v_1^x\}$ to $G$. Let $c$ be the last clause in the ordering, then we add edge $\{r_2,v^c\}$ to $G$.
For an illustration, see \cref{fig:nphardness}. Finally, we set $c=16\cdot |X|+(s+2)\cdot (|Y|+1)+1$. We will assume in the proofs of correctness that $s^2/4>c$. Note that this is clearly the case if $|X|$ and $|Y|$ are sufficiently large.
\end{construction}

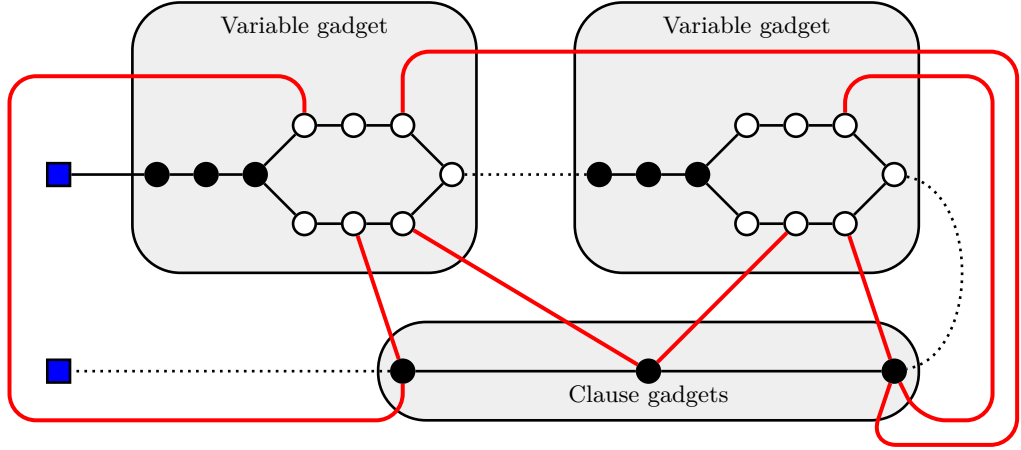
\begin{figure}[t]
\begin{center}
\begin{tikzpicture}[line width=1pt,scale=.65]

\draw[rounded corners=18pt,fill=lightgray!25!white] (.5, 3.5) rectangle (7.5, -2);
\node (E) at (4,3) {\small Variable gadget};

\draw[rounded corners=18pt,fill=lightgray!25!white] (9.5, 3.5) rectangle (16.5, -2);
\node (E) at (13,3) {\small Variable gadget};

\draw[rounded corners=18pt,fill=lightgray!25!white] (5.5, -3) rectangle (16.5, -5);
\node (E) at (11,-4.5) {\small Clause gadgets};

\node[vert2,fill=blue] (R1) at (-1,0) {};
\node[vert,fill=black] (A) at (1,0) {};
\node[vert,fill=black] (B) at (2,0) {};
\node[vert,fill=black] (C) at (3,0) {};
\node[vert] (XT1) at (4,1) {};
\node[vert] (XT2) at (5,1) {};
\node[vert] (XT3) at (6,1) {};
\node[vert] (XF1) at (4,-1) {};
\node[vert] (XF2) at (5,-1) {};
\node[vert] (XF3) at (6,-1) {};
\node[vert] (N1) at (7,0) {};

\node[vert,fill=black] (A2) at (10,0) {};
\node[vert,fill=black] (B2) at (11,0) {};
\node[vert,fill=black] (C2) at (12,0) {};
\node[vert] (YT1) at (13,1) {};
\node[vert] (YT2) at (14,1) {};
\node[vert] (YT3) at (15,1) {};
\node[vert] (YF1) at (13,-1) {};
\node[vert] (YF2) at (14,-1) {};
\node[vert] (YF3) at (15,-1) {};
\node[vert] (N2) at (16,0) {};

\node[vert,fill=black] (N3) at (16,-4) {};

\node[vert,fill=black] (CL1) at (11,-4) {};
\node[vert,fill=black] (CL2) at (6,-4) {};

\node[vert2,fill=blue] (R2) at (-1,-4) {};

\draw (R1) -- (A);

\draw (A) -- (B);
\draw (B) -- (C);
\draw (C) -- (XT1);
\draw (XT1) -- (XT2);
\draw (XT2) -- (XT3);
\draw (XT3) -- (N1);
\draw (C) -- (XF1);
\draw (XF1) -- (XF2);
\draw (XF2) -- (XF3);
\draw (XF3) -- (N1);

\draw[dotted] (N1) -- (A2);

\draw (A2) -- (B2);
\draw (B2) -- (C2);
\draw (C2) -- (YT1);
\draw (YT1) -- (YT2);
\draw (YT2) -- (YT3);
\draw (YT3) -- (N2);
\draw (C2) -- (YF1);
\draw (YF1) -- (YF2);
\draw (YF2) -- (YF3);
\draw (YF3) -- (N2);

\draw[dotted] (N2) to[out=-10,in=10] (N3);

\draw (N3) -- (CL1);
\draw (CL1) -- (CL2);
\draw[dotted] (CL2) -- (R2);

\draw[color=red, line width=1.5pt] (N3) -- (YF3);
\draw[color=red, line width=1.5pt] (CL1) -- (YF2);
\draw[color=red, line width=1.5pt] (CL1) -- (XF3);
\draw[color=red, line width=1.5pt] (CL2) -- (XF2);

\draw[color=red, line width=1.5pt, rounded corners=10pt] (N3) -- (16.5,-5) -- (18,-5) -- (18,2) -- (15,2) -- (YT3);

\draw[color=red, line width=1.5pt, rounded corners=10pt] (N3) -- (15.5,-5.5) -- (18.5,-5.5) -- (18.5,2.5) -- (6,2.5) -- (XT3);

\draw[color=red, line width=1.5pt, rounded corners=10pt] (CL2) -- (6,-5) -- (-2,-5) -- (-2,2) -- (4,2) -- (XT1);

\end{tikzpicture}
    \end{center}
    \caption{Illustration of an embedding described in the proof of \cref{obs:planar} of the graph $G$ produced by  \cref{constr:planarnph}. The black filled round vertices have obnoxious robots on them, and the blue filled square vertices have main robots on them. There are two variable gadgets on the left and three clause gadgets on the right (indicated by light gray boxes). The red lines correspond to paths of length $s$ that have obnoxious robots on all inner vertices.}\label{fig:nphardness}
\end{figure}

This finishes the construction, which can clearly be performed in polynomial time. It is fairly straightforward to check that the constructed graph $G$ is planar.
\begin{observation}\label{obs:planar}
The graph $G$ obtained from \cref{constr:planarnph} is planar.    
\end{observation}
\begin{proof}
    For an illustration, see \cref{fig:nphardness}. Recall that $H$ is the union of a Hamiltonian cycle that first visits all variables $X$ and then all clauses $Y$, and the incidence graph of $\phi$. Consider a planar embedding of $H$ where each edge between a clause and a variable that appears negated in the clause is inside the cycle, and each edge between a clause and a variable that appears non-negated in the clause is outside of the cycle. In \cref{constr:planarnph}, when we replace the vertex corresponding to a variable $x$ with variable gadget $G_x$, we can arrange it such that vertices $u_1^x,u_2^x,u_3^x$ point towards the outside of the cycle and $w_1^x,w_2^x,w_3^x$ to the inside. The gadget $G_x$ itself is clearly planar. Now we can embed the paths in $G$ connecting the clause vertices with the variable gadgets along the edges between clauses and variables in $H$. Note that in the construction of $G$, when we connect the variable gadgets with the clause vertices, we connect clause vertices to vertices of the variable gadget in the same cyclic ordering as the edges around the corresponding variable in $H$. It follows that we can embed the paths connecting variable gadgets and clause vertices in $G$ in a planar way. Finally, vertices $r_1,r_2$ and their incident edges can clearly be embedded in a planar way.  
\end{proof}

Next, we show correctness of the reduction, that is, the constructed instance $(G,R,O,c)$ is a yes-instance of \MoveToConnectedShort if and only if $(\phi, H)$ is a yes-instance of \textsc{Linked Planar 3-SAT}.

\begin{lemma}\label{lem:nphcorr1}
    If $(\phi, H)$ is a yes-instance of \textsc{Linked Planar 3-SAT}, then $(G,R,O,c)$ is a yes-instance of \MoveToConnectedShort.
\end{lemma}

\begin{proof}
Assume we are given a satisfying assignment for $\phi$. We create the following movement plan.
\begin{itemize}
\item We iterate over the variables in some arbitrary but fixed order. Let $x$ be a variable. If $x$ is set to true, then we move the three obnoxious robots on $v_1^x,v_2^x,v_3^x$ to $w_1^x,w_2^x,w_3^x$ in a collision-free way using 9 moves: We first move the robot on $v_3^x$ to $w_3^x$, then we move the robot on $v_2^x$ to $w_2^x$, and finally we move the robot on $v_1^x$ to $w_1^x$. It is straightforward to check that this requires 9 moves and is collision-free. If $x$ is set to false, then we move the three obnoxious robots on $v_1^x,v_2^x,v_3^x$ to $u_1^x,u_2^x,u_3^x$ in an analogous way. Doing this for each variable requires $9\cdot |X|$ moves in total.
\item Let $z$ be the last variable in the ordering (implied by the Hamiltonian cycle). We now move $r_1$ to $v_4^z$. Note that in each variable gadget for some variable $x$, the vertices $v_1^x,v_2^x,v_3^x,v_4^x$ are not occupied by any robot. Furthermore, either vertices $u_1^x,u_2^x,u_3^x$ or vertices $w_1^x,w_2^x,w_3^x$ are not occupied by any robot. It follows that we can move $r_1$ collision-free to $v_4^z$ using $7\cdot |X|+1$ moves.
\item Now for each clause $c$ we identify a literal $\ell$ that satisfies the clause. Let $x$ be the variable of literal $\ell$. If $x$ appears in $\ell$ non-negated, then by construction we have that $v^c$ is connected by a path of length $(s+1)$ to one of the vertices $u_1^x,u_2^x,u_3^x$ and no other literal of any other clause is connected to the same vertex as $\ell$. Furthermore, we have that these vertices are not occupied by any robot. Let $v^c$ be connected to $u_i^x$ with $i\in[3]$. The vertex $v^c$ as well as all inner vertices on the path of length $(s+1)$ from $v^c$ to $u_i^x$ are occupied by obnoxious robots. We move all robots towards $u_i^x$ by one step, starting with the one closest to $u_i^x$ and ending with the one occupying $v^c$. This takes $(s+1)$ moves and is clearly collision-free. Afterwards, vertex $v^c$ is not occupied by any vertex. If $x$ appears in $\ell$ negated, we have an analogous situation with vertices $w_1^x,w_2^x,w_3^x$ instead of $u_1^x,u_2^x,u_3^x$.
Doing this for each clause requires $(s+1)\cdot |Y|$ moves in total.
\item Let $c^\star$ be the last clause in the ordering (implied by the Hamiltonian cycle). We now move $r_1$ from $v_4^z$ to $v^{c^\star}$. This is collision-free since, after the previous moves, all vertices corresponding to clauses are not occupied by any robot, and it requires $|Y|+s+2$ moves.
\end{itemize}
After the above-described movement plan, main robot $r_1$ occupies vertex $v^{c^\star}$ and main robot~$r_2$ still occupies its starting location. Since these two vertices are connected, we have that the end-configuration of the movement plan forms a connected subgraph. Furthermore, the movement plan clearly has length $16\cdot |X|+(s+2)\cdot (|Y|+1)+1$. 
\end{proof}

\begin{lemma}\label{lem:nphcorr2}
    If $(G,R,O,c)$ is a yes-instance of \MoveToConnectedShort, then $(\phi, H)$ is a yes-instance of \textsc{Linked Planar 3-SAT}.
\end{lemma}
\begin{proof}
    Assume we are given a movement plan $\mathcal{M}$ that is a solution to $(G,R,O,c)$. We create a satisfying assignment for $\phi$ as follows.

    First, we start with the following observation: If a main robot moves (collision-free) along one of the paths of length $(s+1)$ from a vertex of a variable gadget to a clause vertex, all obnoxious robots need to move out of the way. This requires at least $2\cdot \sum_{i=1}^{\lfloor s+1\rfloor} i >s^2/4>c$ moves. Hence, we can conclude that in the movement plan $\mathcal{M}$, no main robot moves along such a path. It follows that each variable gadget is traversed by one of the main robots, and each clause vertex is traversed by one of the main robots. In particular, we can conclude that the main robots make at least $7\cdot |X|+|Y|+s+3$ moves.

    We can conclude the following: In each variable gadget, the three obnoxious robots need to be moved such that one of the main robots can pass through the gadget. Moving them inside the gadget to $u_1^x,u_2^x,u_3^x$ or $w_1^x,w_2^x,w_3^x$ requires 9 moves. Moving them in any other way that allows a main robot to pass (e.g.\ to an adjacent gadget) costs strictly more than 9 moves. Hence, at least $9\cdot |X|$ moves are required to move obnoxious robots in the variable gadgets.

    Now consider the clause vertices. Moving the obnoxious robot away from the clause vertex requires either moving each vertex on one of the paths of length $(s+1)$ be one step and hence requires $(s+1)$ moves in total, or moving the obnoxious robot into a path connected to a different clause vertex, , which requires at least $(s+2)$ moves, or moving the obnoxious robot into a variable gadget, which also requires at least $(s+2)$ moves. Hence, at least $(s+1)\cdot |Y|$ moves are required to move obnoxious robots away from the clause vertices.

    Adding up the above-derived lower bounds for the required moves yields exactly $c$. We can conclude that in each variable gadget, the three obnoxious robots move to either $u_1^x,u_2^x,u_3^x$ or $w_1^x,w_2^x,w_3^x$. Furthermore, for each clause vertex, all robots from the vertex along one of the paths of length $(s+1)$ to a variable gadget are moved by one step each.

    We construct an assignment for the variables in $X$ as follows. For each $x\in X$, if three obnoxious robots move to $u_1^x,u_2^x,u_3^x$, then we set $x$ to true. Otherwise, we set~$x$ to false. We claim that this is a satisfying assignment for $\phi$.

    Assume for contradiction that this is not a satisfying assignment. Then there is a clause $c$ that is not satisfied. Consider the corresponding clause vertex $v^c$. If none of $c$'s literals are satisfied, then by construction the vertices of the variable gadgets that are connected to $v^c$ by paths of length $(s+1)$ are occupied by obnoxious robots. It follows that at least $(s+2)$ moves are necessary to move the obnoxious robot from $v^c$. Then, however, by the arguments given above, the total number of moves is strictly larger than $c$, a contradiction.
\end{proof}

From \cref{obs:planar,lem:nphcorr1,lem:nphcorr2} follows that \MoveToConnectedShort is NP-hard even if $k=2$ and the input graph is planar. In the following, we show how to modify \cref{constr:planarnph} to produce a graph that is a grid.

First, we argue that we can find an alternative embedding for $G$ by ``unfolding'' it, such that $r_1$ and $r_2$ lie on the outer face, see \cref{fig:nphardness2} for an ``unfolded'' version of \cref{fig:nphardness}. More formally, we can ``break'' the Hamiltonian cycle in $H$ between the first variable vertex and the last clause vertex, creating a Hamiltonian path that first visits all variable vertices and then all clause vertices. We can draw this path along a horizontal line. Then all edges that were outside the cycle are now above the path, and all edges inside the cycle are now below the path. We take this embedding as a starting point to construct a modified graph $G'$ that is a grid (together with an embedding that shows this). Note that there are three main obstacles to overcome:
\begin{itemize}
    \item The clause vertices in $G$ may have degree 5, which is not possible in a grid. Hence, we have to modify the vertices to ``clause gadgets'' that are grids.
    \item All paths connecting clause gadgets to variable gadgets need to have the same length in order for the reduction to work.
    \item The path between the last variable gadget and the first clause gadget has to be sufficiently long (that is, longer than any path from a clause gadget to a variable gadget) for the reduction to work.
\end{itemize}
Formally, we construct $G'$ as follows.

\begin{figure}[t]
\begin{center}
\begin{tikzpicture}[line width=1pt,scale=.65,xscale=.75]

\draw[rounded corners=18pt,fill=lightgray!25!white] (.5, 4) rectangle (7.5, -1.5);
\node (E) at (4,3.5) {\small Variable gadget};

\draw[rounded corners=18pt,fill=lightgray!25!white] (9.5, 4) rectangle (16.5, -1.5);
\node (E) at (13,3.5) {\small Variable gadget};

\draw[rounded corners=18pt,fill=lightgray!25!white] (18.2, 1) rectangle (23.8, -1.5);

\node[vert2,fill=blue] (R1) at (-1,0) {};
\node[vert,fill=black] (A) at (1,0) {};
\node[vert,fill=black] (B) at (2,0) {};
\node[vert,fill=black] (C) at (3,0) {};
\node[vert] (XT1) at (4,1) {};
\node[vert] (XT2) at (5,1) {};
\node[vert] (XT3) at (6,1) {};
\node[vert] (XF1) at (4,-1) {};
\node[vert] (XF2) at (5,-1) {};
\node[vert] (XF3) at (6,-1) {};
\node[vert] (N1) at (7,0) {};

\node[vert,fill=black] (A2) at (10,0) {};
\node[vert,fill=black] (B2) at (11,0) {};
\node[vert,fill=black] (C2) at (12,0) {};
\node[vert] (YT1) at (13,1) {};
\node[vert] (YT2) at (14,1) {};
\node[vert] (YT3) at (15,1) {};
\node[vert] (YF1) at (13,-1) {};
\node[vert] (YF2) at (14,-1) {};
\node[vert] (YF3) at (15,-1) {};
\node[vert] (N2) at (16,0) {};

\node[vert,fill=black] (N3) at (19,0) {};

\node[vert,fill=black] (CL1) at (21,0) {};
\node[vert,fill=black] (CL2) at (23,0) {};

\node[vert2,fill=blue] (R2) at (26,0) {};

\draw (R1) -- (A);

\draw (A) -- (B);
\draw (B) -- (C);
\draw (C) -- (XT1);
\draw (XT1) -- (XT2);
\draw (XT2) -- (XT3);
\draw (XT3) -- (N1);
\draw (C) -- (XF1);
\draw (XF1) -- (XF2);
\draw (XF2) -- (XF3);
\draw (XF3) -- (N1);

\draw[dotted] (N1) -- (A2);

\draw (A2) -- (B2);
\draw (B2) -- (C2);
\draw (C2) -- (YT1);
\draw (YT1) -- (YT2);
\draw (YT2) -- (YT3);
\draw (YT3) -- (N2);
\draw (C2) -- (YF1);
\draw (YF1) -- (YF2);
\draw (YF2) -- (YF3);
\draw (YF3) -- (N2);

\draw[dotted] (N2) -- (N3);

\draw (N3) -- (CL1);
\draw (CL1) -- (CL2);
\draw[dotted] (CL2) -- (R2);

\draw[color=red, line width=1.5pt, rounded corners=10pt] (N3) -- (19,-2) -- (15,-2) -- (YF3);
\draw[color=red, line width=1.5pt, rounded corners=10pt] (CL1) -- (20,-2.5) -- (14,-2.5) -- (YF2);
\draw[color=red, line width=1.5pt, rounded corners=10pt] (CL1) -- (21,-3) -- (6,-3) -- (XF3);
\draw[color=red, line width=1.5pt, rounded corners=10pt] (CL2) -- (23,-3.5) -- (5,-3.5) -- (XF2);

\draw[color=red, line width=1.5pt, rounded corners=10pt] (N3) -- (18,2) -- (15,2) -- (YT3);

\draw[color=red, line width=1.5pt, rounded corners=10pt] (N3) -- (19,2.5) -- (6,2.5) -- (XT3);

\draw[color=red, line width=1.5pt, rounded corners=10pt] (CL2) -- (23,3) -- (4,3) -- (XT1);

\node[fill=lightgray!25!white,inner sep=1.1] (E) at (21,-1) {\small Clause gadgets};

\end{tikzpicture}
    \end{center}
    \caption{Illustration of a ``unfolded'' embedding of the graph $G$ produced by  \cref{constr:planarnph}. The blue filled round vertices have obnoxious robots on them, and the black filled square vertices have main robots on them. The red edges correspond to paths of length $s$ that have obnoxious robots on all inner vertices.}\label{fig:nphardness2}
\end{figure}
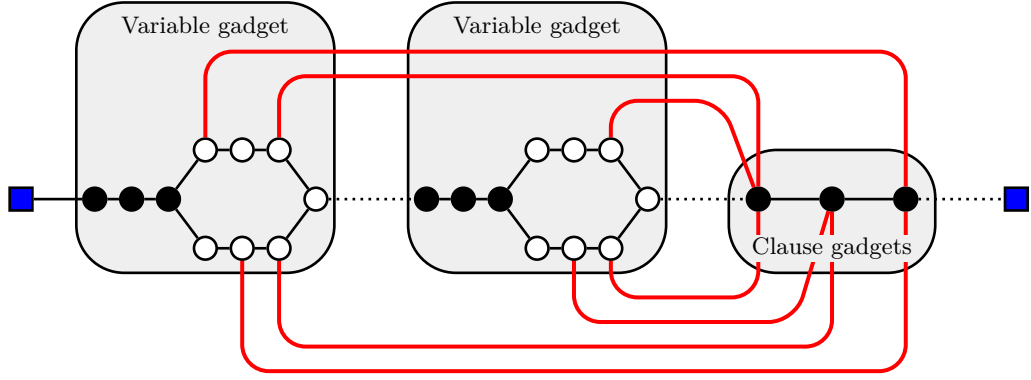

\begin{construction}\label{constr:gridnph}
Given an instance $(\phi, H)$ of \textsc{Linked Planar 3-SAT}, we construct an instance $(G'=(V',E'), R', O', c')$ of \MoveToConnectedShort with $k=2$ as follows. 
Let the variables $X$ and clauses $Y$ in $\phi$ be ordered in the way they appear in the Hamiltonian cycle in $H$.

For each variable $x\in X$, we create a variable gadget $G'_x=(V'_x, E'_x)$ with 
\begin{itemize}
    \item $V'_x=\{v_1^x,v_2^x,v_3^x,v_4^x,v_5^x,v_6^x,v_7^x,u_1^x,u_2^x,u_3^x,u_4^x,u_5^x,w_1^x,w_2^x,w_3^x,w_4^x,w_5^x\}$, and
    \item $E'_x=\{\{v_1^x,v_2^x\},\{v_2^x,v_3^x\},\{v_3^x,v_4^x\},\{v_4^x,v_5^x\},\{v_5^x,u_1^x\},\{u_1^x,u_2^x\},\{u_2^x,u_3^x\},\{u_3^x,u_4^x\},$
    
    $\{u_4^x,u_5^x\},\{u_5^x,v_6^x\},\{v_5^x,w_5^x\},\{w_5^x,w_4^x\},\{w_4^x,w_3^x\},\{w_3^x,w_2^x\},\{w_2^x,w_1^x\},$ 
    
    $\{w_1^x,v_6^x\},\{v_6^x,v_7^x\}\}$. 
\end{itemize}
Furthermore, we place obnoxious robots on the vertices $v_1^x,v_2^x,v_3^x,v_4^x,v_5^x$. For an illustration of the variable gadget see \cref{fig:vargadget2}.

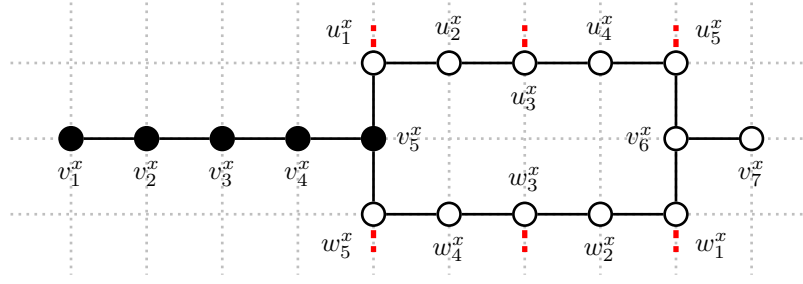
\begin{figure}[t]
\begin{center}
\begin{tikzpicture}[line width=1pt,scale=1,xscale=1]

\draw[color=lightgray,dotted] (0.2,0) -- (10.8,0);
\draw[color=lightgray,dotted] (0.2,1) -- (10.8,1);
\draw[color=lightgray,dotted] (0.2,-1) -- (10.8,-1);

\draw[color=lightgray,dotted] (1,1.8) -- (1,-1.8);
\draw[color=lightgray,dotted] (2,1.8) -- (2,-1.8);
\draw[color=lightgray,dotted] (3,1.8) -- (3,-1.8);
\draw[color=lightgray,dotted] (4,1.8) -- (4,-1.8);
\draw[color=lightgray,dotted] (5,1.8) -- (5,-1.8);
\draw[color=lightgray,dotted] (6,1.8) -- (6,-1.8);
\draw[color=lightgray,dotted] (7,1.8) -- (7,-1.8);
\draw[color=lightgray,dotted] (8,1.8) -- (8,-1.8);
\draw[color=lightgray,dotted] (9,1.8) -- (9,-1.8);
\draw[color=lightgray,dotted] (10,1.8) -- (10,-1.8);

\draw[line width=2pt, red, dashed] (5,1) -- (5,1.5);
\draw[line width=2pt, red, dashed] (7,1) -- (7,1.5);
\draw[line width=2pt, red, dashed] (9,1) -- (9,1.5);
\draw[line width=2pt, red, dashed] (5,-1) -- (5,-1.5);
\draw[line width=2pt, red, dashed] (7,-1) -- (7,-1.5);
\draw[line width=2pt, red, dashed] (9,-1) -- (9,-1.5);

\node[vert,fill=black,label=below:$v_1^x$] (A) at (1,0) {};
\node[vert,fill=black,label=below:$v_2^x$] (B) at (2,0) {};
\node[vert,fill=black,label=below:$v_3^x$] (C) at (3,0) {};
\node[vert,fill=black,label=below:$v_4^x$] (D) at (4,0) {};
\node[vert,fill=black,label=right:$v_5^x$] (E) at (5,0) {};
\node[vert,label=above left:$u_1^x$] (XT1) at (5,1) {};
\node[vert,label=above:$u_2^x$] (XT2) at (6,1) {};
\node[vert,label=below:$u_3^x$] (XT3) at (7,1) {};
\node[vert,label=above:$u_4^x$] (XT4) at (8,1) {};
\node[vert,label=above right:$u_5^x$] (XT5) at (9,1) {};
\node[vert,label=below left:$w_5^x$] (XF1) at (5,-1) {};
\node[vert,label=below:$w_4^x$] (XF2) at (6,-1) {};
\node[vert,label=above:$w_3^x$] (XF3) at (7,-1) {};
\node[vert,label=below:$w_2^x$] (XF4) at (8,-1) {};
\node[vert,label=below right:$w_1^x$] (XF5) at (9,-1) {};
\node[vert,label=left:$v_6^x$] (N1) at (9,0) {};
\node[vert,label=below:$v_7^x$] (N2) at (10,0) {};

\draw (A) -- (B);
\draw (B) -- (C);
\draw (C) -- (D);
\draw (D) -- (E);
\draw (E) -- (XT1);
\draw (XT1) -- (XT2);
\draw (XT2) -- (XT3);
\draw (XT3) -- (XT4);
\draw (XT4) -- (XT5);
\draw (XT5) -- (N1);
\draw (E) -- (XF1);
\draw (XF1) -- (XF2);
\draw (XF2) -- (XF3);
\draw (XF3) -- (XF4);
\draw (XF4) -- (XF5);
\draw (XF5) -- (N1);
\draw (N1) -- (N2);

\end{tikzpicture}
    \end{center}
    \caption{Illustration of a variable gadget $G'_x$ for a variable $x$ (\cref{constr:gridnph}). The black filled vertices have obnoxious robots on them. Red dashed lines indicate where the paths between variable gadgets and clause gadgets are connected.}\label{fig:vargadget2}
\end{figure}

For each variable $c\in Y$, we create a variable gadget $G'_c=(V'_c, E'_c)$ with 
\begin{itemize}
    \item $V'_c=\{v_1^c,v_2^c,v_3^c,v_4^c,v_5^c,v_6^c,u_1^c,u_2^c,u_3^c,u_4^c,u_5^c,w_1^c,w_2^c,w_3^c,w_4^c,w_5^c\}$, and
    \item $E'_c=\{\{v_1^c,v_2^c\},\{v_2^c,v_3^c\},\{v_3^c,v_4^c\},\{v_4^c,v_5^c\},\{v_3^c,u_3^c\},\{u_1^c,u_2^c\},\{u_2^c,u_3^c\},\{u_3^c,u_4^c\},$
    
    $\{u_4^c,u_5^c\},\{v_3^c,w_3^c\},\{w_5^c,w_4^c\},\{w_4^c,w_3^c\},\{w_3^c,w_2^c\},\{w_2^c,w_1^c\}\}$. 
\end{itemize}
Furthermore, we place obnoxious robots on the vertices $v_3^c,u_1^c,u_2^c,u_3^c,u_4^c,u_5^c,w_1^c,w_2^c,w_3^c,w_4^c,w_5^c$. For an illustration of the clause gadget see \cref{fig:clgadget2}.

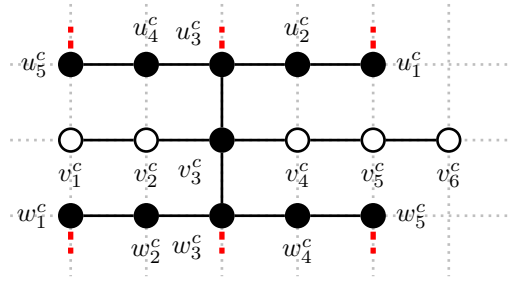
\begin{figure}[t]
\begin{center}
\begin{tikzpicture}[line width=1pt,scale=1,xscale=1]

\draw[color=lightgray,dotted] (-1.8,0) -- (4.8,0);
\draw[color=lightgray,dotted] (-1.8,1) -- (4.8,1);
\draw[color=lightgray,dotted] (-1.8,-1) -- (4.8,-1);

\draw[color=lightgray,dotted] (4,1.8) -- (4,-1.8);
\draw[color=lightgray,dotted] (-1,1.8) -- (-1,-1.8);
\draw[color=lightgray,dotted] (0,1.8) -- (0,-1.8);
\draw[color=lightgray,dotted] (1,1.8) -- (1,-1.8);
\draw[color=lightgray,dotted] (2,1.8) -- (2,-1.8);
\draw[color=lightgray,dotted] (3,1.8) -- (3,-1.8);

\draw[line width=2pt, red, dashed] (1,1) -- (1,1.5);
\draw[line width=2pt, red, dashed] (3,1) -- (3,1.5);
\draw[line width=2pt, red, dashed] (-1,1) -- (-1,1.5);
\draw[line width=2pt, red, dashed] (1,-1) -- (1,-1.5);
\draw[line width=2pt, red, dashed] (3,-1) -- (3,-1.5);
\draw[line width=2pt, red, dashed] (-1,-1) -- (-1,-1.5);

\node[vert,label=below:$v_1^c$] (A) at (-1,0) {};
\node[vert,label=below:$v_2^c$] (B) at (0,0) {};
\node[vert,fill=black,label=below left:$v_3^c$] (C) at (1,0) {};
\node[vert,label=below:$v_4^c$] (D) at (2,0) {};
\node[vert,label=below:$v_5^c$] (E) at (3,0) {};
\node[vert,label=below:$v_6^c$] (F) at (4,0) {};
\node[vert,fill=black,label=left:$u_5^c$] (XT1) at (-1,1) {};
\node[vert,fill=black,label=above:$u_4^c$] (XT2) at (0,1) {};
\node[vert,fill=black,label=above left:$u_3^c$] (XT3) at (1,1) {};
\node[vert,fill=black,label=above:$u_2^c$] (XT4) at (2,1) {};
\node[vert,fill=black,label=right:$u_1^c$] (XT5) at (3,1) {};
\node[vert,fill=black,label=left:$w_1^c$] (XF1) at (-1,-1) {};
\node[vert,fill=black,label=below:$w_2^c$] (XF2) at (0,-1) {};
\node[vert,fill=black,label=below left:$w_3^c$] (XF3) at (1,-1) {};
\node[vert,fill=black,label=below:$w_4^c$] (XF4) at (2,-1) {};
\node[vert,fill=black,label=right:$w_5^c$] (XF5) at (3,-1) {};

\draw (A) -- (B);
\draw (B) -- (C);
\draw (C) -- (D);
\draw (D) -- (E);
\draw (E) -- (F);
\draw (C) -- (XT3);
\draw (XT1) -- (XT2);
\draw (XT2) -- (XT3);
\draw (XT3) -- (XT4);
\draw (XT4) -- (XT5);
\draw (C) -- (XF3);
\draw (XF1) -- (XF2);
\draw (XF2) -- (XF3);
\draw (XF3) -- (XF4);
\draw (XF4) -- (XF5);

\end{tikzpicture}
    \end{center}
    \caption{Illustration of a clause gadget $G'_c$ for a clause $c$ (\cref{constr:gridnph}). The black filled vertices have obnoxious robots on them. Red dashed lines indicate where the paths between variable gadgets and clause gadgets are connected.}\label{fig:clgadget2}
\end{figure}

We construct the graph $G'$ of the \MoveToConnectedShort instance by first taking the graph $H$ and replacing each vertex corresponding to a variable, say $x$, with variable gadget $G'_x$ in the following way. If two vertices corresponding to variables $x,y$ are adjacent in $H$, that is, they appear consecutive in the ordering (first $x$ and then $y$), then we add edge $\{v_7^x,v_1^y\}$ to $G'$. 
Let $x$ be the last variable in the ordering and let $c$ be the first clause in the ordering, then we add edge $\{v_7^x,v_1^c\}$ to $G'$ and subdivide it $s+4$ times, where $s=90\cdot(|X|+|Y|)-4$.
For each clause $c\in C$ we replace the corresponding vertex with clause gadget $G'_c$ in the following way.
If two clauses $c,c'$ appear consecutive in the ordering (first $c$ and then $c'$), then we add edge $\{v_6^c,v_1^{c'}\}$ to $G$. 
 Let variable~$x$ appear in clauses $c_1,c_2,c_3$. 
 Then the vertex $x$ in $H$ is connected to vertices $v^{c_1},v^{c_2},v^{c_3}$ in~$H$. Let w.l.o.g.\ the order $c_1,c_2,c_3$ be the linearization of the clockwise cyclic ordering in which the edges incident with $x$ to the vertices $c_1,c_2,c_3$ are embedded around $x$ in the planar embedding of $H$, such that all clauses in which $x$ appears non-negated are ordered before the clauses in which $x$ appears negated. 
 For each clause $c$, let w.l.o.g.\ the order $\ell_c^1,\ell_c^2,\ell_c^3$ the linearization of the anti-clockwise cyclic ordering in which edges incident with $c$ to the variables corresponding to literals $\ell_c^1,\ell_c^2,\ell_c^3$ are embedded around $c$.
 For $i\in\{1,2,3\}$, if~$x$ appears in $c_i$ non-negated as the $j$th literal according to the above ordering, then we connect $u_i^x$ with $u^{c_i}_j$ with an edge. If $j\neq 2$ and subdivide the edge~$s$ times and otherwise $s+2$ times. Furthermore, we place an obnoxious robot on each vertex created by a subdivision. If $x$ appears in~$c_i$ negated as the $j$th literal according to the above ordering, then we connect $w_i^x$ with $w^{c_i}_j$ with an edge. If $j\neq 2$ and subdivide the edge~$s$ times and otherwise $s+2$ times. Furthermore, we place an obnoxious robot on each vertex created by a subdivision. 
Finally, we add two vertices $r_1,r_2$ to $G$ and place a main robot on each of them. Let $x$ be the first variable in the ordering, then we add edge $\{r_1,v_1^x\}$ to $G$. Let $c$ be the last clause in the ordering, then we add edge $\{r_2,v_6^c\}$ to $G$.
For an illustration, see \cref{fig:nphardnessgrid}. Finally, we set $c=37\cdot |X|+(s+8)\cdot |Y|+s+5$. We will assume in the proofs of correctness that $s^2/4>c$. Note that this is clearly the case if $|X|$ and $|Y|$ are sufficiently large.
\end{construction}

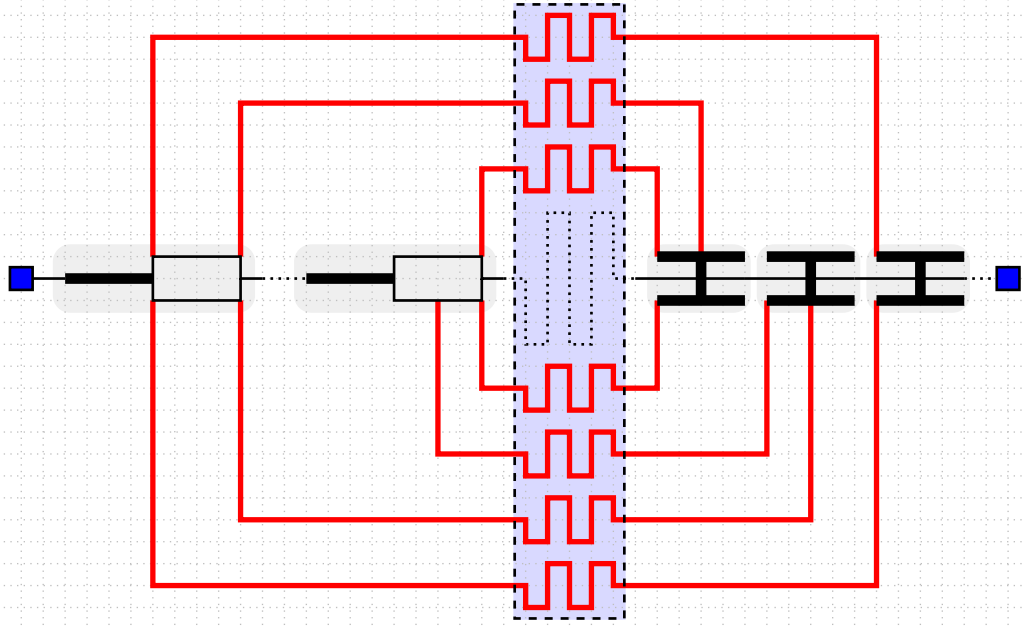
\begin{figure}[t]
\begin{center}
\begin{tikzpicture}[line width=1pt,scale=.29,xscale=1]

\draw[blue!15!white,fill=blue!15!white] (20.5, 12.5) rectangle (25.5, -15.5);

\draw[rounded corners=6pt,lightgray!25!white,fill=lightgray!25!white] (-.5, 1.5) rectangle (8.6, -1.5);

\draw[rounded corners=6pt,lightgray!25!white,fill=lightgray!25!white] (10.5, 1.5) rectangle (19.6, -1.5);

\draw[rounded corners=6pt,lightgray!25!white,fill=lightgray!25!white] (26.6, 1.5) rectangle (31.2, -1.5);

\draw[rounded corners=6pt,lightgray!25!white,fill=lightgray!25!white] (31.6, 1.5) rectangle (36.2, -1.5);

\draw[rounded corners=6pt,lightgray!25!white,fill=lightgray!25!white] (36.6, 1.5) rectangle (41.2, -1.5);

\foreach \i in {-15,...,12}
{
    \draw[line width=.5pt,color=lightgray,dotted] (-2.8,\i) -- (43.8,\i);
}

\foreach \i in {-2,...,43}
{
    \draw[line width=.5pt,color=lightgray,dotted] (\i,-15.8) -- (\i,12.8);
}

\draw[line width=2pt,red] (27,1) -- (27,5) -- (25,5) -- (25,6) -- (24,6) -- (24,4) -- (23,4) -- (23,6) -- (22,6) -- (22,4) -- (21,4) -- (21,5) -- (19,5) -- (19,1);

\draw[line width=2pt,red] (29,1) -- (29,8) -- (25,8) -- (25,9) -- (24,9) -- (24,7) -- (23,7) -- (23,9) -- (22,9) -- (22,7) -- (21,7) -- (21,8) -- (8,8) -- (8,1);

\draw[line width=2pt,red] (37,1) -- (37,11) -- (25,11) -- (25,12) -- (24,12) -- (24,10) -- (23,10) -- (23,12) -- (22,12) -- (22,10) -- (21,10) -- (21,11) -- (4,11) -- (4,1);

\draw[line width=2pt,red] (27,-1) -- (27,-5) -- (25,-5) -- (25,-4) -- (24,-4) -- (24,-6) -- (23,-6) -- (23,-4) -- (22,-4) -- (22,-6) -- (21,-6) -- (21,-5) -- (19,-5) -- (19,-1);

\draw[line width=2pt,red] (32,-1) -- (32,-8) -- (25,-8) -- (25,-7) -- (24,-7) -- (24,-9) -- (23,-9) -- (23,-7) -- (22,-7) -- (22,-9) -- (21,-9) -- (21,-8) -- (17,-8) -- (17,-1);

\draw[line width=2pt,red] (34,-1) -- (34,-11) -- (25,-11) -- (25,-10) -- (24,-10) -- (24,-12) -- (23,-12) -- (23,-10) -- (22,-10) -- (22,-12) -- (21,-12) -- (21,-11) -- (8,-11) -- (8,-1);

\draw[line width=2pt,red] (37,-1) -- (37,-14) -- (25,-14) -- (25,-13) -- (24,-13) -- (24,-15) -- (23,-15) -- (23,-13) -- (22,-13) -- (22,-15) -- (21,-15) -- (21,-14) -- (4,-14) -- (4,-1);

\draw[dashed] (25.5,12.5) -- (25.5,-15.5) -- (20.5,-15.5) -- (20.5,12.5) -- (25.5,12.5);

\draw (-2,0) -- (0,0);
\draw[line width=4pt] (0,0) -- (4,0);
\draw (8,0) -- (8,-1) -- (4,-1) -- (4,1) -- (8,1) -- (8,0) -- (9,0);

\draw[dotted] (9,0) -- (11,0);
\draw[line width=4pt] (11,0) -- (15,0);
\draw (19,0) -- (19,-1) -- (15,-1) -- (15,1) -- (19,1) -- (19,0) -- (20,0);

\draw[dotted] (20,0) -- (21,0) -- (21,-3) -- (22,-3) -- (22,3) -- (23,3) -- (23,-3) -- (24,-3) -- (24,3) -- (25,3) -- (25,0) -- (26,0);

\draw[line width=4pt] (27,1) -- (31,1);
\draw[line width=4pt] (27,-1) -- (31,-1);
\draw[line width=4pt] (29,1) -- (29,-1);
\draw (26,0) -- (31,0);

\draw[line width=4pt] (32,1) -- (36,1);
\draw[line width=4pt] (32,-1) -- (36,-1);
\draw[line width=4pt] (34,1) -- (34,-1);
\draw (31,0) -- (36,0);

\draw[line width=4pt] (37,1) -- (41,1);
\draw[line width=4pt] (37,-1) -- (41,-1);
\draw[line width=4pt] (39,1) -- (39,-1);
\draw (36,0) -- (41,0);

\draw[dotted] (41,0) -- (43,0);

\node[vert2,fill=blue] (R1) at (-2,0) {};
\node[vert2,fill=blue] (R2) at (43,0) {};

\end{tikzpicture}
    \end{center}
    \caption{Illustration of an embedding described in the proof of \cref{lem:grid} of the graph $G'$ produced by  \cref{constr:gridnph}. The drawing is not up to scale but rather should give the rough idea on how the embedding looks like. The grid points on the thick black lines have obnoxious robots on them, and the blue filled square vertices have main robots on them. The red edges correspond to long paths that have obnoxious robots on all inner vertices. The light blue area surrounded by the black dashed line indicates the part of the grid, where we make sure that all red paths have the same length, and the path connecting the last variable gadget with the first clause gadget is sufficiently long. This is also indicated by the snake-like shape of the paths in the area.}\label{fig:nphardnessgrid}
\end{figure}

This finishes the construction, which can clearly be performed in polynomial time. In the following, we argue that $G'$ is indeed a grid.
\begin{lemma}\label{lem:grid}
The graph $G'$ obtained from \cref{constr:gridnph} is a grid.    
\end{lemma}
\begin{proof}
For an illustration see \cref{fig:nphardnessgrid}. 
The main idea of the embedding is as follows. It has a very similar structure as the embedding of $G$ (by \cref{constr:planarnph}) illustrated in \cref{fig:nphardness2}. Informally speaking, we arrange the variable gadgets in a line (more specifically, according to \cref{fig:vargadget2} and the next leftmost vertex of the next gadget connects via a horizontal edge to the rightmost vertex of the current gadget). Afterwards, we leave a horizontal ``gap'' of size $10\cdot (|X|+|Y|)$ and then we arrange the clause gadgets (\cref{fig:clgadget2}) in a line. This gap corresponds to the horizontal size of the light blue area in \cref{fig:nphardnessgrid}, where the paths make the snake-like shaped. Informally speaking, this gap is large enough (and the paths are sufficiently far apart from each other), that we can realize paths of length $s$, $s+2$, or $s+4$ no matter from which vertices exactly they start and finish, and no matter how far up or dowm from the middle we embed those paths.

Formally, we start replacing the variable gadget of $G'$ in the way shown in \cref{fig:vargadget2} on the grid, such that the vertex $v_1^x$ for the first variable $x$ according to the ordering is placed on grid point $(0,0)$. Then we place the second variable gadget next to it, formally also in the way depicted in \cref{fig:vargadget2} such that $v_1^{x'}$ is (where $x'$ is the second variable) is placed on grid point $(10,0)$, and so on.
After the last variable gadget, we keep $10\cdot (|X|+|Y|)$ free grid points (horizontally). This corresponds to the horizontal size of the light blue area on \cref{fig:nphardnessgrid}. Then we embed the clause gadgets according to \cref{fig:clgadget2}. Formally, the first clause gadget is positions such that $v_1^c$ (where $c$ is the first clause) is placed on grid point $(10\cdot |X|+10\cdot (|X|+|Y|),0)$, and so on.
Note that by \cref{constr:gridnph}, the last variable gadget and the first clause gadget are connected by a path of length $s+4=90\cdot(|X|+|Y|)$. This we can do by arranging the path snake-like, going 4 grid-cells up and 4 grid-cells down, as depicted in \cref{fig:nphardnessgrid}. 
We embed each path connecting a clause gadget with a variable gadget as follows. Assume that the path connects the $j$th literal (for $j\in [3]$) of the $i$th clause with a variable gadget. Assume that the literal is non-negated. Then we embed the path upward using $10 \le 10\cdot (3i+j-3)\le 30\cdot |Y|$ grid points.
Note that this ensures that any two paths are at least 10 horizontal grid points apart from each other.
The vertex of the variable gadget to which the path connects is at most $6\cdot |Y|+10\cdot |X|+10\cdot (|X|+|Y|)$ and at least $10(|X|+|Y|)$ grid points to the left. Note that $66\cdot |Y|+10\cdot |X|+10\cdot (|X|+|Y|)\le s$, hence embedding a path like this never creates a path that is already too long. If the path is too short when embedded like this, we can add snake-like shapes going 4 grid-cells up and 4 grid-cells down (note that does not interfere with any other path, since they are at least 10 grid cells above or below) until the path has the correct length. Here it is important to see that the connecting points of variable gadgets and clause gadgets have an even horizontal distance and the ``initial'' embedding has even length. By adding snake-like shaped, we can make the path longer by an arbitrary even number that is smaller than $s+4$. Hence, we can ensure all paths can be embedded in the described way by adding an appropriate amount of snake-like shapes.
If the literal is negated. Then we embed the path downward using $10 \le 10\cdot (3i+j-3)\le 30\cdot |Y|$ grid cells. The rest of the embedding is analogous.
\end{proof}

Now we have all ingredients to prove \cref{thm:nphardness}.

\begin{proof}[Proof of \cref{thm:nphardness}]
    By \cref{lem:grid} we have that the graph produced by \cref{constr:gridnph} is a grid. Clearly, there are only two main robots. The correctnenss can be shown analogously to \cref{lem:nphcorr1,lem:nphcorr2} by adjusting the movement steps needed in the gadgets and to traverse the gadgets. We omit the details.
\end{proof}


\section{Collision Free Movement on Grids}\label{sec:grids}

In this section, we present parameterized algorithms for \MoveToPi parameterized by $k+c$ on grids. Specifically, we give kernelization algorithms that produce problems kernels of size $O(k\cdot c^2)$ for \MoveToPi and $O(\min\{k\cdot c^2,k^2+c^2\})$ for \MoveToConnectedShort. This means, if $c\in \omega(\sqrt{k})$, we have a strictly better kernel size for \MoveToConnectedShort. These kernelization algorithms can then e.g.\ be combined with \cref{thm:xp} to obtain an algorithm with a running time in $(k\cdot c)^{O(c)}+n^{O(1)}$.


\subparagraph{Kernelization.} To obtain the kernelization results, we give a number of straightforward reduction rules.  A \emph{reduction rule} is a description of a procedure that takes a problem instance as input and outputs an instance of the same problem. We say that a reduction rule is \emph{safe}, if it holds that the input instance is a yes-instance if and only if the output instance is a yes-instance. First, we prove \cref{thm:kernel} (restated below).

\kernel*

To prove \cref{thm:kernel}, we present the following simple reduction rule. Essentially, it allows us to remove vertices that are too far away from all main robots. The intuition here is that if a vertex $v$ has distance more than $c$ to any main robot starting location, then it cannot be visited by any main robot, and it also cannot be visited by any obnoxious robot that \emph{needs} to move in order to avoid collision.
\begin{reductionrule}\label{rr:1}
    Let $(G=(V,E),R,O,c)$ be an instance of \MoveToPi. If there is a $v\in V$ such that for all $r\in R$ it holds that $\dist_G(v,r)>c$, then remove $v$ from~$G$ (and from $O$, if it is contained in $O$).
\end{reductionrule}

\cref{rr:1} can clearly be applied in polynomial time and it will also be useful in the algorithms we present in later sections. Next, we prove that it can be safely applied.

\begin{lemma}\label{lem:rr1safe}
    \cref{rr:1} is safe.
\end{lemma}
\begin{proof}
    Assume for contradiction that there is a vertex $v$ that \cref{rr:1} would remove and to or from which a robot is moved in a solution to the instance. Assume that the solution has a minimum number of movements. Note that this cannot be a main robot, since all of them have distance more than $c$ to $v$ and hence cannot be moved to $v$. It follows that an obnoxious robot is moved to $v$ (or moved away from $v$). In a solution with a minimum number of movements, an obnoxious robot is only moved to avoid a collision with a main robot or another obnoxious robot. In the latter case, this property holds recursively for the other obnoxious robot. We can conclude that in the end-configuration, there is a path from $v$ to a vertex $v'$ with a main robot such that all vertices of the path (except $v'$) have obnoxious robots on them an each one has been moved by at least one step. Let $\ell'$ be the number of obnoxious robots that have been moved. Then the distance of $v'$ to any main robot starting location is at least $c-\ell'+1$. It follows that moving a main robot to $v'$ and moving the $\ell'$ obnoxious robots by one step each requires more that $c$ movements in total, a contradiction.
\end{proof}

Now we can prove \cref{thm:kernel}.
\begin{proof}[Proof of \cref{thm:kernel}]
We apply \cref{rr:1} exhaustively, which can clearly be done in polynomial time.
By \cref{lem:rr1safe} we have that the rule is safe.
The definition of \cref{rr:1} implies that we keep a ball of radius $c$ around every vertex that has a main robot. In a grid, a ball of radius $c$ around a vertex contains $O(c^2)$ vertices. It follows that after exhaustively applying \cref{rr:1}, the remaining graph $G$ has $O(k\cdot c^2)$ vertices.    
\end{proof}

Now we refine the kernelization algorithm for \MoveToConnectedShort by introducing some additional reduction rules with which we can achieve a smaller kernel if $k<c^2$. 
This allows us to prove \cref{thm:kernel2} (restated below).

\kernell*

To show \cref{thm:kernel2}, we introduce an additional reduction rule. Informally, this rule says that if a vertex $v$ is too far away (distance more than $c+k$) from \emph{some} main robot $r$'s starting location, then is is irrelevant. 

\begin{reductionrule}\label{rr:2}
    Let $(G=(V,E),R,O,c)$ be an instance of \MoveToConnectedShort. If there is a $v\in V\setminus R$ such that there exists $r\in R$ such that $\dist_G(v,r)>c+ k$, then remove $v$ from~$G$ (and from $O$, if it is contained in $O$).
\end{reductionrule}

\cref{rr:2} can clearly be applied in polynomial time. Next we prove that it is safe.

\begin{lemma}\label{lem:rr2safe}
    \cref{rr:2} is safe.
\end{lemma}
\begin{proof}
    Assume for contradiction that there is a vertex $v$ that \cref{rr:2} would remove and to or from which a robot is moved in a solution to the instance. Assume that the solution has a minimum number of movements. Let $r\in R$ such that $\dist_G(v,r)>c+ k$.
    Assume that an obnoxious robot is moved to $v$ (or moved away from $v$). In a solution with a minimum number of movements, an obnoxious robot is only moved to avoid a collision with a main robot or another obnoxious robot. In the latter case, this property holds recursively for the other obnoxious robot. We can conclude that in the end-configuration, there is a path from $v$ to a vertex $v'$ with a main robot such that all vertices of the path (except $v'$) have obnoxious robots on them an each one has been moved by at least one step. Let $\ell'$ be the number of obnoxious robots that have been moved. 
    If a main robot is moved to or from $v$, set $\ell'=0$ and $v=v'$ in the following argument.
    Then the distance of $v'$ to $r$ is at least $c+k-\ell'+1$. Note that at most $k$ vertices on any path from $v'$ to $r$ can have main robots on them. It follows that connecting the main robot with starting location $r$ to the main robot on $v'$ and moving the $\ell'$ obnoxious robots by one step each requires more that $c$ movements in total, a contradiction.
\end{proof}




Now we can prove \cref{thm:kernel2}.
\begin{proof}[Proof of \cref{thm:kernel2}]
We apply \cref{rr:1,rr:2} exhaustively, which can clearly be done in polynomial time.
By \cref{lem:rr1safe,lem:rr2safe} we have that the rules are safe.
By \cref{thm:kernel} we know that after exhaustively applying \cref{rr:1}, the remaining graph $G$ has $O(k\cdot c^2)$ vertices. After applying \cref{rr:2} we have that in particular, for any fixed $r\in R$, we have that for all $v\in V\setminus R$ it holds that $\dist_G(v,r)\le c+ k$. It follows that we keep a ball of radius $c+ k$ around $r$ that contains all vertices except possibly some vertices in $R$. In a grid, a ball of radius $c+k$ around a vertex contains $O(c^2+k^2)$ vertices. It follows that after exhaustively applying \cref{rr:2}, the remaining graph $G$ has $O(k^2+ c^2)$ vertices. The theorem statement follows.    
\end{proof}

\subparagraph{Kernelization Lower Bounds.}
Now we present several kernelization lower bounds.
We use the popular OR-cross-compositions framework~\cite{BJK14,Fom+19} to refute the existence of a polynomial kernel for a parameterized problem under the assumption that \NNoKernelAssume, the negation of which would cause a collapse of the polynomial-time hierarchy to the third
level. 
In order to formally define OR-cross-compositions, we introduce equivalence relations first.
An equivalence
relation~$R$ on the instances of some problem~$L$ is a
\emph{polynomial equivalence relation}~if
\begin{enumerate}
 \item one can decide for every two instances in time polynomial in their sizes whether they belong to the same equivalence class, and
 \item for each finite set~$S$ of instances, $R$ partitions the set into at most~$(\max_{x \in S} |x|)^{\OO(1)}$ equivalence classes.  
\end{enumerate}
An \emph{OR-cross-composition} of a problem~$L\subseteq \{0,1\}^*$ into a
parameterized problem~$P$ (with respect to a polynomial equivalence
relation~$R$ on the instances of~\(L\)) is an algorithm that takes
$t$ $R$-equivalent instances~$x_1,\ldots,x_t$ of~$L$ and
constructs in time polynomial in $\sum_{i=1}^t |x_i|$ an instance
$(y,k)$ of~\(P\) such that
\begin{enumerate}
\item $k\le (\max_{i\in [t]}|x_i|+\log t)^{O(1)}$, and 
\item $(y,k)$ is a \yes-instance of $P$ if and only if there is an $i\in [t]$ s.t.\ $x_{i}$ is a \yes-instance of~$L$. 
\end{enumerate}
If an \NP-hard problem~\(L\) OR-cross-composes into a parameterized
problem~$P$, then~$P$ does not admit a polynomial kernel, unless \NoKernelAssume~\cite{BJK14,Fom+19}.

First, we show \cref{thm:nopkplanar} (restated below), which implies that we cannot get similar kernelization results as \cref{thm:kernel,thm:kernel2} for planar graphs.

\nopkplanar*

\begin{proof}
We provide an OR-cross-composition from \textsc{Linked Planar 3-SAT}~\cite{Pilz19} based on \cref{constr:planarnph} that is used in the proof of \cref{thm:nphardness}. 
We say that two instances of \textsc{Linked Planar 3-SAT} are equivalent if the two formulas have the same number of variables and the same number of clauses. 
Assume we are given $t$ equivalent instances $x_1,x_2,\ldots,x_t$ of \textsc{Linked Planar 3-SAT}.
We apply \cref{constr:planarnph} to each instance of \textsc{Linked Planar 3-SAT}. We identify the vertices $r_1$ and $r_2$ (the start positions of the main robots) of all instances. Then each edge incident with $r_1$ and $r_2$ we subdivide $s+2=|X|+|Y|+2$ times, where $|X|$ and $|Y|$ are the number of variables and clauses, respectively. We set $c=16\cdot |X|+(s+2)\cdot (|Y|+3)+1$. This finishes the construction. For an illustration see \cref{fig:nopk1}.

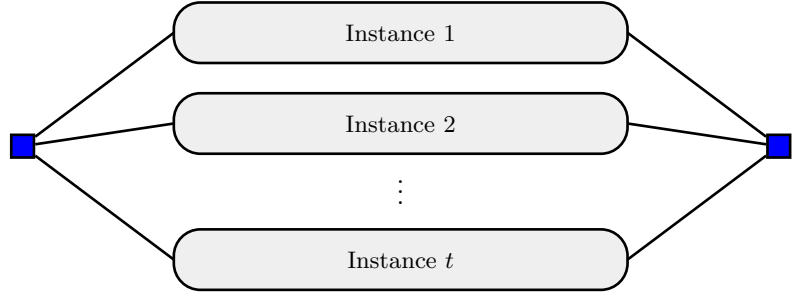
\begin{figure}[t]
\begin{center}
\begin{tikzpicture}[line width=1pt,scale=1,yscale=.8]

\draw[rounded corners=10pt,fill=lightgray!25!white] (0, 0) rectangle (6, 1);
\node (E) at (3,.5) {\small Instance 1};

\draw[rounded corners=10pt,fill=lightgray!25!white] (0, -.5) rectangle (6, -1.5);
\node (E) at (3,-1) {\small Instance 2};

\node (E) at (3, -2) {$\vdots$};

\draw[rounded corners=10pt,fill=lightgray!25!white] (0, -2.75) rectangle (6, -3.75);
\node (E) at (3,-3.25) {\small Instance $t$};

\node[vert2,fill=blue] (R1) at (-2,-1.375) {};
\node[vert2,fill=blue] (R2) at (8,-1.375) {};

\draw (R1) -- (0,.5);
\draw (R1) -- (0,-1);
\draw (R1) -- (0,-3.25);

\draw (R2) -- (6,.5);
\draw (R2) -- (6,-1);
\draw (R2) -- (6,-3.25);

\end{tikzpicture}
    \end{center}
    \caption{Illustration of the OR-cross-composition from \cref{thm:nopkplanar}. Blue filled square vertices have main robots on them.}\label{fig:nopk1}
\end{figure}

The edge subdivisions ensure that it is too expensive to move obnoxious robots between the instances. Intuitively, the allowed movement amount forces $r_1$ and $r_2$ to meet through one of the instances. The formal correctness is analogous to \cref{lem:nphcorr1,lem:nphcorr2}. 
\end{proof}

Next, we prove \cref{thm:nopkgrid} (restated below), that is, we presumably cannot drop the parameter $k$ from the kernelization results in \cref{thm:kernel,thm:kernel2}.

\nopkgrid*

\begin{proof}
We provide an OR-cross-composition from \textsc{Linked Planar 3-SAT}~\cite{Pilz19} based on \cref{constr:gridnph} that is used in the proof of \cref{thm:nphardness}. 
We say that two instances of \textsc{Linked Planar 3-SAT} are equivalent if the two formulas have the same number of variables and the same number of clauses. 
Assume we are given $t$ equivalent instances $x_1,x_2,\ldots,x_t$ of \textsc{Linked Planar 3-SAT}.
We apply \cref{constr:gridnph} to each instance of \textsc{Linked Planar 3-SAT} with one small modification: we do not place any obnoxious robots in the variable gadgets. 
Note that the created instances all have the same horizontal size when embedded into a grid using the description on \cref{lem:grid}. We place the instances on top of each other (aligned) such that they to not overlap, see \cref{fig:nopk2} for an illustration.
Then in each instance, the we subdivide the edges incident with $r_1$ and $r_2$, respectively, $c+1$ times, where we specify $c$ later, respectively. We put a main robot on each newly created vertex.
We embed the new vertices produced by the edge subdivisions horizontally to the left and right, respectively, from the original positions of $r_1$ and $r_2$, respectively. Now we create a vertical path connecting all $r_1$-vertices of all instances and a vertical path connecting all $r_2$-vertices of all instances. We place main robots on all vertices of those two paths. For an illustration see \cref{fig:nopk2}. Now, for each instance, we create one additional path with main robots on each vertex as follows. Let $(x,y)$ be the coordinates of the original grid point of the $r_1$-vertex of an instance (before the edge subdivisions). Then we start a path at grid point $(x-1,y-1)$ of length $k'<c-1$ horizontally to the left, where we specify $k'$ later. Furthermore, we connect the vertex on grid point $(x-1,y-1)$ to the one on grid point $(x-1,y)$.

The intuition here is that instead of having two main robots meet from each side of an instance constructed by \cref{constr:gridnph}, we create a connection between the main robots on the left side of the instances and the main robots on the right side of the instances by moving $k'$ robots into one of the instances. This is why the obnoxious robots in the variable gadgets are removed. The main robots will play the role of the obnoxious robots instead. This also means that $k'$ essentially needs to equal the number of steps made by main robots in the reduction of \cref{thm:nphardness}. Since the other paths on both sides of the instances have length $c+1$, main robots on those paths cannot be moved, since broken connections cannot be repaired with less than $c$ moves.

To finish the construction, we set $k'=12\cdot |X|+6\cdot |Y|+s+5$, where $|X|$ and $|Y|$ are the numbers of variables and clauses, respectively, and $s=90\cdot(|X|+|Y|)-4$. We set $c=(k')^2+(s+2)\cdot |Y|$. 

\begin{figure}[t]
\begin{center}
\begin{tikzpicture}[line width=1pt,scale=.56,yscale=1]

\foreach \i in {-12,...,3}
{
    \draw[line width=.5pt,color=lightgray,dotted] (-8.8,\i) -- (14.8,\i);
}

\foreach \i in {-8,...,14}
{
    \draw[line width=.5pt,color=lightgray,dotted] (\i,-12.8) -- (\i,3.8);
}

\draw[fill=lightgray!25!white] (0, -1) rectangle (6, 3);
\node (E) at (3,1) {\small Instance 1};

\draw[fill=lightgray!25!white] (0, -2) rectangle (6, -6);
\node (E) at (3,-4) {\small Instance 2};

\node (E) at (3, -6.75) {$\vdots$};

\draw[fill=lightgray!25!white] (0, -8) rectangle (6, -12);
\node (E) at (3,-10) {\small Instance $t$};

\draw[line width=2pt, blue] (0,1) -- (-8,1);
\draw[line width=2pt, blue] (0,-4) -- (-8,-4);
\draw[line width=2pt, blue] (0,-10) -- (-8,-10);

\draw[line width=4pt, blue] (-1,1) -- (-1,0) -- (-6,0);
\draw[line width=4pt, blue] (-1,-4) -- (-1,-5) -- (-6,-5);
\draw[line width=4pt, blue] (-1,-10) -- (-1,-11) -- (-6,-11);

\draw[line width=2pt, blue] (-8,1) -- (-8,-10);

\draw[line width=2pt, blue] (6,1) -- (14,1);
\draw[line width=2pt, blue] (6,-4) -- (14,-4);
\draw[line width=2pt, blue] (6,-10) -- (14,-10);

\draw[line width=2pt, blue] (14,1) -- (14,-10);

\path[draw,decorate,decoration=brace] (-7.8,1.2) -- (-.2,1.2)
node[midway,yshift=1em]{$c+1$};

\path[draw,decorate,decoration=brace] (-1.2,-.3) -- (-5.8,-.3)
node[midway,yshift=-1em]{$k'$};

\end{tikzpicture}
    \end{center}
    \caption{Illustration the OR-cross-composition from \cref{thm:nopkgrid}. The grid points on blue paths have main robots on them. The thick blue paths have the main robots on them that would move through the corresponding instance.}\label{fig:nopk2}
\end{figure}
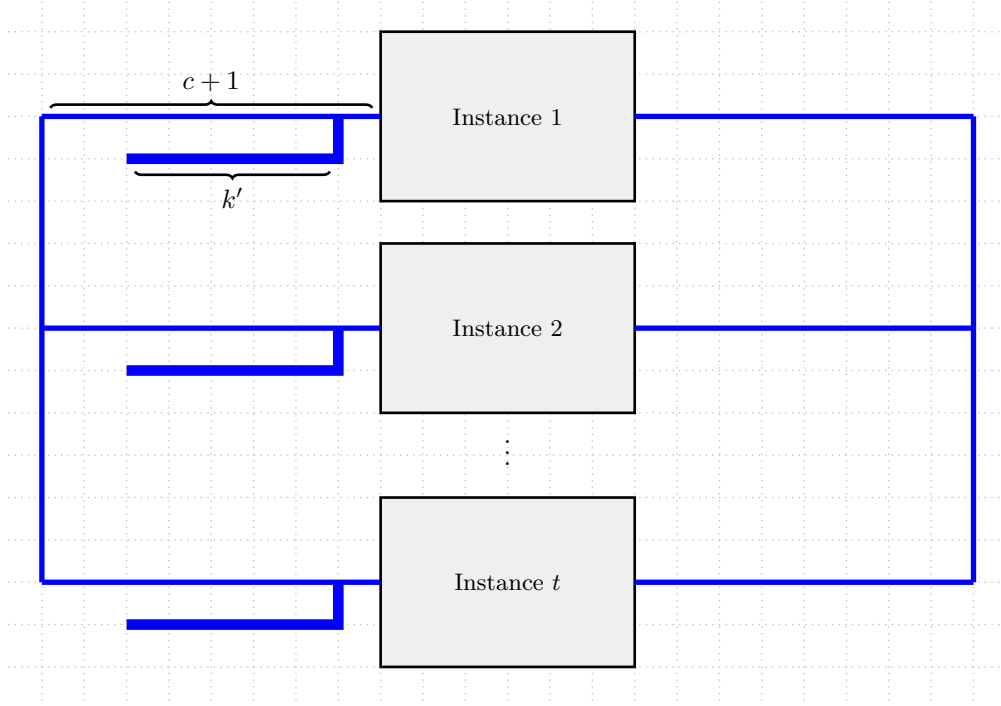

Note that for the main robots on the left and the right side of the instances to be connected, a connection through one of the instances has to be established. The formal correctness is analogous to \cref{thm:nphardness} with a few minor differences:
If one of the instances is a yes-instance, then we can move the main robots on the path of length $k'$ next to the instance through the instance using $(k')^2$ moves. In the variable gadget, the main robots assume the role of the obnoxious robots in the original reduction, that is, they leave the same vertices in the variable gadgets unoccupied as the obnoxious robots, such that the obnoxious robots on the paths from clause gadgets to variable gadgets can move if the variable is set to a truth value that satisfied the corresponding literal of the clause. Those are $(s+2)\cdot |Y|$ moves in total.

For the other direction, also analogous arguments hold with the following modification. We argue that the main robots cannot take ``short-cuts'' through paths connecting variable gadgets and clause gadgets since moving all obnoxious robots out of the way and moving main robots into the path costs at least $2\cdot (s+2)^2$ moves. For large enough $|X|$ and $|Y|$ we have that $2\cdot (s+2)^2>c$, and hence this is not possible. It follows that we can use the same arguments for the correctness as in the proof of \cref{thm:nphardness}. 
\end{proof}

Finally, we prove \cref{thm:nopkgrid2} (restated below), that is, we presumably cannot improve the kernel size from \cref{thm:kernel2}.
For this kernelization lower bound, we use the weak cross-compositions framework~\cite{hermelin2012weak,Fom+19} to refute the existence of a polynomial kernel of a certain size for a parameterized problem under the assumption that \NNoKernelAssume. 
A \emph{weak cross-composition} of a problem~$L\subseteq \{0,1\}^*$ into a
parameterized problem~$P$ (with respect to a polynomial equivalence
relation~$R$ on the instances of~\(L\)) is an algorithm that takes
$t$ $R$-equivalent instances~$x_1,\ldots,x_t$ of~$L$ and
constructs in time polynomial in $\sum_{i=1}^t |x_i|$ an instance
$(y,k)$ of~\(P\) such that
\begin{enumerate}
\item $k\le(\max_{i\in [t]}|x_i|)^{O(1)}\cdot t^{\frac{1}{d}}$ for some $d\in\mathbb{N}$, and 
\item $(y,k)$ is a \yes-instance of $P$ if and only if there is an $i\in [t]$ s.t.\ $x_{i}$ is a \yes-instance of~$L$. 
\end{enumerate}
If an \NP-hard problem~\(L\) weakly cross-composes into a parameterized
problem~$P$, then~$P$ does not admit a polynomial kernel of size $O(k^{d-\varepsilon})$ for any $\varepsilon>0$, unless \NoKernelAssume~\cite{hermelin2012weak,Fom+19}.

\nopkgridd*

\begin{proof}
We provide a weak cross-composition from \textsc{Linked Planar 3-SAT}~\cite{Pilz19} based on \cref{constr:gridnph} that is used in the proof of \cref{thm:nphardness}. 
We say that two instances of \textsc{Linked Planar 3-SAT} are equivalent if the two formulas have the same number of variables and the same number of clauses. Assume we are given $t$ equivalent instances $x_1,x_2,\ldots,x_t$ of \textsc{Linked Planar 3-SAT}. Furthermore, assume that $\sqrt{t}=t'$ for some $t'\in\mathbb{N}$ (if not, we can copy some of the instances sufficiently many times).
We apply \cref{constr:gridnph} to each instance of \textsc{Linked Planar 3-SAT}. We remove the main robots from vertices $r_1$ and $r_2$ of all instances. Then in each instance, the we subdivide the edges incident with $r_1$ and $r_2$, respectively, $s+4=90\cdot(|X|+|Y|)$ times, where $|X|$ and $|Y|$ are the number of variables and clauses, respectively. 
We use the embedding for the instances described in \cref{lem:grid} and embed the new vertices produced by the edge subdivisions horizontally to the left and right, respectively, from the original positions of $r_1$ and $r_2$, respectively.

Let $x_{\max}=\max_{i\in[t]} |x_t|$. We arrange the $t$ instances in a $(\sqrt{t}\times \sqrt{t})$-``super grid'' where each cell has size $(x_{\max}+4)\times (x_{\max}+4)$ and place one of the instances into each super grid cell so that there are at least 2 empty rows / columns of grid cells between each instance and the boundaries of the super grid cells. For an illustration see \cref{fig:nopk3}. Now we place two verices $r^\star_1$ and $r^\star_2$ at the bottom-left and top-right corner of the super grid, respectively, and place main robots on them.

We create paths from $r^\star_1$ horizontally (to the right) and vertically (to the top) of length $\sqrt{t}\cdot(x_{\max}+4)-1$, that is, they reach just short of the bottom-right and top-left corner of the super grid. From the vertical path, we create horizontal paths (to the left) of length $\sqrt{t}\cdot(x_{\max}+4)-1$ along the bottom (plus one grid cell up) of each super grid cell. We connect this path (via vertical paths up) to the $r_1$-vertex of each instance in each super grid cell that it traverses.
Now we do an analogous construction from $r^\star_2$. We create paths from $r^\star_2$ horizontally (to the left) and vertically (to the bottom) of length $\sqrt{t}\cdot(x_{\max}+4)-1$, that is, they also reach just short of the bottom-right and top-left corner of the super grid and do not touch the paths starting at $r^\star_1$. From the vertical path, we create horizontal paths (to the left) of length $\sqrt{t}\cdot(x_{\max}+4)-1$ along the top (minus one grid cell up) of each super grid cell. We connect this path (via vertical paths down) to the $r_2$-vertex of each instance in each super grid cell that it traverses. The path just described are depicted in red in \cref{fig:nopk3}.

We set $c=2\cdot \sqrt{t}\cdot(x_{\max}+4)+2\cdot (s+4) + c'$, where $c'$ is the allowed movement amout of an instance $x_i$ with $i\in[t]$ (note that since the instances are equivalent, they all have the same allowed movement amount). This finishes the construction.

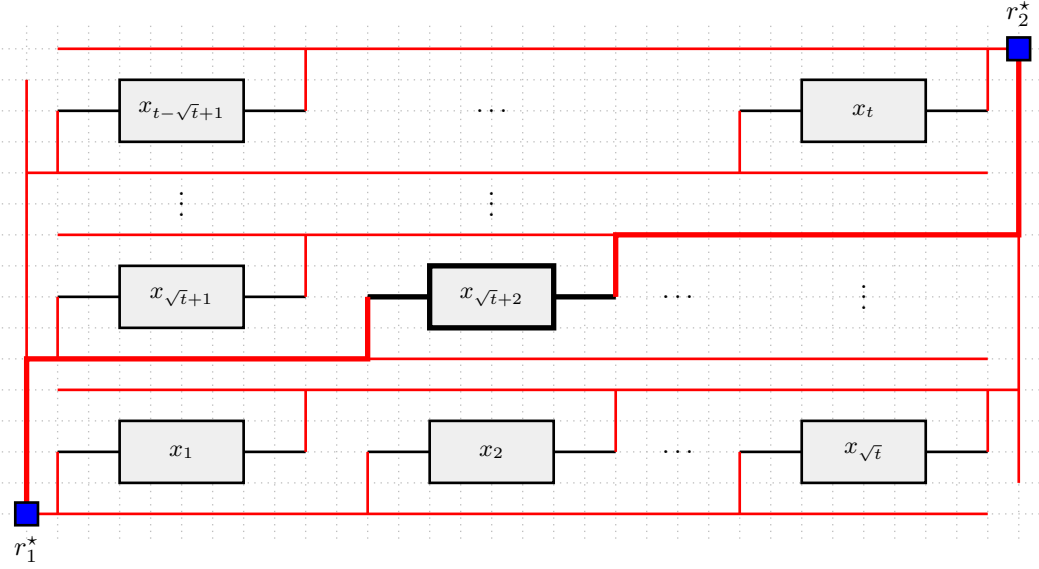
\begin{figure}[t]
\begin{center}
\begin{tikzpicture}[line width=1pt,scale=.41]

\foreach \i in {-1,...,14}
{
    \draw[line width=.5pt,color=lightgray,dotted] (-3.8,\i) -- (29.8,\i);
}

\foreach \i in {-3,...,29}
{
    \draw[line width=.5pt,color=lightgray,dotted] (\i,-1.8) -- (\i,14.8);
}

\draw (0,1) -- (-2,1);
\draw[fill=lightgray!25!white] (0, 0) rectangle (4, 2);
\node (E) at (2,1) {\small $x_1$};
\draw (4,1) -- (6,1);

\draw (8,1) -- (10,1);
\draw[fill=lightgray!25!white] (10, 0) rectangle (14, 2);
\node (E) at (12,1) {\small $x_2$};
\draw (14,1) -- (16,1);

\node (E) at (18, 1) {$\ldots$};

\draw (20,1) -- (22,1);
\draw[fill=lightgray!25!white] (22, 0) rectangle (26, 2);
\node (E) at (24,1) {\small $x_{\sqrt{t}}$};
\draw (26,1) -- (28,1);

\draw (0,6) -- (-2,6);
\draw[fill=lightgray!25!white] (0, 5) rectangle (4, 7);
\node (E) at (2,6) {\small $x_{\sqrt{t}+1}$};
\draw (4,6) -- (6,6);

\draw (8,6) -- (10,6);
\draw[fill=lightgray!25!white] (10, 5) rectangle (14, 7);
\node (E) at (12,6) {\small $x_{\sqrt{t}+2}$};
\draw (14,6) -- (16,6);

\node (E) at (18, 6) {$\ldots$};

\node (E) at (24, 6.25) {$\vdots$};

\node (E) at (2, 9.25) {$\vdots$};

\draw (0,12) -- (-2,12);
\draw[fill=lightgray!25!white] (0, 11) rectangle (4, 13);
\node (E) at (2,12) {\small $x_{t-\sqrt{t}+1}$};
\draw (4,12) -- (6,12);

\node (E) at (12, 12) {$\ldots$};

\node (E) at (12, 9.25) {$\vdots$};

\draw (20,12) -- (22,12);
\draw[fill=lightgray!25!white] (22, 11) rectangle (26, 13);
\node (E) at (24,12) {\small $x_{t}$};
\draw (26,12) -- (28,12);

\draw[red] (-3,-1) -- (-3,13);

\draw[red] (-3,-1) -- (28,-1);
\draw[red] (-2,-1) -- (-2,1);
\draw[red] (8,-1) -- (8,1);
\draw[red] (20,-1) -- (20,1);

\draw[red] (-3,4) -- (28,4);
\draw[red] (-2,4) -- (-2,6);
\draw[red] (8,4) -- (8,6);

\draw[red] (-3,10) -- (28,10);
\draw[red] (-2,10) -- (-2,12);
\draw[red] (20,10) -- (20,12);

\draw[red] (29,0) -- (29,14);

\draw[red] (-2,3) -- (29,3);
\draw[red] (6,3) -- (6,1);
\draw[red] (16,3) -- (16,1);
\draw[red] (28,3) -- (28,1);

\draw[red] (-2,8) -- (29,8);
\draw[red] (6,8) -- (6,6);
\draw[red] (16,8) -- (16,6);

\draw[red] (-2,14) -- (29,14);
\draw[red] (6,14) -- (6,12);
\draw[red] (28,14) -- (28,12);

\draw[line width=2pt]  (8,6) -- (10,6);
\draw[line width=2pt] (10, 5) rectangle (14, 7);
\draw[line width=2pt]  (14,6) -- (16,6);
\draw[red,line width=2pt] (-3,-1) -- (-3,4) -- (8,4) -- (8,6);
\draw[red,line width=2pt] (16,6) -- (16,8) -- (29,8) -- (29,14);

\node[vert2,fill=blue,label=below:$r^\star_1$] (R1) at (-3,-1) {};
\node[vert2,fill=blue,label=above:$r^\star_2$] (R2) at (29,14) {};

\end{tikzpicture}
    \end{center}
    \caption{Illustration of the weak cross-composition from \cref{thm:nopkgrid2}. The drawing is not up to scale but rather should give the rough idea on how the embedding looks like. Black boxes correspond to the instances and red lines to the connecting paths. Blue filled square vertices have main robots on them. The thick lines correspond to the paths visited by the main robots when instance $x_{\sqrt{t}+2}$ is selected.}\label{fig:nopk3}
\end{figure}

Note that for the main robots on $r^\star_1$ and $r^\star_2$, respectively, to meet, they have to traverse at least one of the instances. Assume w.l.o.g.\ that the two robots meet in an instance (if they do not, we can move the meeting point to $r_1$ or $r_2$ with a respective neighbor of one of the instances). In order to reach said instance, the have to travel at least $2\cdot \sqrt{t}\cdot(x_{\max}+4)$ along red paths in \cref{fig:nopk3}. Furthermore, they have to traverse the edge-subdivision after $r_1$ and $r_2$ of said instance, which takes another $2\cdot (s+4)$ steps. There are now $c'$ steps remaining to traverse at least one instance.
The edge subdivisions ensure that it is too expensive to move obnoxious robots between the instances. Intuitively, the allowed movement amount forces $r_1$ and $r_2$ to meet through one of the instances. The formal correctness is analogous to \cref{thm:nphardness}.
\end{proof}

\section{Collision Free Movement on Planar Graphs}\label{sec:planar}

In this section, we prove \cref{thm:mso} (restated below), that is, we show that \MoveToPi is fixed-parameter tractable when parameterized by $k+c$ if the input graph is planar.

\mso*

To prove \cref{thm:mso}, we show that \MoveToPi is expressable in monadic second-order logic (MSO). This allows us to employ Courcelle's famous theorem~\cite{arnborg1991easy,courcelle1990monadic,courcelle2012graph} to obtain the result. 

We begin by recalling some basic notation and terminology for relational structures.
A \emph{relational vocabulary} is a finite set $\tau$ of relation symbols, each of which is associated with a natural number, known as its \emph{arity}.  Given any relational vocabulary $\tau$, a \emph{$\tau$-structure} is a pair $\mathcal{A} = (A, \{R^{\mathcal{A}}\mid R \in \tau\})$; $A$ is said to be the \emph{universe} of $\mathcal{A}$ while, for each $R \in \tau$, the \emph{interpretation} $R^{\mathcal{A}}$ of $R$ in $\mathcal{A}$ is a subset of $A^r$, where $r$ is the arity of $R$. Here, we are interested only in relational structures with finite universes, and where the maximum arity of any relation is two.  

For any vocabulary $\tau$, the set of \emph{first-order formulas} is built up from a countably infinite set of individual variables $x_1,x_2,\ldots$, the relation symbols $R \in \tau$, the connectives $\wedge, \vee, \neg$ and the quantifiers $\forall x, \exists x$ ranging over elements of the universe of the structure (for notational convenience we will also use standard shorthand such as $\Rightarrow$, $\Leftrightarrow$, $=$, and $\in$). \emph{Monadic second-order logic} (MSO) additionally allows quantification over subsets of the universe via unary relation variables (which we will call \emph{set variables}). Given an MSO-formula $\phi$, an individual variable $x$ (respectively a unary relation variable $X$) appearing in $\phi$ is said to be a \emph{free variable} if $x$ (respectively $X$) is not in the scope of a quantifier $\exists x$ or $\forall x$ (respectively $\exists X$ or $\forall X$).  We write $\phi(X_1,\dots,X_n,x_1,\ldots,x_m)$ for a formula $\psi$ with with free relation variables $X_1,\ldots,X_n$ and individual variables $x_1,\ldots,x_m$.  Given subsets $A_1,\dots,A_n \subseteq A$ (formally these define interpretations of unary relation variables over $A$) and elements $a_1,\dots,a_m \in A$, we write $\mathcal{A} \models \phi(A_1,\ldots,A_n,a_1,\ldots,a_m)$ to mean that $\mathcal{A}$ satisfies $\phi$ if the variables $X_1,\ldots,X_n,x_1,\ldots,x_m$ are interpreted as $A_1,\ldots,A_n,a_1,\ldots,a_m$ respectively.

A \emph{tree decomposition} of a $\tau$-structure $\mathcal{A}$ with universe $A$ is a pair $(T,\mathcal{B})$, where $T = (V_T,E_T)$ is a tree and $\mathcal{B} = (B_v)_{v \in V_T}$ is a collection of subsets of $A$ such that:
\begin{enumerate}
    \item for all $a \in A$, the set $\{v \in V_T: a \in B_v\}$ is non-empty and induces a connected subtree in $T$, and
    \item for every relation symbol $R \in \tau$ and every tuple $(a_1,\dots,a_r) \in R^{\mathcal{A}}$, there exists $B_v \in \mathcal{B}$ such that $a_1,\dots,a_r \in B_V$.
\end{enumerate}
As for graphs, the \emph{width} of the tree decomposition $(T,\mathcal{B})$ is $\max_{B_v \in \mathcal{B}} |B_v| - 1$, and the \emph{treewidth} of~$\mathcal{A}$ (denoted $\tw(\mathcal{A})$) is the minimum width over all tree decompositions of $\mathcal{A}$.


We now give the formal definition of the \textsc{MSO Model Checking} problem.

\problemdef{\textsc{MSO Model Checking}}{A relational structure $\mathcal{A}$ and an MSO-formula $\phi$.}{Are there subsets $A_1,\dots,A_n \subseteq A$ and elements $a_1,\dots,a_m \in A$ such that $\mathcal{A} \models \phi(A_1,\ldots,A_n,a_1,\ldots,a_m)$?}

Courcelle's  theorem~\cite{arnborg1991easy,courcelle1990monadic,courcelle2012graph} implies the following.

\begin{theorem}[\cite{arnborg1991easy,courcelle1990monadic,courcelle2012graph}]\label{thm:courcelle}
\textsc{MSO Model Checking} is fixed-parameter tractable when parameterized by $\tw(\mathcal{A})+|\phi|$.
\end{theorem}

We begin by defining the relational structure $\mathcal{A}$ that will encode our instance of \MoveToPi. Assume we are given an instance $(G=(V,E),R,O,c)$ of \MoveToPi. The universe $A$ of $\mathcal{A}$ consists of the following elements:
\begin{itemize}
    \item The set $V$ of \emph{vertices}.
    \item The set $E$ of \emph{edges}.
    \item The set $[c]$ of \emph{steps}.
    \item The set $S_v=V\times [c]$ of \emph{vertex-step pairs}.
    \item The set $S_e=V\times [c]$ of \emph{edge-step pairs}.
\end{itemize}
The structure $\mathcal{A}$ has eight relation symbols (two unary and six binary relations), with the following interpretations:
\begin{itemize}
    \item The unary relation $\main$ on $V$, where $\main(v) \Leftrightarrow v\in R$.
    \item The unary relation $\robot$ on $V$, where $\robot(v) \Leftrightarrow v\in R\cup O$.
    \item The binary relation $\inc$ on $V\times E$, where $\inc(v,e) \Leftrightarrow v\in e$.
    \item The binary relation $\vertex$ on $S_v\times V$, where $\vertex((v,i),w)\Leftrightarrow v=w$.
    \item The binary relation $\edge$ on $S_e\times E$, where $\edge((e,i),f)\Leftrightarrow e=f$.
    \item The binary relation $\step_v$ on $S_v\times [c]$, where $\step_v((v,i),j)\Leftrightarrow i=j$.
    \item The binary relation $\step_e$ on $S_e\times [c]$, where $\step_e((e,i),j)\Leftrightarrow i=j$.
    \item The binary relation $\nextstep$ on $[c]\times [c]$, where $\nextstep(i,j)\Leftrightarrow i=j+1$.
\end{itemize}

We now bound the treewidth of $\mathcal{A}$. Let $(T,\mathcal{B})$ be a tree decomposition for $G$ of width $\tw(G)$. We construct a tree decomposition $(T,\mathcal{B}')$ for $\mathcal{A}$, indexed by the same tree $T = (V_T,E_T)$, as follows.  Fix a vertex $v \in V_T$, and let $B_v$ be the corresponding element of $\mathcal{B}$.  We define the corresponding element of $\mathcal{B}'$ to be 
\[
    B_v' := B_v \cup \{e=\{u,w\} \in E: v,w \in B_v\} \cup \{s=(u,i)\in S_v: u\in B_v\} \cup \{s=(\{u,w\},i)\in S_e: u,w\in B_v\}  \cup [c].
\]
It is straightforward to verify that, with this definition, $(T,\mathcal{B}')$ is indeed a tree decomposition for~$\mathcal{A}$. Furthermore, it is immediate that $|B_v'| \le |B_v| + \binom{|B_v|}{2} + c\cdot \binom{|B_v|}{2} + c\cdot |B_v| + c$. Hence, we have the following.
\begin{lemma}\label{lem:msotw}
It holds that $\tw(\mathcal{A})\in O(c\cdot \tw(G)^2)$.    
\end{lemma}

We now proceed to define the formula $\phi$. The main idea (intuitively) is as follows.
\begin{itemize}
    \item For each main robot, we try to find a movement of length $c$.
    \item If an obnoxious robot is encountered in a movement, then we try to find a (non-trivial) movement of length $c$ for the obnoxious robot.
    \item We know that in total, we can only have $c$ non-trivial movements in the movement plan.
    \item The movement plan should be collision-free and the end-configuration of the main robots should induce a subgraph with property $\Pi$.
    \item The length of the moment plan should be at most $c$.
\end{itemize}

\newcommand{\movement}{\operatorname{movement}}
\newcommand{\collisionfree}{\operatorname{collision-free}}
\newcommand{\move}{\operatorname{move}}
\newcommand{\start}{\operatorname{start}}
\newcommand{\eend}{\operatorname{end}}
\newcommand{\visit}{\operatorname{visit}}

We introduce several subformulas:
\begin{itemize}
    \item $\movement(X)$ has as free variable $X$ and checks whether these are exactly $c$ elements from $S_v$ that correspond to a movement.
\begin{align*}
  \movement(X)= \ & \exists x_1,\ldots,x_c\in S_v \Big( \bigwedge_{i\in[c]} (x_i\in X\wedge \step_v(x,i))\wedge\forall x\in X\bigvee_{i\in[c]} (x=x_i) \wedge \\
& \bigwedge_{i\in[c-1]}(\exists v\in V(\vertex(x_i,v)\wedge\vertex(x_{i+1},v))\vee \\
& \exists e\in E \ \exists v,w\in V(\inc(v,e)\wedge\inc(w,e)\wedge \vertex(x_i,v)\wedge\vertex(x_{i+1},w)\Big)  
\end{align*}
    Note that the size of $\movement(X)$ is in $O(c)$. The formula checks whether there are $c$ elements in $S_v$ with $c$ different steps (note that this implies that the elements are different) and that each element in $X$ is equal to one of these $c$ elements from $S_v$. This implies that $X$ contains exactly~$c$ elements from $S_v$. Then the formula checks for every pair of vertex-step pairs in $X$ with consecutive time steps whether they have the same vertex or adjacent vertices.
    \item $\collisionfree(X_1,X_2)$ has as free variables $X_1$ and $X_2$, which are expected to be movements, and checks whether they are collision-free.
    \begin{align*}
    \collisionfree(X_1,X_2)= \  
    &\forall x_1\in X_1 \ \forall x_2\in X_2 \ \forall i\in[c] \ \exists v,w\in V\Big((\step_v(x_1,i)\wedge \step(x_2,i)) \Rightarrow\\
    &(v\neq w\wedge\vertex(x_1,v)\wedge\vertex(x_2,w)\Big) \wedge\\\
    &\forall x_1,x_2\in X_1 \ \forall x_3,x_4\in X_2 \ \forall i,j\in[c] \ \forall v,w\in V\Big((\nextstep(i,j)\wedge\\
    & \step_v(x_1,i)\wedge \step(x_2,j)\wedge \step_v(x_3,i)\wedge \step(x_4,j))\Rightarrow\\
    & \neg(\vertex(x_1,v)\wedge\vertex(x_2,w)\wedge\vertex(x_3,w)\wedge\vertex(x_4,v))\Big)
    \end{align*}
    Note that the size of $\collisionfree(X_1,X_2)$ is constant.
    \item $\move(X,x)$ has as free variables $X$ and $x$ where $X$ is expected to be a movement and $x$ is expected to be an element from $S_e$. It checks whether the edge-step pair $x=(e,i)$ is a move in the movement, that is, whether at step $i$ there is a move along edge $e$ in the movement $X$.
    \begin{align*}
        \move(X,x) = \  & \exists x_1,x_2\in X \ \exists v,w\in V \ \exists e\in E \ \exists i,j\in [c]\Big(\nextstep(i,j)\wedge \step_v(x_1,i)\wedge\step_v(x_2,j)\wedge\\
        &v\neq w\wedge \vertex(x_1,v)\wedge\vertex(x_2,w)\wedge\inc(v,e)\wedge\inc(w,e)\wedge\edge(x,e)\wedge\step_e(x,i)\Big)
    \end{align*}
    Note that the size of $\move(X,x)$ is constant.
    \item $\start(X,x)$, $\eend(X,x)$, and $\visit(X,x)$ have as free variables $X$ and $x$, respectively, where $X$ is expected to be a movement and $x$ is expected to be an element from $V$. $\start(X,x)$ checks whether the movement $X$ starts at vertex $x$. $\eend(X,x)$ checks whether the movement $X$ ends at vertex $x$. $\visit(X,x)$ checks whether the movement $X$ visits vertex $x$ at some step.
    \begin{align*}
        &\start(X,x)=\exists x_1\in X (\vertex(x_1,x)\wedge\step_v(x,1))\\
        &\eend(X,x)=\exists x_1\in X (\vertex(x_1,x)\wedge\step_v(x,c))\\
        &\visit(X,x)=\exists x_1 \in X (\vertex(x_1,x))
    \end{align*}
    Note that the sizes of $\start(X,x)$, $\eend(X,x)$, and $\visit(X,x)$ are each constant.
\end{itemize}
Now we are ready to give the formula $\phi$. Denote with $c^\star=\min(k+\ell,c)$.
\begin{align*}
    \phi(X_1,\ldots,X_{c^\star},x_1,\ldots,x_c)=&\bigwedge_{i\in[c^\star]}\movement(X_i)\wedge\bigwedge_{i,j\in[c^\star]}((i\neq j)\Rightarrow \collisionfree(X_i,X_j))\wedge\\
    &\forall v\in V \Big((\robot(v)\wedge (\bigvee_{i\in[c^\star]}\visit(X_i,v)))\Rightarrow(\bigvee_{i\in[c^\star]}\start(X_i,v))\Big)\wedge\\
    &\forall x\in S_e \Big((\bigvee_{i\in[c^\star]} \move(X_i,x))\Rightarrow(\bigvee_{i\in[c]}(x=x_i))\Big)\wedge\\
    &\exists X\subseteq V \Big((\forall v\in V ((\main(v)\wedge\bigwedge_{i\in[c^\star]}\neg\start(X_i,v))\vee\\
    &\exists w\in V (\main(w)\wedge(\bigvee_{i\in[c^\star]}(\start(X_i,w)\wedge\eend(X_i,v)))\Leftrightarrow (v\in X)))\wedge \Pi(X) \Big)
\end{align*}
Note that the size of the formula $\phi$ is in $O(c^2+|\Pi|)$. Since the graph property $\Pi$ is recursively enumerable, we can conclude $|\Pi(X)|=f(k)$ for $|X|=k$ and some computable function $f$. Hence, we have the following.
\begin{lemma}\label{lem:msosize}
    It holds that $|\phi|\in O(c^2)+f(k)$, for some computable function $f$.
\end{lemma}

It remains to show that the formula is correct.

\begin{lemma}\label{lem:msocorrect}
    The formula $\phi$ is satisfiable if and only if $(G=(V,E),R,O,c)$ is a yes-instance of \MoveToPi.
\end{lemma}

\begin{proof}
    First, assume that $\phi$ is satisfiable. We show that then $(G=(V,E),R,O,c)$ is a yes-instance of \MoveToPi. Let $X_1,\ldots,X_{c^\star},x_1,\ldots,x_c$ denote a satisfying assignment for $\phi$. We construct a movement plan $\mathcal{M}$ as follows. We have that the sets $X_i$ with $i\in[c^\star]$ satisfy $\movement(X_i)$. It follows that $X_i=\{(v_1,1),(v_2,2),\ldots,(v_c,c)\}$ (see the first line of the definition of $\movement(X)$). We claim that $M_i=(v_1,v_2,\ldots,v_c)$ is a movement. The second and third line of the definition of $\movement(X)$ give us exactly that for all $v_j,v_{j+1}$ with $j\in[c-1]$ we have that $v_j=v_{j+1}$ or $\{v_j,v_{j+1}\}\in E$. Hence, we can conclude that $M_i$ is a movement. Now for each $M_i$ that is a movement for a robot $r\in R\cup O$, we add $M_i$ to the movement plan $\mathcal{M}$. For each $r\in R\cup O$ such that none of the $M_i$ with $i\in[c^\star]$ is a movement for $r$, we add a trivial movement to the movement plan $\mathcal{M}$. Now we show the following, from which clearly follows that $(G=(V,E),R,O,c)$ is a yes-instance of \MoveToPi.
    \begin{enumerate}
        \item For each $r\in R\cup O$ there is at most one $i\in[c^\star]$ such that $M_i$ is a movement for $i$, that is, the movement plan $\mathcal{M}$ is well-defined.
        \item The movement plan $\mathcal{M}$ is collision-free.
        \item $\len(\mathcal{M})\le c$ and $G[R^\star(\mathcal{M})]\in\Pi$.
    \end{enumerate}
    
Consider the first statement and assume for contradiction that there are two movements $M_i$ and $M_j$ with $i\neq j$ for some $r\in R\cup O$. Note that in $\phi$, the subformula $\collisionfree(X_i,X_j)$ is checked and evaluates to true, since otherwise $\phi$ would evaluate to false. In $\collisionfree(X_i,X_j)$ it is, in particular, checked, that $(v_1,1)\in X_i$ and $(v'_1,1)\in X_j$ implies $v_1\neq v'_1$. It follows that $M_i$ and $M_j$ cannot both be a movement for the same $r\in R\cup O$.

Consider the second statement and assume for contradiction that there are two movements $M_i$ and $M_j$ with $i\neq j$ for robots $r,r'\in R\cup O$, respectively, that are not collision-free. First, consider that both $M_i$ and $M_j$ are non-trivial. Then in $\phi$, the subformula $\collisionfree(X_i,X_j)$ is checked and evaluates to true, since otherwise $\phi$ would evaluate to false. However, if $M_i$ and $M_j$ then there exists some vertex $v$ and some step $t$ such that both $r$ and $r'$ occupy $v$ in step $t$ or there exit an edge $\{v,w\}\in E$ and some step $t$ such that $r$ occupies $v$ in step $t$ and $w$ in step $t+1$ and $r'$ occupies $w$ in step $t$ and $v$ in step $t+1$. In this case we can verify that $\collisionfree(X_i,X_j)$ evaluates to false: In the former case the first two lines in the definition of $\collisionfree(X_1,X_2)$ evaluate to false, and in the latter case, the last three lines in the definition of $\collisionfree(X_1,X_2)$ evaluate to false. In each case we reach a contradiction. 
Now consider the case that w.l.o.g.\ $M_i$ is a non-trivial movement and $M_j$ is a trivial movement. If $M_i$ and $M_j$ are not collision-free, then there must be a step $t$ such that robot $r$ occupies $r'$ in step $t$. However, then $\robot(r')$ and $\visit(X_i,r')$ both evaluate to true. It follows that $\bigvee_{j\in[c^\star]}\start(X_j,r')$ must also evaluate to true. It follows that $M_j$ corresponds to some $X_j$ and $\collisionfree(X_i,X_j)$ evaluates to true. By the same arguments as earlier we can conclude that $M_i$ and $M_j$ are collision-free.

Finally, consider the third statement. Let $M_i$ be a non-trivial movement for some $r\in R\cup O$. First, it is straightforward to verify that if robot $r$ moves along edge $\{v,w\}$ in step $t$, then $\move(X_i,(\{v,w\},t))$ evaluates to true. Since the movement plan $\mathcal{M}$ is collision-free according to the previous arguments, we have that for a fixed $(\{v,w\},t)$, the subformula $\move(X_i,(\{v,w\},t))$ evaluates to true for at most one $X_i$. In the third line of the definition of $\phi$ it is checked that whenever $\move(X_i,(\{v,w\},t))$ evaluates to true, then $(\{v,w\},t)=x_j$ for some $j\in[c]$. It follows that the set $\{x_1,x_2,\ldots,x_c\}$ contains all moves made by robots according to movement plan $\mathcal{M}$. It follows that there are at most $c$ different moves, and hence $\len(\mathcal{M})\le c$.
Next, consider the last two lines of the definition of $\phi$. We show that the set $X$ here corresponds to $R^\star(\mathcal{M})$: The formula ensures that a vertex $v$ is contained in $X$ if and only if $v$ is a main robot and not the start of any $X_i$ (in this case $v$ has a trivial movement), or there is a main robot $w$ and an $X_i$ starting with $w$ that ends in $v$, that is, $X_i$ corresponds to a movement for a main robot $w$ that ends in $v$. Finally $\phi$ checks whether $\Pi(X)$ evaluates to true. It follows that $G[R^\star(\mathcal{M})]\in\Pi$.

Now assume that $(G=(V,E),R,O,c)$ is a yes-instance of \MoveToPi. Then there is a collision-free movement plan $\mathcal{M}$ that contains at most $c^\star$ non-trivial movements which make at most $c$ moves. Let $X_1,X_2,\ldots,X_{c^\star}$ correspond to the non-trivial movements (if there are strictly less than $c^\star$ non-trivial movements, we can add sets $X$ that only contain a vertex that is never visited (we can easily ensure that such a vertex always exists). Let $x_1,x_2,\ldots,x_c$ correspond to the moves. If there are strictly less than $c$ moves, then we can duplicate arbitrary moves.
By analogous arguments to the first part of this proof, we can argue that $\phi(X_1,X_2,\ldots,X_{c^\star},x_1,x_2,\ldots,x_c)$ evaluates to true. 
\end{proof}

Now we have all ingredients to prove \cref{thm:mso}.
\begin{proof}[Proof of \cref{thm:mso}]
Given an instance $(G=(V,E),R,O,c)$ of \MoveToPi, we first apply \cref{rr:1}. By \cref{lem:rr1safe} this creates an equivalent instance $(G'=(V',E'),R,O',c)$. We can observe that $\diam(G')\in O(k\cdot c)$ (where if $G'$ is disconnected, we assume that $\diam(G')$ is the maximum diameter of the connected components of $G'$). This follows from the fact that after \cref{rr:1} is applied, every vertex in $G'$ has distance at most $c$ from some $r\in R$, and for every $r,r'\in R$ such that $r$ and $r'$ are in the same connected component of $G'$, we have that $\dist(r,r')\in O(k\cdot c)$. Since $G'$ is planar, it is well-known that this implies that $\tw(G')\in O(k\cdot c)$.
\cref{thm:mso} now follows from \cref{lem:msocorrect,lem:msotw,lem:msosize}, and \cref{thm:courcelle}.    
\end{proof}

We remark that the described algorithm does not exploit planarity. It is only used to upper-bound the treewidth with the diameter of the input graph. Hence, the algorithm also implies that \MoveToPi is fixed-parameter tractable when parameterized by $k+c$ and the treewidth of the input graph combined.

\section{Collision Free Movement on Unit Disk Graphs}\label{sec:udgs}

In this section, we consider the canonical optimization version of the problem \MoveToConnectedShort, where the objective is to minimize the length of the movement plan.
We prove \cref{thm:approx} (restated below), that is, we give a fixed-parameter approximation algorithm for  \MoveToConnectedShort parameterized by $k+c$ on unit disk graphs, that produces a solution with energy at most $2\cdot \text{OPT}+3k$, where OPT is the energy of an optimal solution (\cref{thm:approx}). Here, the parameter $c$ is the number of moves in an optimal solution.
In fact, the algorithm we present also works on a slightly more general graph class, so-called \emph{clique grids}, which are generalizations of grids and defined as follows. 

\begin{definition}[Clique Grid) (\cite{fomin2019finding}]\label{def:cliquegrid}
    A graph $G=(V,E)$ is a \emph{clique grid graph} if there exists a function $f:V\rightarrow [t]\times [t']$ for some $t,t'\in\mathbb{N}$ such that
    \begin{itemize}
        \item for all $(i,j)\in [t]\times [t']$ it holds that $f^{-1}(i,j)$ is a clique in $G$, and
        \item for all $\{u,v\}\in E$ it holds that if $f(u)=(i,j)$ and $f(v)=(i',j')$, then $|i-i'|\le 2$ and $|j-j'|\le 2$.
    \end{itemize}
    The function $f$ is a \emph{representation} of $G$.
\end{definition}
It is known that unit disk graphs are clique grids and that, given the point set of the unit disk graph, a representation can be computed in polynomial time.
\begin{lemma}[\cite{fomin2019finding}]\label{lem:cliquegrid}
    Let $G$ be a unit disk graph with point set $D$. Then $G$ is a clique grid graph and a representation $f$ of $G$ can be computed from $D$ in polynomial time.
\end{lemma}
In this section, we will assume that a representation of the input graph in form of a point set or a clique grid representation is given as part of the input. Now we prove the following result.

\approx*

We present an algorithm to prove \cref{thm:approx} with the following steps. Assume we are given an instance $(G=(V,E),R,O)$ of (the optimization version) of \MoveToConnectedShort and we have a clique grid representation of $G$. We first start with some basic ``guessing'' steps, which in this context means that we exhaustively enumerate and check all possibilities for the guess.
\begin{itemize}
    \item We first guess the optimal value for $c$.
    \item We exhaustively apply \cref{rr:1}, that is, we remove each vertex that has distance more than $c$ from every main robot.
    \item We guess a tree $T$ with $k$ vertices that witnesses the connectivity in an optimal solution of the instance. 
    \item We guess a bijection between the main robots and the vertices of $T$.
\end{itemize}

Next, we define a \emph{continuation type $t_i(v)$ of a vertex $v$ for step $i$} and we defines a \emph{type of a vertex $v$} as follows:
\begin{itemize}
    \item For step $c$, the \emph{continuation type} $t_c(v)$ of vertex $v$ is an element from $\{F,O,r_1,r_2,\ldots,r_k\}$ and set as follows. We set $t_c(v)=F$ meaning the vertex is \emph{free} if $v\notin O\cup R$, we set $t_c(v)=O$  meaning the vertex is \emph{occupied by an obnoxious robot} if $v\in O$, and we set $t_c(v)=r_i$ meaning the vertex is \emph{occupied by a main robot} if $v=r_i\in R$.

     \item For step $i<c$, the \emph{continuation type} $t_i(v)$ of a vertex $v$ is defined as the following set:

     For each clique $C$ that is neighboring the clique of $i$, and each $v'$ that is a neighbor of $c$ in $C$, add the pair $(C,t_{i+1}(v'))$ to $t_i(v)$.
     \item The \emph{type} of $v$ is a triple consisting of the clique $C$ of $v$, an element from $\{F,O,r_1,\ldots,r_k\}$ indicating whether $v$ is free or occupied (as above), and the list of all continuation types of $v$ for all steps $i\le c$.
     \item If there is a vertex type such that there are at most $2k+c+2$ vertices of that type, we additionally add the identity of the vertices to their type.
\end{itemize}
    Note that since we apply \cref{rr:1}, we have $O(k\cdot c^2)$ different cliques (this follows from an analogous argument as the proof of \cref{thm:kernel}).
     Since each clique only has constantly many neighboring cliques, the number of different continuation types is in $2^{O(k\cdot c^2)}$. We can iteratively compute the continuation types of all vertices for all steps in $2^{O(k\cdot c^2)}\cdot n^{O(1)}$ time, starting with step $c$, in a straightforward way.

The intuition behind the continuation types is as follows. If a robot arrives at a vertex in step $c$ it cannot travel any further afterwards. So the only relevant information is whether the vertex is free or occupied by another robot. In the latter case, the other robot needs to be moved before step $c$ to avoid a collision. If a robot arrives at a vertex at some step $i<c$, it can travel further, so it is relevant which types of vertices it can reach. This is encoded in the list of clique-vertex type pairs. Of course, it is also still relevant whether the vertex is free or occupied by another robot. Note that if a clique-neighbor type pair is in the list of the type of a vertex, it only means that the vertex can reach \emph{some} neighbor in the clique with the type (not all). 
This, intuitively, is the reason for the approximation factor two, since we can only guarantee that we move a robot to a vertex of the ``correct'' vertex type in the ``correct'' clique at each step, but we cannot guarantee that we move the robot to the ``correct'' vertex. So the solution we construct may need an extra step in each clique to move to the ``correct'' vertex. 

The algorithm now continues with a second guessing step, where we guess up to $\max\{c,k\}$ vectors $M_i$ of vertex types that each have length up to $c$ (possibly including duplicates). Whenever there is a type switch from one position to the next in the vector, we consider this to be a move. We reject all guesses where the total number of moves is greater than $c$. 
These vectors will be templates for the movements in the movement plan we construct. In a later step, we will replace the types with actual vertices. The main goal is to first move the main robots into the ``correct'' cliques, and then afterwards to the ``correct'' vertices inside the cliques.
For a given guess, we make the following checks:
\begin{itemize}
    \item For each vector $M_i$, we check whether it is a movement. More specifically, whenever two subsequent entries in the vector are two different types $t,t'$ we check whether the types are \emph{neighboring types}, that is, whether there are two vertices $u,v$ which are neighbors and $u$ has type $t$, and $v$ has type $t'$. Furthermore, we check whether the first entry of the vector is a vertex type for a vertex that is occupied by a robot.
    \item For each vector $M_i$, we check whether it is collision-free, that is, whenever there is a vertex type $t$ in the vector for a vertex that is occupied by a robot, there is another movement vector $M_j$ starting with type $t$ that moves the robot away (sufficiently early).

    If several vectors (say $\ell$) move a robot to a vertex of type $t$ for a vertex that is occupied by another robot, such that the time steps where the moved robots occupy the vertex of type $t$ intersect, we check whether there are $\ell$ other movements vector starting with type $t$ that move the robots away (sufficiently early).
    \item For each pair of vectors $M_i,M_j$ we check whether they are collision-free. This is only relevant if the vectors contain types that have vertex identities.
    \item We want to find vectors where the number of moves in each step is minimum, that is, moves only happen simultaneously if they have to, that is, when robots move around a cycle of length at least 3.

    We check whether at each step across all vectors, there is either at most one move, or at most one cyclic set of moves.
\end{itemize}

If none of the above checks fail, we check whether $T$ is a subgraph of $G$ such that the location of the vertices of $T$ in $G$ complies with our guess, that is, they have the correct type.
To do this, we use a subgraph finding algorithm that allows us to define a cost function for possible mappings from the vertices of $T$ to the vertices of $G$. We use this to make sure that the subgraph embedding complies with our guesses.
\begin{theorem}[\cite{DemaineHM14}]\label{thm:subgraph}
    Let $F$ be an undirected graph with $s$ vertices and treewidth $t$. Let $G$ be an undirected graph
with $n$ vertices and let $c : V(F) \times V(G) \rightarrow \mathbb{N}$ be a cost function of mapping a vertex of $F$ to a vertex of $G$.
If $F$ is a subgraph of $G$, then it is possible to find in time $2^{O(s)} \cdot n^{O(t)}$ a subgraph embedding $\phi$ that minimizes $c(\phi)$.
\end{theorem}
We define a cost function $c$ as follows: For $v\in V(T)$ such that we guessed that $v$ is mapped to main robot $r$, and the movement for $r$ moves the robot to a vertex of type $t$,
and $v'\in V(G)$, we set
\[
c(v,v')=\begin{cases}
1, & \text{ if } v'\text{ has type } t,\\
n, & \text{ otherwise.}
\end{cases}
\]
We use the algorithm behind \cref{thm:subgraph} to find subgraph embedding $\phi$ for $T$ into $G$ of minimal cost. If the cost is larger than $k$, we reject the guess.

Now we have all ingredients to construct a movement plan $\mathcal{M}$ for the \MoveToConnectedShort instance.
For each movement $M\in\mathcal{M}$ we have that the first vertex type for a vertex that is occupied by a robot, otherwise the checks described above fail. We pick an arbitrary vertex of that type (a different one for each movement) and identify the robot on that vertex as the one corresponding to the movement.
We iterate through the steps from $i=1$ to $c$:
\begin{itemize}
    \item We perform the $i$th move of each move of each movement. That is, for each $M\in \mathcal{M}$, we pick a (different) vertex $v$ of type $M[i+1]$, and move the robot corresponding to $M$ from its previous position to $v$.

    When picking the vertices, we avoid collision where two robots use the same edge in opposite directions. This can always be done since there are strictly more than $2k$ vertices of each type that does not include the vertex identity. Recall that if such a collision happens due to types that contain vertex identities, then the third of the above-described checks fails.
    \item If after the $i$th move is performed for each movement, a robot $r$ is moved to a vertex $v$ of some type $t$ that is occupied by another robot $r'$, then we moved robot $r$ to the ``wrong'' vertex type $t$.

    Since the second of the above-described checks did not fail, there must be another vertex~$v'$ of type $t$ that is not occupied by any robot after the $i$th move. Furthermore, since the last of the above-described checks did not fail, we know that $r'$ was already on $v$ in the previous step.

    We introduce an intermediate step (between steps $i-1$ and $i$) in which we move robot $r'$ from $v$ to $v'$. 
\end{itemize}
By construction, the above described movements are collision-free. Furthermore, the resulting movement plan as length at most $2c$ since for every move in the type vectors we add at most one additional move in the second of the above-described steps of the construction.

To ensure that the end-configuration is connected, we have to move the main robots to the vertices specified by the embedding $\phi$ of the tree $T$ on the main robots. By construction, we have that $\mathcal{M}$ moves each main robot to a vertex that has the same type as its target vertex in the embedding $\phi$. Note that all vertices of the same type are in the same clique. In a final step, we move each main robot onto their target vertex using at most three moves.

Assume that there is a main robot $r$ which is not on its target vertex $v$. If the target vertex $v$ is free, we can move $r$ to that vertex with one additional move. If the target vertex $v$ is occupied by some robot $r'$, then we pick another vertex $v'$ of the same type that is not occupied by a main robot. Note that since there are at least $2k+c+2$ vertices of each type (that does not include the vertex identity), such a vertex exits. If $v'$ is free, then we move the robot $r'$ to $v'$, and afterwards we move $r$ to $v$. If $v'$ is occupied by an obnoxious robot $r''$, we perform a cyclic move. Let $v''$ be the current position of $r$. We simultaneously move robot $r$ to $v$, robot $r'$ to $v'$, and robot $r''$ to $v''$. All described movements are clearly collision-free, contain at most three moves, and reduce the number of main robots that are not on their target vertex by one. Hence, by iterating this procedure, we can move all main robots to their respective target vertices in a collision-free manner by using at most three additional moves per main robot.

Now we are ready to prove \cref{thm:approx}.

\begin{proof}[Proof of \cref{thm:approx}]
    It is straightforward to check that the described guesses create $f(k+c)\cdot n^{O(1)}$ cases, the described checks can be performed in $f(k+c)\cdot n^{O(1)}$ time, and the movement plan can be constructed in $f(k+c)\cdot n^{O(1)}$ time.

    As argued during the description, the algorithm clearly produces a collision-free movement plan such that the end-configuration of the main robots is connected.

    Now assume there is an optimal solution $\mathcal{M}^\star$ of length $c$ to the \MoveToConnectedShort instance. Then there is a guess such that the each type vector corresponds to one movement in $\mathcal{M}^\star$ and the tree $T$ corresponds to a spanning tree in the subgraph of the main robots in their end-configuration. As argued, the movement plan constructed by the algorithm then has length at most $2c+3k$. It follows that we compute a 2-approximation with an additional additive error of $3k$.
\end{proof}

\section{Conclusion}\label{sec:conclusion}

To the best of our knowledge, we give the first computational results for the problems \MoveToPi and \MoveToConnectedShort. We leave several open questions. We believe that the following are of particular interest.
\begin{itemize}
    \item Is \MoveToConnectedShort fixed-parameter tractable when parameterized by $c$ if the input graph is planar?
    \item Is \MoveToConnectedShort fixed-parameter tractable when parameterized by $k+c$ if the input graph is a unit disk graph?
\end{itemize}
It would also be interesting to identify graph classes on which \MoveToConnectedShort becomes solvable in polynomial time.
\begin{itemize}
    \item Is \MoveToConnectedShort solvable in polynomial time if the input graph is a tree?
\end{itemize}
It is also natural to consider other parameters, such as e.g.\ the number of obnoxious robots or structural parameters of the input graph such as the treewidth.

Apart from these relatively specific questions, our work opens several natural future work directions. Instead of minimizing the total movement (energy) one could ask to minimize the makespan of the movement plan instead. There are also many other natural graph properties that one could focus on, for example one could require that the target formation admits a perfect matching among the main robots (for an even number of robots). Notably, for two robots this coincides with connectivity, and hence some of our hardness results transfer.

\bibliography{bibliography}	

\begin{thebibliography}{10}

\bibitem{AgarwalBHSS25}
Pankaj~K. Agarwal, Mark de~Berg, Benjamin Holmgren, Alex Steiger, and Martijn
  Struijs.
\newblock Optimal motion planning for two square robots in a rectilinear
  environment.
\newblock In {\em Proceedings of the 41st International Symposium on
  Computational Geometry (SoCG)}, volume 332 of {\em LIPIcs}, pages 5:1--5:17.
  Schloss Dagstuhl - Leibniz-Zentrum f{\"{u}}r Informatik, 2025.

\bibitem{arnborg1991easy}
Stefan Arnborg, Jens Lagergren, and Detlef Seese.
\newblock Easy problems for tree-decomposable graphs.
\newblock {\em Journal of Algorithms}, 12(2):308--340, 1991.

\bibitem{berman20111}
Piotr Berman, Erik~D Demaine, and Morteza Zadimoghaddam.
\newblock {$O(1)$}-approximations for maximum movement problems.
\newblock In {\em Proceedings of the 14th International Workshop on
  Approximation Algorithms for Combinatorial Optimization (APPROX)}, pages
  62--74. Springer, 2011.

\bibitem{BJK14}
Hans~L Bodlaender, Bart~MP Jansen, and Stefan Kratsch.
\newblock Kernelization lower bounds by cross-composition.
\newblock {\em {SIAM} Journal on Discrete Mathematics}, 28(1):277--305, 2014.

\bibitem{courcelle1990monadic}
Bruno Courcelle.
\newblock The monadic second-order logic of graphs. {I}. {R}ecognizable sets of
  finite graphs.
\newblock {\em Information and computation}, 85(1):12--75, 1990.

\bibitem{courcelle2012graph}
Bruno Courcelle and Joost Engelfriet.
\newblock {\em Graph structure and monadic second-order logic: a
  language-theoretic approach}, volume 138.
\newblock Cambridge University Press, 2012.

\bibitem{CalinescuDP08}
Gruia C{u{a}}linescu, Adrian Dumitrescu, and J{\'{a}}nos Pach.
\newblock Reconfigurations in graphs and grids.
\newblock {\em {SIAM} Journal of Discrete Mathematics}, 22(1):124--138, 2008.

\bibitem{Cyg+15}
Marek Cygan, Fedor~V. Fomin, {\L{}}ukasz Kowalik, Daniel Lokshtanov,
  D{\'{a}}niel Marx, Marcin Pilipczuk, Micha\l{} Pilipczuk, and Saket Saurabh.
\newblock {\em Parameterized Algorithms}.
\newblock Springer, 2015.

\bibitem{DeligkasEGKLR26}
Argyrios Deligkas, Eduard Eiben, Robert Ganian, Iyad Kanj, Dominik Leko, and
  M.~S. Ramanujan.
\newblock Routing few robots in a crowded network.
\newblock {\em Journal of Computer and System Sciences}, 157:103753, 2026.

\bibitem{DeligkasEGK024}
Argyrios Deligkas, Eduard Eiben, Robert Ganian, Iyad Kanj, and M.~S. Ramanujan.
\newblock Parameterized algorithms for coordinated motion planning: Minimizing
  energy.
\newblock In {\em Proceedings of the 51st International Colloquium on Automata,
  Languages, and Programming ({ICALP})}, volume 297 of {\em LIPIcs}, pages
  53:1--53:18. Schloss Dagstuhl - Leibniz-Zentrum f{\"{u}}r Informatik, 2024.

\bibitem{DemaineHMSGZ09}
Erik~D. Demaine, Mohammad~Taghi Hajiaghayi, Hamid Mahini, Amin~S.
  Sayedi{-}Roshkhar, Shayan~Oveis Gharan, and Morteza Zadimoghaddam.
\newblock Minimizing movement.
\newblock {\em {ACM} Transactions on Algorithms}, 5(3):30:1--30:30, 2009.

\bibitem{DemaineHM14}
Erik~D. Demaine, Mohammad~Taghi Hajiaghayi, and D{\'{a}}niel Marx.
\newblock Minimizing movement: Fixed-parameter tractability.
\newblock {\em {ACM} Transactions on Algorithms}, 11(2):14:1--14:29, 2014.

\bibitem{demaine2018simple}
Erik~D. Demaine and Mikhail Rudoy.
\newblock A simple proof that the $(n^2- 1)$-puzzle is hard.
\newblock {\em Theoretical Computer Science}, 732:80--84, 2018.

\bibitem{Die16}
Reinhard Diestel.
\newblock {\em Graph Theory, 5th Edition}, volume 173 of {\em Graduate Texts in
  Mathematics}.
\newblock Springer, 2016.

\bibitem{DF13}
Rodney~G. Downey and Michael~R. Fellows.
\newblock {\em Fundamentals of Parameterized Complexity}.
\newblock Springer, 2013.

\bibitem{eiben23}
Eduard Eiben, Robert Ganian, and Iyad Kanj.
\newblock {The Parameterized Complexity of Coordinated Motion Planning}.
\newblock In {\em Proceedings of the 39th International Symposium on
  Computational Geometry (SoCG)}, volume 258 of {\em Leibniz International
  Proceedings in Informatics (LIPIcs)}, pages 28:1--28:16. Schloss Dagstuhl --
  Leibniz-Zentrum f{\"u}r Informatik, 2023.

\bibitem{eiben2025minor}
Eduard Eiben, Robert Ganian, Iyad Kanj, and M.~S. Ramanujan.
\newblock A minor-testing approach for coordinated motion planning with sliding
  robots.
\newblock In {\em Proceedings of the 41st International Symposium on
  Computational Geometry (SoCG)}, volume 332 of {\em LIPIcs}, pages
  44:1--44:15. Schloss Dagstuhl - Leibniz-Zentrum f{\"{u}}r Informatik, 2025.

\bibitem{fekete2022computing}
S{\'a}ndor~P Fekete, Phillip Keldenich, Dominik Krupke, and Joseph~SB Mitchell.
\newblock Computing coordinated motion plans for robot swarms: The {CG}: {SHOP}
  challenge 2021.
\newblock {\em ACM Journal of Experimental Algorithmics (JEA)}, 27:1--12, 2022.

\bibitem{fellows2009multipleinterval}
Michael~R. Fellows, Danny Hermelin, Frances Rosamond, and St{\'e}phane
  Vialette.
\newblock On the parameterized complexity of multiple-interval graph problems.
\newblock {\em Theoretical Computer Science}, 410(1):53--61, 2009.

\bibitem{FernauHNRR03}
Henning Fernau, Torben Hagerup, Naomi Nishimura, Prabhakar Ragde, and Klaus
  Reinhardt.
\newblock On the parameterized complexity of the generalized rush hour puzzle.
\newblock In {\em Proceedings of the 15th Canadian Conference on Computational
  Geometry (CCCG)}, pages 6--9, 2003.

\bibitem{FioravantesKKMO24}
Foivos Fioravantes, Dusan Knop, Jan~Maty{\'{a}}s Kristan, Nikolaos Melissinos,
  and Michal Opler.
\newblock Exact algorithms and lowerbounds for multiagent path finding: Power
  of treelike topology.
\newblock In {\em Proceedings of the 38th {AAAI} Conference on Artificial
  Intelligence ({AAAI})}, pages 17380--17388. {AAAI} Press, 2024.

\bibitem{FlakeB02}
Gary~William Flake and Eric~B. Baum.
\newblock Rush hour is {PSPACE}-complete, or ``why you should generously tip
  parking lot attendants''.
\newblock {\em Theoretical Computer Science}, 270(1-2):895--911, 2002.

\bibitem{FG06}
J{\"o}rg Flum and Martin Grohe.
\newblock {\em Parameterized Complexity Theory}, volume XIV of {\em Texts in
  Theoretical Computer Science. An EATCS Series}.
\newblock Springer, 2006.

\bibitem{fomin2019finding}
Fedor~V Fomin, Daniel Lokshtanov, Fahad Panolan, Saket Saurabh, and Meirav
  Zehavi.
\newblock Finding, hitting and packing cycles in subexponential time on unit
  disk graphs.
\newblock {\em Discrete \& Computational Geometry}, 62(4):879--911, 2019.

\bibitem{Fom+19}
Fedor~V Fomin, Daniel Lokshtanov, Saket Saurabh, and Meirav Zehavi.
\newblock {\em Kernelization: Theory of Parameterized Preprocessing}.
\newblock Cambridge University Press, 2019.

\bibitem{friggstad2011minimizing}
Zachary Friggstad and Mohammad~R Salavatipour.
\newblock Minimizing movement in mobile facility location problems.
\newblock {\em ACM Transactions on Algorithms (TALG)}, 7(3):1--22, 2011.

\bibitem{garey1974some}
Michael~R Garey, David~S Johnson, and Larry Stockmeyer.
\newblock Some simplified {NP}-complete problems.
\newblock In {\em Proceedings of the 6th annual ACM Symposium on Theory of
  Computing (STOC)}, pages 47--63, 1974.

\bibitem{goldreich2011finding}
Oded Goldreich.
\newblock Finding the shortest move-sequence in the graph-generalized 15-puzzle
  is {NP}-hard.
\newblock In {\em Studies in complexity and cryptography. Miscellanea on the
  interplay between randomness and computation}, pages 1--5. Springer, 2011.

\bibitem{gupta2020parameterized}
Siddharth Gupta, Guy Sa'ar, and Meirav Zehavi.
\newblock The parameterized complexity of motion planning for snake-like
  robots.
\newblock {\em Journal of Artificial Intelligence Research}, 69:191--229, 2020.

\bibitem{hermelin2012weak}
Danny Hermelin and Xi~Wu.
\newblock Weak compositions and their applications to polynomial lower bounds
  for kernelization.
\newblock In {\em Proceedings of the 23rd annual ACM-SIAM Symposium on Discrete
  Algorithms (SODA)}, pages 104--113. SIAM, 2012.

\bibitem{KanjP24}
Iyad Kanj and Salman Parsa.
\newblock On the parameterized complexity of motion planning for rectangular
  robots.
\newblock In {\em Proceedings of the 40th International Symposium on
  Computational Geometry (SoCG)}, volume 293 of {\em LIPIcs}, pages
  65:1--65:15. Schloss Dagstuhl - Leibniz-Zentrum f{\"{u}}r Informatik, 2024.

\bibitem{Kar72}
Richard~M. Karp.
\newblock Reducibility among combinatorial problems.
\newblock In {\em Complexity of Computer Computations}, pages 85--103.
  Springer, 1972.

\bibitem{khot2002parameterized}
Subhash Khot and Venkatesh Raman.
\newblock Parameterized complexity of finding subgraphs with hereditary
  properties.
\newblock {\em Theoretical Computer Science}, 289(2):997--1008, 2002.

\bibitem{KlobasMMNZ23}
Nina Klobas, George~B. Mertzios, Hendrik Molter, Rolf Niedermeier, and Philipp
  Zschoche.
\newblock Interference-free walks in time: temporally disjoint paths.
\newblock {\em Autonomous Agents Multi Agent Systems}, 37(1):1, 2023.

\bibitem{KunzMZ23}
Pascal Kunz, Hendrik Molter, and Meirav Zehavi.
\newblock In which graph structures can we efficiently find temporally disjoint
  paths and walks?
\newblock In {\em Proceedings of the 32nd International Joint Conference on
  Artificial Intelligence ({IJCAI})}, pages 180--188. ijcai.org, 2023.

\bibitem{papadimitriou1994motion}
Christos~H Papadimitriou, Prabhakar Raghavan, Madhu Sudan, and Hisao Tamaki.
\newblock Motion planning on a graph.
\newblock In {\em Proceedings 35th Annual Symposium on Foundations of Computer
  Science (FOCS)}, pages 511--520. IEEE, 1994.

\bibitem{Pilz19}
Alexander Pilz.
\newblock Planar 3-{SAT} with a clause/variable cycle.
\newblock {\em Discrete Mathematics \& Theoretical Computer Science}, 21(3),
  2019.

\bibitem{ratner1990n2}
Daniel Ratner and Manfred Warmuth.
\newblock The $(n^2- 1)$-puzzle and related relocation problems.
\newblock {\em Journal of Symbolic Computation}, 10(2):111--137, 1990.

\bibitem{salzman2020research}
Oren Salzman and Roni Stern.
\newblock Research challenges and opportunities in multi-agent path finding and
  multi-agent pickup and delivery problems.
\newblock In {\em Proceedings of the 19th International Conference on
  Autonomous Agents and Multiagent Systems (AMAAS '20)}, pages 1711--1715,
  2020.

\bibitem{sharon2015conflict}
Guni Sharon, Roni Stern, Ariel Felner, and Nathan~R Sturtevant.
\newblock Conflict-based search for optimal multi-agent pathfinding.
\newblock {\em Artificial Intelligence}, 219:40--66, 2015.

\bibitem{Stern19}
Roni Stern.
\newblock Multi-agent path finding - an overview.
\newblock In {\em Artificial Intelligence - 5th {RAAI} Summer School}, pages
  96--115, Dolgoprudny, Russia, 2019. Springer.

\bibitem{SternSFK0WLA0KB19}
Roni Stern, Nathan~R Sturtevant, Ariel Felner, Sven Koenig, Hang Ma, Thayne~T.
  Walker, Jiaoyang Li, Dor Atzmon, Liron Cohen, T.~K.~Satish Kumar, Eli
  Boyarski, and Roman Bart{\'a}k.
\newblock Multi-agent pathfinding: Definitions, variants, and benchmarks.
\newblock In {\em Proceedings of the 12th International Symposium on
  Combinatorial Search ({SOCS})}, pages 151--159, 2019.

\bibitem{yu2013structure}
Jingjin Yu and Steven~M LaValle.
\newblock Structure and intractability of optimal multi-robot path planning on
  graphs.
\newblock In {\em Proceedings of the 27th AAAI Conference on Artificial
  Intelligence (AAAI '13)}, pages 1443--1449, 2013.

\end{thebibliography}

\end{document}